\documentclass[10pt,conference,letterpaper]{IEEEtran}

\usepackage{graphicx}
\usepackage{balance}
\usepackage[noend]{algorithmic}
\usepackage{subfigure}
\usepackage{multirow}
\usepackage{times}
\usepackage{color}
\usepackage{colortbl}
\usepackage{times}
\usepackage{amsmath}
\usepackage{epsfig}
\usepackage{amsfonts}

\newcommand{\G}{\mathcal{G}}
\newcommand{\F}{\mathcal{F}}
\newcommand{\R}{\mathcal{R}}
\newcommand{\Cov}{\textsf{Cov}}

\newcommand{\Sup}{\textsf{Num}}

\newcommand{\figwidth}{3.3in}

\newcommand{\st}{\scriptsize}

\newtheorem{lemma}{Lemma}
\newtheorem{property}{Property}
\newtheorem{theorem}{Theorem}
\newtheorem{claim}{Claim}

\allowdisplaybreaks

\definecolor{mygray}{gray}{.9}

\begin{document}

\title{Diversified Coherent Core Search \\ on Multi-Layer Graphs}

\author{
{Rong Zhu, Zhaonian Zou, and Jianzhong Li }
%\vspace{1.6mm}
\\
\fontsize{10}{10}\selectfont\itshape
% 20080211 CAUSAL PRODUCTIONS
% separate superscript on following line from affiliation using narrow space
%School of Computer Science and Technology, Harbin Institute of Technology\\
%92 West Dazhi Street, Harbin, Heilongjiang, China\\
Harbin Institute of Technology, Harbin, Heilongjiang, China\\
\fontsize{9}{9}\selectfont\ttfamily\upshape
\{rzhu, znzou, lijzh\}@hit.edu.cn
}

\maketitle

%%% Abstract %%%
\begin{abstract}
Mining dense subgraphs on multi-layer graphs is an interesting problem, which has witnessed lots of applications in practice. To overcome the limitations of the quasi-clique-based approach, we propose \emph{d}-coherent core (\emph{d}-CC), a new notion of dense subgraph on multi-layer graphs, which has several elegant properties. We formalize the diversified coherent core search (DCCS) problem, which finds \emph{k} \emph{d}-CCs that can cover the largest number of vertices. We propose a greedy algorithm with an approximation ratio of $1 - 1/e$ and two search algorithms with an approximation ratio of 1/4. The experiments verify that the search algorithms are faster than the greedy algorithm and produce comparably good results as the greedy algorithm in practice. As opposed to the quasi-clique-based approach, our DCCS algorithms can fast detect larger dense subgraphs that cover most of the quasi-clique-based results.
\end{abstract}

%%% Section 1 %%%
\section{Introduction}
\label{Sec:Introduction}

Dense subgraph mining, that is, finding vertices cohesively connected by internal edges, is an important issue in graph mining. In the literature, many dense subgraph notions have been formalized~\cite{LeeRJA10}, e.g., clique, quasi-clique, $k$-core, $k$-truss, $k$-plex and $k$-club. Meanwhile, a large number of dense subgraph mining algorithms have also been proposed.

In many real-world scenarios, a graph often contains various types of edges, which represent various types of relationships between entities. For example, in biological networks, interactions between genes can be detected by different methods~\cite{Hu2005Mining}; in social networks, users can interact through different social media~\cite{Qi2012Community}. In~\cite{Boden2012Mining} and~\cite{Pei2005On}, such a graph with multiple types of edges is modelled as a \emph{multi-layer graph}, where each layer independently accommodates a certain type of edges.

Finding dense subgraphs on multi-layer graphs has witnessed many real-world applications.

%% Applications

\noindent{\underline{\bf Application~1 (Biological Module Discovery).}}
In biological networks, densely connected vertices (genes or proteins), also known as biological modules, play an important role in detecting protein complexes and co-expression clusters~\cite{Hu2005Mining}. Due to data noise, there often exist a number of spurious biological interactions (edges), so a group of vertices only cohesively connected by interactions detected by a certain method may not be a convincing biological module. To filter out the effects of spurious interactions and make the detected modules more reliable, biologists detect interactions using multiple methods, i.e., build a multi-layer biological network, where each layer contains interactions detected by a certain method. A set of vertices is regarded as a reliable biological module if they are simultaneously densely connected on multiple layers~\cite{Hu2005Mining}.

\noindent{\underline{\bf Application~2 (Story Identification in Social Media.)}}
Social media, such as Twitter and Facebook, is updating with numerous new posts every day. A story in a social media is an event capturing popular attention recently~\cite{Angel2014Dense}. Stories can be identified by leveraging some real-world entities involved them, such as people, locations, companies and products. To identify them, scientists often abstract all new posts at each moment as a snapshot graph, where each vertex represents an entity and each edge links two entities if they frequently occur together in these new posts, and maintain a number of snapshot graphs in a time window. After that, each story can be identified by finding a group of strongly associated entities on multiple snapshot graphs~\cite{Angel2014Dense}. Obviously, this is an instance of finding dense subgraphs on multi-layer graphs.

Different from dense subgraph mining on single-layer graphs, dense subgraphs on multi-layer graphs must be evaluated by the following two orthogonal metrics: 1) \textbf{\em Density}: The interconnections between the vertices must be sufficiently dense on some individual layers. 2) \textbf{\em Support}: The vertices must be densely connected on a sufficiently large number of layers.

In the literature, the most representative and widely used notion of dense subgraphs on multi-layer graphs is \emph{cross-graph quasi-clique}~\cite{Boden2012Mining, Pei2005On, Zeng2006Coherent}. On a single-layer graph, a vertex set $Q$ is a $\gamma$-quasi-clique if each vertex in $Q$ is adjacent to at least $\gamma(|Q| - 1)$ vertices in $Q$, where $\gamma \in [0, 1]$. Given a set of graphs $G_1, G_2, \ldots, G_n$ with the same vertices (i.e., layers in our terminology) and $\gamma \in [0, 1]$, a vertex set $Q$ is a cross-graph quasi-clique if $Q$ is a $\gamma$-quasi-clique on all of $G_1, G_2, \ldots, G_n$. Although the cross-graph quasi-clique notion considers both density and support, it has several limitations:

1) A single cross-graph quasi-clique only characterizes a microscopic cluster. Finding all cross-graph quasi-cliques is computationally hard and is not scalable to large graphs~\cite{Boden2012Mining}.

2) The diameter of a cross-graph quasi-clique is often very small. As proved in~\cite{Pei2005On}, the diameter of a cross-graph quasi-clique is at most $2$ if $\gamma \geq 0.5$. Therefore, the quasi-clique-based methods face the following dilemma: When $\gamma$ is large, some large-scale dense subgraphs may be lost; When $\gamma$ is small, some sparsely connected subgraphs may be falsely recognized as dense subgraphs. For example, in the $4$-layer graph in Fig.~\ref{Fig:Motivation}, the vertex set $Q = \{a, b, c, d, e, f, g, h, i\}$ naturally induces a dense subgraph on all layers. However, in terms of cross-graph quasi-clique, if $\gamma \geq 0.5$, $Q$ is missing from the result; If $\gamma < 0.5$, the sparsely connected vertex set $\{g, h, i, j\}$ is recognized as a cross-graph quasi-clique.

Hence, there naturally arises the first question:

\noindent\textbf{Q1:} {\em What is a better notion of dense subgraphs on multi-layer graphs, which can avoid the limitations of cross-graph quasi-cliques?}

\begin{figure}[!t]
	\centering
	\includegraphics[width = \columnwidth]{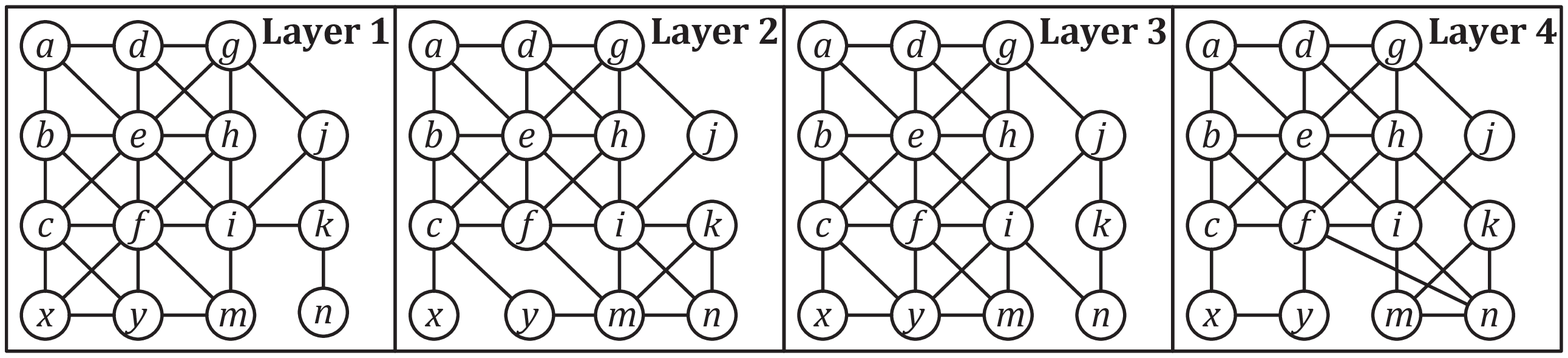}
	\vspace{-2em}
	\caption{Example of 4-Layer Graph.}
	\label{Fig:Motivation}
	\vspace{-2em}
\end{figure}

Additionally, as discovered in~\cite{Boden2012Mining}, dense subgraphs on multi-layer graphs have significant overlaps. For practical usage, it is better to output a small subset of \emph{diversified} dense subgraphs with little overlaps. Ref.~\cite{Boden2012Mining} proposed an algorithm to find diversified cross-graph quasi-cliques. One of our goal in this paper is to find dense subgraphs on even larger multi-layer graphs. There will be even more dense subgraphs, so the problem of finding diversified dense subgraphs will be more critical. Hence, we face the second question:

\noindent\textbf{Q2:} {\em How to design efficient algorithms to find diversified dense subgraphs according to the new notion?}

%\noindent{\textbf{\underline{Contributions.}}}
To deal with the first question \textbf{Q1}, we present a new notion called \emph{$d$-coherent core} (\emph{$d$-CC} for short) to characterize dense subgraphs on multi-layer graphs. It is extended from the $d$-core notion on single-layer graphs~\cite{Batagelj2003An}. Specifically, given a multi-layer graph $\G$, a subset $L$ of layers of $\G$ and $d \in \mathbb{N}$, the $d$-CC with respect to (w.r.t. for short) $L$ is the maximal vertex subset $S$ such that each vertex in $S$ is adjacent to at least $d$ vertices in $S$ on all layers in $L$. The $d$-CC w.r.t.~$L$ is unique. The $d$-CC notion is a natural fusion of density and support. It has the following advantages:

1) There is no limit on the diameter of a $d$-CC, and a $d$-CC often consists of a large number of densely connected vertices. Our experiments show that a $d$-CC can cover a large amount of cross-graph quasi-cliques.

2) A $d$-CC can be computed in linear time in the graph size.

3) The $d$-CC notion inherits the hierarchy property of $d$-core: The $(d + 1)$-CC w.r.t.~$L$ is a subset of the $d$-CC w.r.t.~$L$; The $d$-CC w.r.t.~$L$ is a subset of the $d$-CC w.r.t.~$L'$ if $L' \subseteq L$.

The $d$-CC notion overcomes the limitations of cross-graph quasi-cliques. Based on this notion, we formalize the \emph{diversified coherent core search (DCCS)} problem that finds dense subgraphs on multi-layer graphs with little overlaps: Given a multi-layer graph $\G$, a minimum degree threshold $d$, a minimum support threshold $s$, and the number $k$ of $d$-CCs to be detected, the \textsc{DCCS} problem finds $k$ most diversified $d$-CCs recurring on at least $s$ layers of $\G$. As in~\cite{Ausiello2011Online, Boden2012Mining}, we assess the diversity of the $k$ discovered $d$-CCs by the number of vertices they cover and try to maximize the diversity of these $d$-CCs. We prove that the \textsc{DCCS} problem is NP-complete.

To deal with the second question \textbf{Q2}, we propose a series of approximation algorithms for the \textsc{DCCS} problem. First, we propose a simple greedy algorithm, which finds $k$ $d$-CCs in a greedy manner. The algorithm have an approximation ratio of $1 - 1/e$. However, it must compute all candidate $d$-CCs and therefore is not scalable to large multi-layer graphs.

To prune unpromising candidate $d$-CCs early and improve efficiency, we propose two search algorithms, namely the bottom-up search algorithm and the top-down search algorithm. In both algorithms, the process of generating candidate $d$-CCs and the process of updating diversified $d$-CCs interact with each other. Many $d$-CCs that are unpromising to appear in the final results are pruned in early stage. The bottom-up and top-down algorithms adopt different search strategies. In practice, the bottom-up algorithm is preferable if $s < l/2$, and the top-down algorithm is preferable if $s \ge l/2$, where $l$ is the number of layers. Both of the algorithms have an approximation ratio of $1/4$.

We conducted extensive experiments on a variety of datasets to evaluate the proposed algorithms and obtain the following results: 1) The bottom-up and top-down algorithms are $1$--$2$ orders of magnitude faster than the greedy algorithm for small and large $s$, respectively. 2) The practical approximation quality of the bottom-up and top-down algorithms is comparable to that of the greedy algorithm. 3) Our \textsc{DCCS} algorithms outperform the quasi-clique-based dense subgraph mining algorithm~\cite{Boden2012Mining} on multi-layer graphs in terms of both execution time and result quality.

%The rest of the paper is organized as follows. Section~II introduces the basic concepts and formalizes the \textsc{DCCS} problem. Section~III presents the greedy algorithms. The bottom-up and the top-down search algorithms are described in Sections~IV and~V, respectively. The experimental results are reported in Section~VI. Section~VII reviews the related work, and Section~VIII concludes this paper.

%%% Section 2 %%%
\section{Problem Definition}
\label{Sec:ProblemDefinition}

\noindent{\underline{\bf Multi-Layer Graphs.}}
A \emph{multi-layer graph} is a set of graphs $\{G_1, G_2, \dots, G_l\}$, where $l$ is the number of layers, and $G_i$ is the graph on layer $i$. Without loss of generality, we assume that $G_1, G_2, \dots, G_l$ contain the same set of vertices because if a vertex is missing from layer $i$, we can add it to $G_i$ as an isolated vertex. Hence, a multi-layer graph $\{G_1, G_2, \dots, G_l\}$ can be equivalently represented by $(V, E_1, E_2, \dots, E_l)$, where $V$ is the universal vertex set, and $E_i$ is the edge set of $G_i$.

Let $V(G)$ and $E(G)$ be the vertex and the edge set of graph $G$, respectively. For a vertex $v \in V(G)$, let $N_G(v) = \{u | (v, u) \in E(G)\}$ be the set of neighbors of $v$ in $G$, and let $d_G(v) = |N_G(v)|$ be the degree of $v$ in $G$. The subgraph of $G$ induced by a vertex subset $S \subseteq V(G)$ is $G[S] = (S, E[S])$, where $E[S]$ is the set of edges with both endpoints in $S$.

Given a multi-layer graph $\G = (V, E_1, E_2, \dots, E_l)$, let $l(\G)$ be the number of layers of $\G$, $V(\G)$ the vertex set of $\G$, and $E_i(\G)$ the edge set of the graph on layer $i$. The multi-layer subgraph of $\G$ induced by a vertex subset $S \subseteq V(\G)$ is $\G[S] = (S, E_1[S], E_2[S], \dots, E_l[S])$, where $E_i[S]$ is the set of edges in $E_i$ with both endpoints in $S$.

%\begin{figure}
%	\centering
%	\includegraphics[width=0.78\columnwidth]{./ExpFig/ExpGraph.pdf}
%	\vspace{-1em}
%	\caption{Example of 6-Layer Graph.}
%	\label{Fig:ExpGraph}
%	\vspace{-2em}
%\end{figure}

\noindent{\underline{\bf \emph{d}-Coherent Cores.}}
We define the notion of \emph{$d$-coherent core ($d$-CC)} on a multi-layer graph by extending the \emph{$d$-core} notion on a single-layer graph~\cite{Batagelj2003An}. A graph $G$ is \emph{$d$-dense} if $d_G(v) \geq d$ for all $v \in V(G)$, where $d\in \mathbb{N}$. The \emph{$d$-core} of graph $G$, denoted by $C^{d} (G)$, is the maximal subset $S \subseteq V(G)$ such that $G[S]$ is $d$-dense. As stated in~\cite{Batagelj2003An}, $C^{d} (G)$ is unique, and $C^{d}(G) \subseteq C^{d - 1}(G) \subseteq \dots \subseteq C^{1}(G) \subseteq C^{0}(G)$ for $d \in \mathbb{N}$.

For ease of notation, let $[n] = \{1, 2, \ldots, n\}$, where $n \in \mathbb{N}$. Let $\G$ be a multi-layer graph and $L \subseteq [l(\G)]$ be a non-empty subset of layer numbers. For $S \subseteq V(\G)$, the induced subgraph $\G[S]$ is $d$-dense w.r.t.~$L$ if $G_i[S]$ is $d$-dense for all $i \in L$. The \emph{$d$-coherent core} (\emph{$d$-CC}) of $\G$ w.r.t.~$L$, denoted by $C^{d}_{L}(\G)$, is the maximal subset $S \subseteq V(\G)$ such that $\G[S]$ is $d$-dense w.r.t.~$L$. Similar to $d$-core, the concept of $d$-CC has the following properties.

\begin{property}[Uniqueness]
\label{Lem:kUnique}
Given a multi-layer graph $\G$ and a subset $L \subseteq [l(\G)]$, $C^{d}_{L}(\G)$ is unique for $d \in \mathbb{N}$.
\end{property}

\begin{property}[Hierarchy]
\label{Lem:kHierarchy}
Given a multi-layer graph $\G$ and a subset $L \subseteq [l(\G)]$, we have $C^{d}_{L}(\G) \subseteq C^{d - 1}_{L}(\G) \subseteq \dots \subseteq C^{1}_{L}(\G) \subseteq C^{0}_L(\G)$ for $d \in \mathbb{N}$.
\end{property}

\begin{property}[Containment]
\label{Lem:PHierarchy}
Given a multi-layer graph $\G$ and two subsets $L, L' \subseteq [l(\G)]$, if $L \subseteq L'$, we have $C^{d}_{L'}(\G) \subseteq C^{d}_{L}(\G)$ for $d \in \mathbb{N}$.
\end{property}

\noindent \textbf{Note:} {\em We put all proofs in Appendix~A.}

\noindent{\underline{\bf Problem Statement.}}
Given a multi-layer graph $\G$, a minimum degree threshold $d \in \mathbb{N}$ and a minimum support threshold $s \in \mathbb{N}$, let $\F_{d, s}(\G)$ be the set of $d$-CCs of $\G$ w.r.t.~all subsets $L \subseteq [l(\G)]$ such that $|L| = s$. When $\G$ is large, $|\F_{d, s}(\G)|$ is often very large, and a large number of $d$-CCs in $\F_{d, s}(\G)$ significantly overlap with each other. For practical usage, it is better to output $k$ diversified $d$-CCs with little overlaps, where $k$ is a number specified by users. Like~\cite{Ausiello2011Online, Boden2012Mining}, we assess the diversity of the discovered $d$-CCs by the number of vertices they cover and try to maximize the diversity of these $d$-CCs. Let the \emph{cover set} of a collection of sets $\R = \{R_1, R_2, \dots, R_n\}$ be $\Cov(\R) = \bigcup_{i = 1}^{n} R_i$. We formally define the \emph{Diversified Coherent Core Search (DCCS) problem} as follows.

Given a multi-layer graph $\G$, a minimum degree threshold $d$, a minimum support threshold $s$ and the number $k$ of $d$-CCs to be discovered, find the subset $\R \subseteq \F_{d, s}(\G)$ such that 1) $|\R| = k$; and 2) $|\Cov (\R)|$ is maximized. The $d$-CCs in $\R$ are called the \emph{top-$k$ diversified $d$-CCs} of $\G$ on $s$ layers.

\begin{theorem}
The \textsc{DCCS} problem is NP-complete.
\end{theorem}

Let $d= 3$, $s = 2$ and $k = 2$. The top-$2$ diversified $d$-CCs for the multi-layer graph in Fig.~\ref{Fig:Motivation} is $\R = \{ C^{d}_{ \{ 1, 3 \} }(\G), C^{d}_{ \{ 2, 4\} }(\G) \}$, where $C^{d}_{ \{ 1, 3 \} }(\G) = \{ a, b, c, d, e, f, g, h, i, y, m\}$, $C^{d}_{ \{ 2, 4\} }(\G) = \{ a, b, c, d, e, f, g, h, i, m, n, k\}$ and $|\Cov(\R)| = 14$.

%%% Section 3 %%%
\section{Greedy Algorithm}
\label{Sec:GAlgorithm}

A straightforward solution to the \textsc{DCCS} problem is to generate all candidate $d$-CCs and select $k$ of them that cover the maximum number of vertices. However, the search space of all $k$-combinations of $d$-CCs is extremely large, so this method is intractable even for small multi-layer graphs. Alternatively, fast approximation algorithms with guaranteed performance may be more preferable. In this section, we propose a simple greedy algorithm with an approximation ratio of $1 - 1/e$.

Before describing the algorithm, we present the following lemma based on Property~\ref{Lem:PHierarchy}. The lemma enables us to remove irrelevant vertices early.

\begin{lemma}[Intersection Bound]
\label{Lem:kInjection}
Given a multi-layer graph $\G$ and two subsets $L_1, L_2 \subseteq [l(\G)]$, we have $C_{L_1 \cup L_2}^{d} (\G) \subseteq C_{L_1}^{d} (\G) \cap C_{L_2}^{d}(\G)$ for $d \in \mathbb{N}$.
\end{lemma}

\begin{figure}[!t]
    \scriptsize
    \fbox{
    \parbox{\figwidth}{
    \textbf{Algorithm} \textsf{GreedyDCCS}$(\G, d, s, k)$
    \begin{algorithmic}[1]
        \STATE $\F \gets \emptyset$, $\R \gets \emptyset$
        \FOR{$i \gets 1$ to $l(\G)$}
            \STATE compute $C^{d}(G_i)$ on $G_i$
        \ENDFOR
        \FOR{each $L \subseteq [l(\G)]$ such that $|L| = s$}
             \STATE $S \gets \bigcap_{i \in L} C^{d}(G_i)$
            \STATE $C^{d}_{L}(\G) \gets \textsf{dCC}(\G[S], L, d)$
            \STATE $\F \gets \F \cup \{C^{d}_{L}(\G)\}$
        \ENDFOR
        \FOR{$j \gets 1  $ to $k$}
            \STATE $C^* \gets \arg\max_{C \in \F} (|\Cov(\R \cup \{C\})| - |\Cov(\R)|)$
            \STATE $\R \gets \R \cup \{C^*\}$, $\F \gets \F - \{C^*\}$
        \ENDFOR
        \RETURN $\R$
    \end{algorithmic}
%    \textbf{Procedure} \textsf{dCC}$(\G, L, d)$
%    \begin{algorithmic}[1]
%        \WHILE {there exists $v \in V(\G)$ and $i \in L$ such that $d_{G_i}(v) < d$ }
%            \STATE remove $v$ from all layers of $\G$
%        \ENDWHILE
%        \RETURN $V(\G)$
%    \end{algorithmic}
    }}
   \vspace{-0.5em}
   \caption{The \textsf{\small GD-DCCS} Algorithm.}
   \label{Fig:GreedyDCCS}
   \vspace{-3em}
\end{figure}

\noindent{\underline{\bf The Greedy Algorithm.}} The greedy algorithm \textsf{GD-DCCS} is described in Fig.~\ref{Fig:GreedyDCCS}. The input is a multi-layer graph $\G$ and $d, s, k \in \mathbb{N}$. \textsf{GD-DCCS} works as follows. Line~1 initializes both the $d$-CC collection $\F$ and the result set $\R$ to be $\emptyset$. Lines~2--3 compute the $d$-core $C^d(G_i)$ on each layer $G_i$ by the algorithm in~\cite{Batagelj2003An}. By definition, we have $C^{d}_{\{i\}}(\G) = C^{d}(G_i)$.

For each $L \subseteq [l(\G)]$ with $|L| = s$, to find $C_{L}^d(\G)$, we first compute the intersection $S = \bigcap_{i \in L} C^{d}(G_i)$ (line~5). By Lemma~\ref{Lem:kInjection}, we have $C_{L}^{d}(\G) \subseteq S$. Thus, we compute $C_{L}^{d}(\G)$ on the induced subgraph $\G[S]$ instead of on $\G$ by Procedure \textsf{dCC} (line~6) and add $C_{L}^{d}(\G)$ to $\F$ (line~7). Procedure \textsf{dCC} follows the similar procedure of computing the $d$-core on a single-layer graph~\cite{Batagelj2003An}. Whenever there exists a vertex $v \in V(\G)$ such that $d_{G_i}(v) < d$
on some layer $i \in L$, $v$ is removed from all layers of $\G$. Due to space limits, we describe the implementation details of \textsf{dCC} in Appendix~B.

Next, lines~8--10 select $k$ $d$-CCs from $\F$ in a greedy manner. In each time, we pick up the $d$-CC $C^* \in \F$ that maximizes $|\Cov(\R \cup \{C^*\})| - |\Cov(\R)|$, add $C^*$ to $\R$ (line~9) and remove $C^*$ from $\F$ (line~10). Finally, $\R$ is output as the result (line~11).

Let $l = l(\G)$, $n = |V(\G)|$  and $m = |\bigcup_{i = 1}^{l} E_i(\G) |$. Procedure \textsf{dCC} in line~6 runs in $O(ns + ms)$ time as shown in Appendix~B. Line~9 runs in $O(n|\F|)$ time since computing $|\Cov(\R \cup \{C\})| - |\Cov(\R)|$ takes $O(n)$ time for each $C \in \F$. In addition, $|\F| = {l \choose s}$. Therefore, the time complexity of \textsf{GD-DCCS} is $O((ns + ms + kn){l \choose s})$, and the space complexity is $O(n{l \choose s})$.

\begin{theorem}\label{Thm:GreedyDCCS}
The approximation ratio of \textsf{GD-DCCS} is $1 - \frac{1}{e}$.
\end{theorem}

\noindent{\underline{\bf Limitations.}}
As verified by the experimental results in Section~\ref{Sec:PEvaluation}, \textsf{GD-DCCS} is not scalable to very large multi-layer graphs. This is due to the following reasons:
1) \textsf{GD-DCCS} must keep all candidate $d$-CCs in $\F$. As $l(\G)$ increases, $|\F|$ grows exponentially. When $\F$ can not fit in main memory, the algorithm incurs large amounts of I/Os.
2) The exponential growth in $|\F|$ significantly increases the running time of selecting $k$ diversified $d$-CCs from $\F$ (lines~8--10 of \textsf{GD-DCCS}).
3) The phase of candidate $d$-CC generation (lines~1--7) and the phase of diversified $d$-CC selection (lines~8--10) are separate. There is no guidance on candidate generation, so many unpromising candidate $d$-CCs are generated in vain.

%%%% Section 4 %%%
\section{Bottom-Up Algorithm}
\label{Sec:BApproach}

This section proposes a bottom-up approach to the \textsc{DCCS} problem. In this approach, the candidate $d$-CC generation and the top-$k$ diversified $d$-CC selection phases are \emph{interleaved}. On one hand, we maintain a set of temporary top-$k$ diversified $d$-CCs and use each newly generated $d$-CC to update them. On the other hand, we guide candidate $d$-CC generation by the temporary top-$k$ diversified $d$-CCs.

In addition, candidate $d$-CCs are generated in a bottom-up manner. Like the frequent pattern mining algorithm~\cite{Yan2002gSpan}, we organize all $d$-CCs by a search tree and search candidate $d$-CCs on the search tree. The bottom-up $d$-CC generation has the following advantage: If the $d$-CC w.r.t.~subset $L$ $(|L| < s)$ is unlikely to improve the quality of the temporary top-$k$ diversified $d$-CCs, the $d$-CCs w.r.t.~all $L'$ such that $L \subseteq L'$ and $|L'| = s$ need not be generated. As verified by the experimental results in Section~\ref{Sec:PEvaluation}, the bottom-up approach reduces the search space by 80\%--90\% in comparison with the greedy algorithm and thus saves large amount of time. Moreover, the bottom-up \textsc{DCCS} algorithm attains an approximation ratio of $1/4$.

%The section is organized as follows. Sections~\ref{Sec:BApproach-2} and \ref{Sec:BApproach-3} introduce two key procedures of the bottom-up approach; Section~\ref{Sec:BApproach-4} proposes some preprocessing methods; Section~\ref{Sec:BApproach-5} describes the complete algorithm.

%% Section 4.2 %%
\subsection{Maintenance of Top-k Diversified d-CCs}
\label{Sec:BApproach-2}

Let $\R$ be a set of temporary top-$k$ diversified $d$-CCs. In the beginning, $\R = \emptyset$. To improve the quality of $\R$, we try to update $\R$ whenever we find a new candidate $d$-CC $C$. In particular, we update $\R$ with $C$ by one of the following rules:

\noindent\underline{\bf Rule~1:} If $|\R| < k$, $C$ is added to $\R$.

\noindent\underline{\bf Rule~2:} For $C' \in \R$, let $\Delta(\R, C') = C' - \Cov(\R - \{C'\})$, that is, $\Delta(\R, C')$ is vertex set in $\Cov(\R)$ exclusively covered by $C'$. Let $C^*(\R) = \arg\min_{C' \in \R} |\Delta(\R, C')|$, that is, $C^*(\R)$ exclusively covers the least number of vertices among all $d$-CCs in $\R$.  We replace $C^*(\R)$ with $C$ if $|\R| = k$ and
\begin{equation}
	\label{Eqn: RUpdate}
	|\Cov((\R - \{C^*(\R)\}) \cup \{C\})| \ge (1 + \tfrac{1}{k}) |\Cov(\R)|.
\end{equation}

On input $\R$ and $C$, Procedure \textsf{Update} tries to update $\R$ with $C$ using the rules above. The details of \textsf{Update} is described in Appendix~C.  By using two index structures, \textsf{Update} runs in $O(\max\{|C|, |C^{*}(\R)|\})$ time.

%% Section 4.2 %%
\subsection{Bottom-Up Candidate Generation}
\label{Sec:BApproach-3}

Candidate $d$-CCs $C^d_L(\G)$ with $|L| = s$ are generated in a bottom-up fashion. As shown in Fig.~\ref{Fig: BottomUp Tree}, all $d$-CCs $C^d_L(\G)$ are conceptually organized by a search tree, in which $C_L^d(\G)$ is the parent of $C_{L'}^d(\G)$ if $L \subset L'$, $|L'| = |L| + 1$ and the only number $\ell \in L' - L$ satisfies $\ell > \max(L)$, where $\max(L)$ is the largest number in $L$ (specially, $\max(\emptyset) = -\infty$). Conceptually, the root of the search tree is $C^{d}_{\emptyset}(\G) = V(\G)$.

\begin{figure}[!t]
    \scriptsize
    \fbox{
    \parbox{\figwidth}{
    \textbf{Procedure} \textsf{BU-Gen}$(\G, d, s, k, L, C^{d}_{L}(\G),  L_Q, \R)$ \!\!\!\!
    \begin{algorithmic}[1]
    \STATE $L_{P} \gets \{ j | \max(L) <  j \leq l(\G)\} - L_Q$, $L_{R} \gets \emptyset$
    \IF{$|\R| < k$}
        \FOR{$j \in L_{P}$}
            \STATE $L' \gets L \cup \{ j \} $
            \STATE $C^{d}_{L'} (\G) \gets$ \textsf{dCC}$(\G[C^{d}_{L}(\G) \cap C^{d}(G_j)], L', d)$
            \IF{$|L'| = s$}
                \STATE $\mathsf{Update}(\R, C^{d}_{L'}(\G))$
            \ELSE
                \STATE $L_{R} \gets L_{R} \cup \{ j \}$
            \ENDIF
        \ENDFOR
    \ELSIF{$|\R| = k$}
        \STATE sort $j \in L_{P}$ in descending order of $|C^{d}_{L}(\G) \cap C^{d}(G_j)|$
        \FOR{each $j$ in the sorted $L_P$}
            \IF{$|C_{L}^d(\G) \cap C^d(G_j)| < \frac{1}{k}|\Cov(\R)| + |\Delta(\R, C^*(\R))|$}
                \STATE \textbf{break}
            \ELSE
                \STATE $L' \gets L \cup \{ j \} $
                \STATE $C^{d}_{L'} (\G) \gets$ \textsf{dCC}$(\G[C^{d}_{L}(\G) \cap C^{d}(G_j)], L', d)$
                \IF{$|L'| = s$}
                    \STATE $\mathsf{Update}(\R, C^{d}_{L'}(\G))$
                \ELSE
                    \IF{$C^{d}_{L'} (\G)$ satisfies Eq.~\eqref{Eqn: RUpdate}}
                        \STATE $L_{R} \gets L_{R} \cup \{ j \}$
                    \ENDIF
                \ENDIF
            \ENDIF
        \ENDFOR
    \ENDIF
    \IF{$|L| < s$}
    \FOR{$j \in L_{R}$}
        \STATE $L' \gets L \cup \{ j \} $
        \STATE \textsf{BU-Gen}$(\G, d, s, k, L', C^{d}_{L'}(\G), L_Q \cup (L_{P} -  L_{R}), \R)$
    \ENDFOR
    \ENDIF
    \end{algorithmic}
    }}
     \vspace{-0.5em}
    \caption{The \textsf{BU-Gen} Procedure.}
    \label{Fig:BottomUpGen}
    \vspace{-3em}
\end{figure}

\begin{figure*}[!t]
\hspace{0.02\textwidth}
	\begin{minipage}[t]{0.27\textwidth}
		\includegraphics[width=0.99\columnwidth, height = 1in]{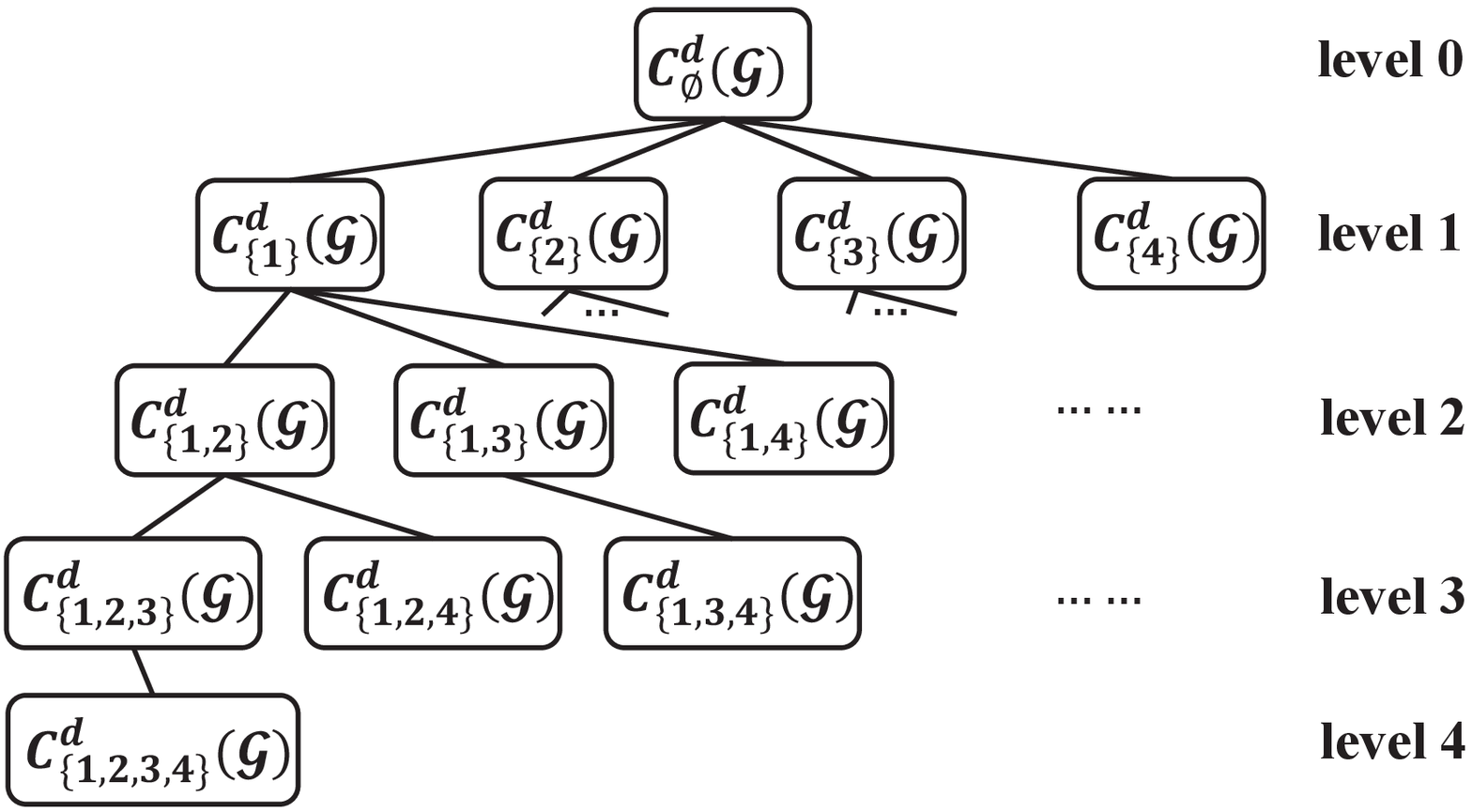}
		\vspace{-1.9em}
		\caption{Bottom-Up Search Tree.}
		\label{Fig: BottomUp Tree}
	\end{minipage}
    \hspace{0.05\textwidth}
	\begin{minipage}[t]{0.27\textwidth}
		\includegraphics[width=0.99\columnwidth, height = 1in]{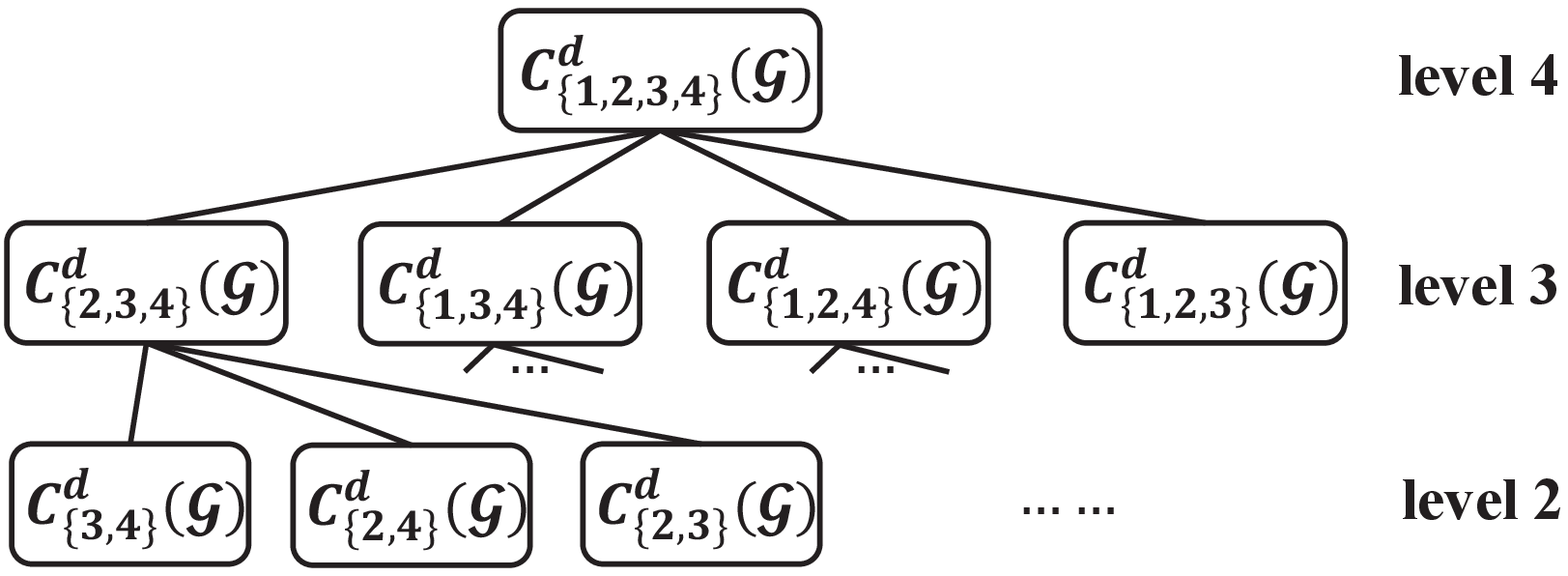}
		\vspace{-1.9em}
		\caption{Top-Down Search Tree.}
		\label{Fig: TopDown Tree}
	\end{minipage}
   \hspace{0.05\textwidth}
	\begin{minipage}[t]{0.3\textwidth}
       \centering
		\includegraphics[width=0.75\columnwidth, height = 1.1in]{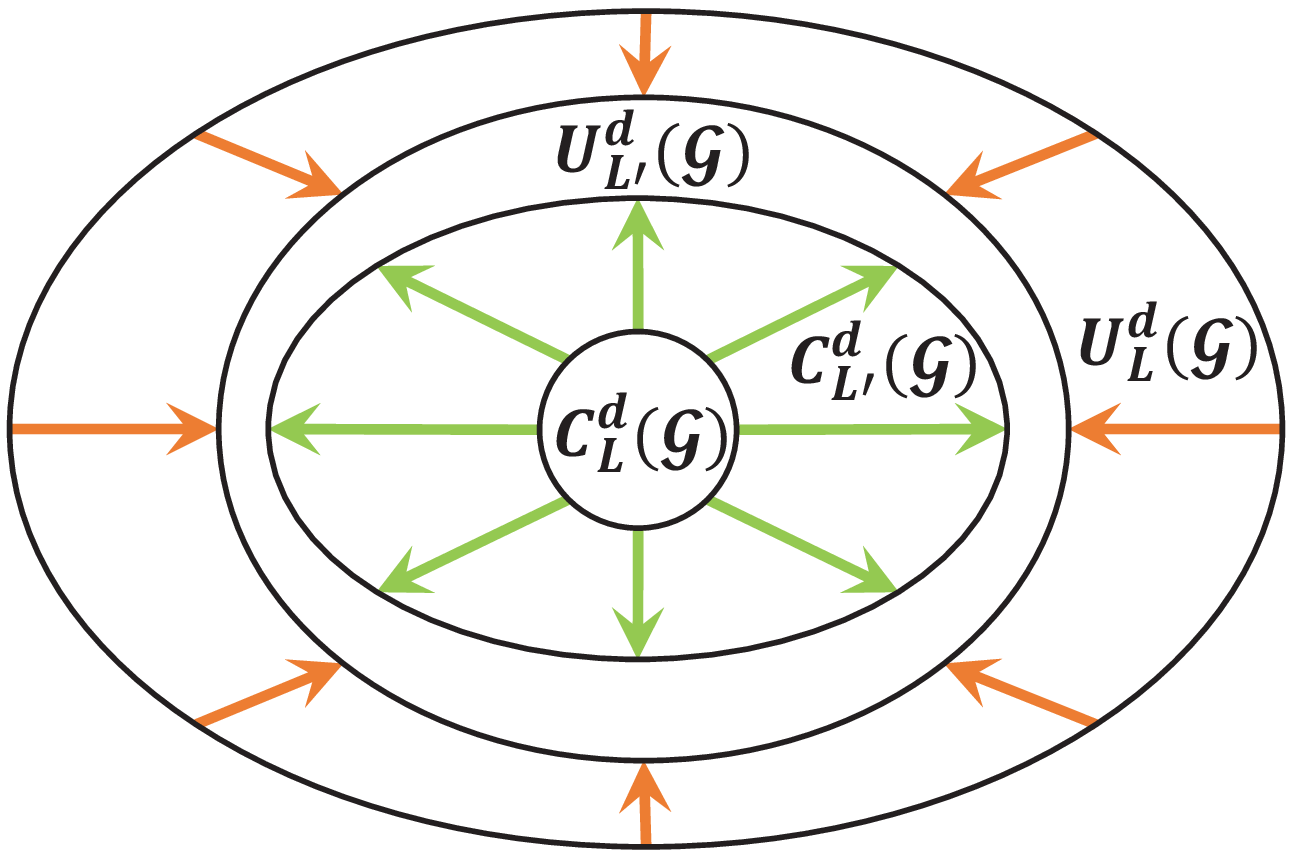}
		\vspace{-1.5em}
		\caption{Relationships between $C^{d}_{L}(\G)$, $U^{d}_{L}(\G)$, $C^{d}_{L'}(\G)$ and $U^{d}_{L'}(\G)$.}
		\label{Fig: CURelationship}
	\end{minipage}

\end{figure*}

The $d$-CCs in the search tree are generated in a depth-first order. First, we generate the $d$-core $C^d(G_i)$ on each single layer $G_i$. By definition, we have $C^{d}_{\{i\}}(\G) = C^d(G_i)$.
Then, starting from $C^d_{\{i\}}(\G)$, we generate the descendants of $C^d_{\{i\}}(\G)$. The depth-first search is realized by recursive Procedure \textsf{BU-Gen} in Fig.~\ref{Fig:BottomUpGen}. In general, given a $d$-CC $C^d_L(\G)$ as input, we first expand $L$ by adding a layer number $j$ such that $\max(L) < j \le l(\G)$. Let $L' = L \cup \{j\}$. By Lemma~\ref{Lem:kInjection}, we have $C_{L'}^{d}(\G) \subseteq C_L^d(\G) \cap C_{\{j\}}^d(\G) = C_L^d(\G) \cap C^d(G_j)$. Thus, we compute $C_{L'}^d(\G)$ on the induced subgraph $\G[C_L^d(\G) \cap C^d(G_j)]$ by Procedure \textsf{dCC} described in Section~\ref{Sec:GAlgorithm}. Next, we process $C_{L'}^d(\G)$ according to the following cases:

\noindent \underline{\bf Case~1:} If $|L'| = s$, we update $\R$ with $C_{L'}^d(\G)$.

\noindent \underline{\bf Case~2:} If $|L'| < s$ and $|\R| < k$, we recursively call \textsf{BU-Gen} to generate the descendants of $C_{L'}^d(\G)$.

\noindent \underline{\bf Case~3:} If $|L'| < s$ and $|\R| = k$, we check if $C_{L'}^d(\G)$ satisfies Eq.~\eqref{Eqn: RUpdate} to update $\R$. If not satisfied, none of the descendants of $C_{L'}^d(\G)$ is qualified to be a candidate, so we prune the entire subtree rooted at $C_{L'}^d(\G)$; otherwise, we recursively call \textsf{BU-Gen} to generate the descendants of $C_{L'}^d(\G)$. The correctness is guaranteed by the following lemma.

\begin{lemma}[Search Tree Pruning]
\label{Lem: dCCPrune}
For a $d$-CC $C^d_{L}(\G)$, if $C^d_{L}(\G)$ does not satisfy Eq.~\eqref{Eqn: RUpdate}, none of the descendants of $C^d_{L}(\G)$ can satisfy Eq.~\eqref{Eqn: RUpdate}.
\end{lemma}

%For example, consider the multi-layer graph $\G$ in Figure~\ref{Fig:ExpGraph}. Let $d = 3$, $s = 3$ and $k = 2$. Suppose $\R = \{\{a, b, c, d, e, f, g, h\}, \{ a, \break b, c, d, w, x, y, z\} \}$. We have $C^d_{\{3, 4\}}(\G) = \{ a, b, c, d, e, g, x, z\}$. Thus, $|\Cov((\R - \{ C^*(\R)\}) \cup C^d_{\{3, 4\}}(\G))| = 10$ and $|\Cov(\R)| = 12$, so $C^d_{\{3, 4\}}(\G)$ does not satisfy Eq.~\eqref{Eqn: RUpdate}. Moreover, $C^d_{\{3, 4, 5\}}(\G) = \{ a, b, c, d, x, z\}$ and $C^d_{\{3, 4, 6\}}(\G) = \{a, b, c, d\}$. All of them are the descendants of $C^d_{\{3, 4\}}(\G)$. It can be easily verified that none of them satisfies Eq.~\eqref{Eqn: RUpdate}.

To further improve efficiency, if $|\R| = k$, we order the layer numbers $j > \max(L)$ in decreasing order of $|C_L^d(\G) \cap C^d(G_j)|$ and generate $C_{L \cup \{j\}}^d(\G)$ according to this order of $j$. For some $j$, if $|C_{L}^d(\G) \cap C^d(G_j)|<\frac{1}{k}|\Cov(\R)| + |\Delta(\R, C^*(\R))|$, we can stop searching the subtrees rooted at $C_{L \cup \{j\}}^d(\G)$ and $C_{L \cup \{j'\}}^d(\G)$ for all $j'$ succeeding $j$ in the order. The correctness is ensured by the following lemma.

\begin{lemma}[Order-based Pruning]
\label{Lem: LayerOrder}
For a $d$-CC $C^d_L(\G)$ and $j > \max(L)$, if $|C_{L}^d(\G) \cap C^d(G_j)| < \frac{1}{k}|\Cov(\R)| + |\Delta(\R, C^*(\R))|$, then $C_{L \cup \{j\}}^d(\G)$ cannot satisfy Eq.~\eqref{Eqn: RUpdate}.
\end{lemma}
%
%Let us continue the previous example. Let $L = \{ 1\}$. We have $|C_{L}^d(\G) \cap C^d(G_j)| = 7, 8, 9, 10, 11$ for $j = 2, 3, 4, 5, 6$, respectively. Thus, the order of the layer numbers is $6, 5, 4, 3, 2$. For $j = 4, 3, 2$, we have $|C_{L}^d(\G) \cap C^d(G_j)| < \frac{1}{k}|\Cov(\R)| + |\Delta(\R, C^*(\R))| \break = 10$. Actually, $C^{d}_{\{1, 4\}} = \{ a, b, c, d, f, x, k\}$, $C^{d}_{\{1, 3\}} = \{ a, b, c, d\}$ and $C^{d}_{\{1, 2\}} = \{ a, b, c, d, x\}$. None of them satisfies Eq.~\eqref{Eqn: RUpdate}.

Another optimization technique is called \emph{layer pruning}. For $\max(L) < j \le  l(\G)$, if $|\R| = k$ and $C^{d}_{L \cup \{j\}}(\G)$ does not satisfy Eq.~\eqref{Eqn: RUpdate}, we need not generate $C^{d}_{L'}(\G)$ for all $L'$ such that $L \cup \{j\} \subseteq L' \subseteq [l(\G)]$. The correctness is guaranteed by the following lemma.

\begin{lemma}[Layer Pruning]
\label{Lem: LayerPrune}
For a $d$-CC $C^d_{L}(\G)$ and $j > \max(L)$, if $C^d_{L \cup \{j\}}(\G)$ does not satisfy Eq.~\eqref{Eqn: RUpdate}, then $C^d_{L' \cup \{j\}}(\G)$ cannot satisfy Eq.~\eqref{Eqn: RUpdate} for all $L'$ such that $L \subseteq L' \subseteq [l(\G)]$.
\end{lemma}

%Following the previous example, let $L = \{1\}$ and $j = 6$. We have $C^{d}_{L \cup \{6\}} = \{a, b, c, d, n, w, y\}$, which does not satisfy Eq.~\eqref{Eqn: RUpdate}. For $L' = \{ 1, 2\}$, we have $C^{d}_{L' \cup \{6\}} = \{ a, b, c, d\}$, which does not satisfy Eq.~\eqref{Eqn: RUpdate}. Similarly, for other $L'$ such that $L \subseteq L'$, $C^{d}_{L' \cup \{6\}}$ also does not satisfy Eq.~\eqref{Eqn: RUpdate}.

Fig.~\ref{Fig:BottomUpGen} describes the pseudocode of Procedure \textsf{BU-Gen}, which naturally follows the steps presented above. Here, we make a few necessary remarks. The input $L_Q$ is the set of layer numbers that cannot be used to expand $L$. They are obtained according to Lemma~\ref{Lem: LayerPrune} when generating the ascendants of $C^{d}_{L}(\G)$. Thus, the layer numbers possible to be added to $L$ are $L_P = \{ j | \max(L) <  j \leq l(\G)\} - L_Q$ (line~1). In \textsf{BU-Gen}, we use set $L_{R}$ to record the layer numbers that can actually be added to $L$ (lines~9 and 22). In lines~24--26, for each $j \in L_R$, we make a recursive call to \textsf{BU-Gen} to generate the descendants of $C^{d}_{L \cup \{ j\} }(\G)$. By Lemma~\ref{Lem: LayerPrune}, the layer numbers that cannot be added to $L'$ are $L_Q \cup (L_{P} - L_{R})$.

%\smallskip

%\begin{figure}[!t]
%    \scriptsize
%    \fbox{
%    \parbox{\figwidth}{
%    \textbf{Procedure} \textsf{InitTopK}$(\G, d, s, k, \R)$
%    \begin{algorithmic}[1]
%    \STATE $\R \gets \emptyset$
%    \FOR{$p \gets 1$ to $k$}
%        \STATE $i \gets \arg\max_{i \in [l(\G)]} |\Cov(\R \cup \{  C^{d}(G_i) \} )| - | \Cov(\R)|$
%        \STATE $L \gets \{ i \} $
%        \STATE $C \gets C^{d}(G_i)$
%        \FOR{$q \gets 1$ to $s-1$}
%            \STATE $j \gets \arg\max_{j \in [l(\G)] - L} |C \cap C^{d} (G_j)|$
%            \STATE $L \gets L \cup \{ j \} $
%            \STATE $C \gets C \cap C^{d} (G_j)$
%        \ENDFOR
%        \STATE $C' \gets \mathsf{dCC}(\G[C], L, d)$
%        \STATE $\mathsf{UpdateTopK}(\R, C')$
%    \ENDFOR
%    \RETURN $\R$
%    \end{algorithmic}
%    }}
%    \vspace{-2em}
%    \caption{The \textsf{InitTopK} Procedure.}
%    \label{Fig:RInit}
%    \vspace{-1em}
%\end{figure}

%% Section 4.4 %%
\subsection{Bottom-Up Algorithm}
\label{Sec:BApproach-5}

Fig.~7 describes the complete bottom-up DCCS algorithm \textsf{BU-DCCS}. Given a multi-layer graph $\G$ and three parameters $d, s, k \in \mathbb{N}$, we can solve the DCCS problem by calling \textsf{BU-Gen}$(\G, d, s, k, \emptyset, V(\G), \emptyset, \R)$ (line~10). To further speed up the algorithm, we propose three preprocessing methods.

\noindent{\underline{\bf Vertex Deletion.}}
Let $\Sup(v)$ denote the support number of layers $i$ such that $v \in C^d(G_i)$, where $i \in [l(\G)]$. If $\Sup(v) < s$, $v$ must not be contained in any $d$-CCs $C_L^d(\G)$ with $|L| = s$. Therefore, we can safely remove all these vertices from $\G$ and recompute the $d$-cores of all layers. This process is repeated until $\Sup(v) \geq s$ for all remaining vertices $v$ in $\G$. Lines~1--7 of \textsf{BU-DCCS} describe this preprocessing method.

\noindent{\underline{\bf Sorting Layers.}}
We sort the layers of $\G$ in descending order of $|C^d(G_i)|$, where $1 \le i \le l(\G)$. Intuitively, the larger $|C^d(G_i)|$ is, the more likely $G_i$ contains a large candidate $d$-CC. Although there is no theoretical guarantee on the effectiveness of this preprocessing method, it is indeed effective in practice. Line~9 of \textsf{BU-DCCS} applies this preprocessing method.

\noindent{\underline{\bf Initialization of $\R$.}}
The pruning techniques in \textsf{BU-Gen} are not applicable unless $|\R| = k$, so a good initial state of $\R$ can greatly enhance pruning power. We develop a greedy procedure \textsf{InitTopK} to initialize $\R$ so that $|\R| = k$. Due to space limits, the details of Procedure \textsf{InitTopK} is described in Appendix~D. Line~8 of \textsf{BU-DCCS} initializes $\R$ by Procedure \textsf{InitTopK}.

\begin{theorem}
The approximation ratio of \textsf{BU-DCCS} is $1/4$.
\end{theorem}

\begin{figure}[!t]
    \scriptsize
    \fbox{
    \parbox{\figwidth}{
    \textbf{Algorithm} \textsf{BU-DCCS}$(\G, d, s, k)$
    \begin{algorithmic}[1]
    \REPEAT
		\FOR{$i \gets 1$ \TO $l(\G)$}
        	\STATE compute the $d$-core $C^{d}(G_i)$ on graph $G_i$
		\ENDFOR
		\FOR{each $v \in V(\G)$}
			\IF{$\Sup(v) < s$}
				\STATE remove $v$ from $\G$
			\ENDIF
		\ENDFOR
	\UNTIL {$\Sup(v) \ge s$ for all $v \in V(\G)$}
    \STATE $\R \gets \mathsf{InitTopK}(\G, d, s, k)$
    \STATE sort all layer numbers in descending order of $|C^{d}(G_i)|$, where $i \in [l(\G)]$
	\STATE \textsf{BU-Gen}$(\G, d, s, k, \emptyset, V(\G), \emptyset, \R)$
	\RETURN $\R$
    \end{algorithmic}
    }}
    \vspace{-0.5em}
    \caption{The \textsf{BU-DCCS} Algorithm.}
    \label{Fig:ABottomup}
    \vspace{-3em}
\end{figure}

%%% Section 5 %%%
\section{Top-Down Algorithm}
\label{Sec:TApproach}

The bottom-up algorithm must traverse a search tree from the root down to level $s$. When $s \ge l(\G)/2$, the efficiency of the algorithm degrades significantly. As verified by the experiments in Section~\ref{Sec:PEvaluation}, the performance of the bottom-up algorithm is close to or even worse than the greedy algorithm when $s \ge l(\G)/2$. To handle this problem, we propose a top-down approach for the \textsc{DCCS} problem when $s \ge l(\G)/2$.

In this section, we assume $s \geq l(\G)/2$. In the top-down algorithm, we maintain a temporary top-$k$ result set $\R$ and update it in the same way as in the bottom-up algorithm. However, candidate $d$-CCs are generated in a top-down manner. The reverse in search direction makes the techniques in the bottom-up algorithm no longer suitable. Therefore, we propose a new candidate $d$-CC generation method and a series of new pruning techniques suitable for top-down search. The top-down algorithm attains an approximation ratio of $1/4$. As verified by the experiments in Section~\ref{Sec:PEvaluation}, the top-down algorithm is superior to the other algorithms when $s \ge l(\G)/2$.

%This section is organized as follows. We present the details of the key procedures in Sections~\ref{Sec:TApproach-2}--\ref{Sec:TApproach-4} and describe the complete algorithm in Section~\ref{Sec:TApproach-5}.

\subsection{Top-Down Candidate Generation}
\label{Sec:TApproach-2}

We first introduce how to generate $d$-CCs in a top-down manner. In the top-down algorithm, all $d$-CCs are conceptually organized as a search tree as illustrated in Fig.~\ref{Fig: TopDown Tree}, where $C^{d}_{L}(\G)$ is the parent of $C^{d}_{L'}(\G)$ if $L' \subset L$, $|L| = |L'| + 1$ and the only layer number $\ell \in L - L'$ satisfies $\ell > \max([l(\G)] - L)$. Except the root $C^{d}_{[l(\G)]}$, all $d$-CCs in the search tree has a unique parent. We generate candidate $d$-CCs by depth-first searching the tree from the root down to level $s$ and update the temporary result set $\R$ during search.

Let $C^{d}_{L}(\G)$ be the $d$-CC currently visited in DFS, where $|L| > s$. We must generate the children of $C^{d}_{L}(\G)$. By Property~\ref{Lem:PHierarchy} of $d$-CCs, we have $C^{d}_{L}(\G) \subseteq C^{d}_{L'}(\G)$ for all $L' \subseteq L$. Thus, to generate $C^{d}_{L'}(\G)$, we only have to add some vertices to $C^{d}_{L}(\G)$ but need not to delete any vertex from $C^{d}_{L}(\G)$.

To this end, we associate $C^{d}_{L}(\G)$ with a vertex set $U^{d}_{L}(\G)$. $U^{d}_{L}(\G)$ must contain vertices in all descendants $C^{d}_{S}(\G)$ of $C^{d}_{L}(\G)$ such that $|S| = s$. $U^{d}_{L}(\G)$ serves as the scope for searching for the descendants of $C^{d}_{L}(\G)$. We call $U^{d}_{L}(\G)$ the \emph{potential vertex set} w.r.t.~$C^{d}_{L}(\G)$. Obviously, we have $C^{d}_{L}(\G) \subseteq U^{d}_{L}(\G)$. Initially, $U^{d}_{[l(\G)]}(\G) = V(\G)$. Section~\ref{Sec:TApproach-3} will describe how to shrink $U^{d}_{L}(\G)$ to $U^{d}_{L'}(\G)$ for $L' \subseteq L$, so we have $U^{d}_{L'}(\G) \subseteq U^{d}_{L}(\G)$ if $L' \subseteq L$. The relationships between $C^{d}_{L}(\G)$, $U^{d}_{L}(\G)$, $C^{d}_{L'}(\G)$ and $U^{d}_{L'}(\G)$ are illustrated in Fig.~\ref{Fig: CURelationship}. The arrows in Fig.~\ref{Fig: CURelationship} indicates that $C^{d}_{L'}(\G)$ is expanded from $C^{d}_{L}(\G)$, and $U^{d}_{L'}(\G)$ is shrunk from $U^{d}_{L}(\G)$. Keeping this in mind, we focus on top-down candidate generation in this subsection. Sections~\ref{Sec:TApproach-3} and~\ref{Sec:TApproach-4} will describe how to compute $U^{d}_{L'}(\G)$ and $C^{d}_{L'}(\G)$, respectively.

\begin{figure}[!t]
    \scriptsize
    \fbox{
    \parbox{\figwidth}{
    \textbf{Procedure} \textsf{TD-Gen}$(\G,\!d,\!s,\!k,\!L,\!C^{d}_{L}(\G),\!U^{d}_{L}(\G),\!\R)$\!\!\!\!\!\!
    \begin{algorithmic}[1]
    \STATE $L_{R} = \{ j | \max([l(\G)] - L) < j \le l(\G) \} \cap L$
    \FOR{each $j \in L_{R}$}
        \STATE $L' \gets  L - \{ j \}$
        \STATE $U^{d}_{L'}(\G) \gets \textsf{RefineU}(\G, d, s, U^{d}_{L}(\G), L')$
        \STATE $C^{d}_{L'}(\G) \gets \textsf{RefineC}(\G, d, s, U^{d}_{L'}(\G), L')$
    \ENDFOR
    \IF{$|\R| < k$}
        \FOR{each $j \in L_{R}$}
            \STATE $L' \gets  L - \{ j \}$
            \IF{$|L'| = s$}
                \STATE \textsf{Update}$(\R, C^{d}_{L'}(\G))$
            \ELSE
                \STATE \textsf{TD-Gen}$(\G, d, s, k, L, C^{d}_{L'}(\G), U^{d}_{L'}(\G),\R)$
            \ENDIF
        \ENDFOR
    \ELSE
        \STATE  sort $j \in L_{R}$ in descending order of $|U^{d}_{L - \{ j \} }(\G)|$
        \FOR{each $j$ in the sorted $L_{R}$}
            \STATE  $L' \gets  L - \{ j \}$
            \IF{$|U^{d}_{L'}(\G)| < |\Cov(\R)|/k + |\Delta(\R, C^*(\R))|$}
                \STATE \textbf{break}
            \ELSE
                    \IF{$|L'| = s$}
                        \STATE \textsf{Update}$(\R, C^{d}_{L'}(\G))$
                    \ELSE
                    \IF{$C^{d}_{L'}(\G)$ satisfies Eq.~\eqref{Eqn: RUpdate}}
                        \IF{$U^{d}_{L'}(\G)$ satisfies Eq.~\eqref{Eqn: UCondition}}
                            \STATE $S \gets L' \!-\! \{|L'| \!-\! s$ numbers randomly chosen from $L_R\}$
                            \STATE $C^{d}_{S}(\G) \gets \textsf{dCC}(\G[U^{d}_{L'}(\G)], S, d)$
                            \STATE \textsf{Update}$(\R, C^{d}_{S}(\G))$
                        \ELSE
                            \STATE \textsf{TD-Gen}$(\G,d,s,k,L,C^{d}_{L'}(\G),U^{d}_{L'}(\G),\R)$
                        \ENDIF
                    \ENDIF
                \ENDIF
            \ENDIF
        \ENDFOR
    \ENDIF
    \end{algorithmic}
    }}
    \vspace{-0.5em}
    \caption{The \textsf{\small TD-Gen} Procedure.}
    \label{Fig: TopDownGen}
    \vspace{-3em}
\end{figure}

The top-down candidate $d$-CC generation is implemented by the recursive procedure \textsf{TD-Gen} in Fig.~\ref{Fig: TopDownGen}. Let $L_{R} = \{ j | \max([l(\G)] - L) < j \le l(\G)\} \cap L$ be the set of layer numbers possible to be removed from $L$ (line~1). For each $j \in L_{R}$, let $L' = L - \{j\}$. We have that $C^{d}_{L'}(\G)$ is a child of $C^{d}_{L}(\G)$. We first obtain $U^{d}_{L'}(\G)$ and $C^{d}_{L'}(\G)$ by the methods in Section~\ref{Sec:TApproach-3} (line~4) and Section~\ref{Sec:TApproach-4} (line~5), respectively. Next, we process $C_{L'}^d(\G)$ based on the following cases:

\noindent\underline{\textbf{Case~1} (lines~9--10):} If $|\R| < k$ and $|L'| = s$, we update $\R$ with $C_{L'}^d(\G)$ by Rule~1 specified in Section~\ref{Sec:BApproach-2}.

\noindent\underline{\textbf{Case~2} (lines~11--12):}  If $|\R| < k$ and $|L'| > s$, we recursively call \textsf{TD-Gen} to generate the descendants of $C_{L'}^d(\G)$.

\noindent\underline{\textbf{Case~3} (lines~20--21):} If $|\R| = k$ and $|L'| = s$, we update $\R$ with $C_{L'}^d(\G)$ by Rule~2 specified in Section~\ref{Sec:BApproach-2}.

\noindent\underline{\textbf{Case~4} (lines~22--29):} If $|\R| = k$ and $|L'| > s$,  we check if $U^d_{L'}(\G)$ satisfies Eq.~\eqref{Eqn: RUpdate} to update $\R$ (line~23). If it is not satisfied, none of the descendants of $C_{L'}^d(\G)$ is qualified to be a candidate $d$-CC, so we prune the entire subtree rooted at $C_{L'}^d(\G)$. Otherwise, we recursively call \textsf{TD-Gen} to generate the descendants of $C_{L'}^d(\G)$ (line~29). The correctness of the pruning method is guaranteed by the following lemma.

\begin{lemma}[Search Tree Pruning]
\label{Lem: TDdCCPrune}
For a $d$-CC $C^d_L(\G)$ and its potential vertex set $U^d_L(\G)$, where $|L| > s$, if $U^d_{L}(\G)$ does not satisfy Eq.~\eqref{Eqn: RUpdate}, any descendant $C^d_{L'}(\G)$ of $C^d_L(\G)$ with $|L'| = s$ cannot satisfy Eq.~\eqref{Eqn: RUpdate}.
\end{lemma}

%For example, let us consider the multi-layer graph in Figure~\ref{Fig:ExpGraph}. Let $d = 3$, $s = 4$, $k = 2$ and $L = \{ 1, 2, 3, 4, 5, 6\}$. We have $C^d_{L}(\G) = \{ a, b, c, d\}$ and $U^{d}_{L}(\G) = V(\G) - C^{d}_{L}(\G)$. Suppose $\R = \{ \{ a, b, c, d, e, f, g, h\}, \{ a, b, c, d, w, x, y, z\} \}$. We have $|\Cov(\R)| = 12$. For $L' = \{2,3,4,5,6\}$, we have $C^d_{L'}(\G) = \{ a, b, c, d\}$ and $U^d_{L'}(\G) = \{ e, g, i, k, x, z \}$. Therefore, $|\Cov((\R - \{C^*(\R)\}) \cup \{C^d_{L'}(\G) \cup U^d_{L'}(\G)\})| = 12$, so $C^d_{L'}(\G)$ does not satisfy Eq.~\eqref{Eqn: RUpdate}. The descendants of $C_{L'}^d(\G)$ include $C_{\{3, 4, 5, 6\}}^d(\G) = \{ a, b, c, d\}$, $C_{\{2, 4, 5, 6\}}^d(\G) = \{ a, b, c, d\}$, $C_{\{2, 3, 5, 6\}}^d(\G) = \{ a, b, c, d\}$ and $C_{\{2, 3, 4, 5\}}^d(\G) = \{ a, b, c, d, x, z\}$. It is easy to verify that none of them can satisfy Eq.~\eqref{Eqn: RUpdate}.

To make top-down candidate $d$-CC generation even faster, we further propose some methods to prune the search tree.

If $|\R| = k$ (Cases~3 and~4), we order the layer numbers $j \in L_{R}$ in descending order of $|U^{d}_{L - \{ j \} }(\G)|$ (line~14). For some $j \in L_R$, if $|U^{d}_{L - \{ j \} }(\G)| < \frac{|\Cov(\R)|}{k} + |\Delta(\R, C^*(\R))|$, we need not to consider all layer numbers in $L_R$ succeeding $j$ and can terminate searching the subtrees rooted at $C^{d}_{L - \{j\}}(\G)$ immediately (lines~17--18). The correctness of this pruning method is ensured by the following lemma.

\begin{lemma}[Order-based Pruning]
\label{Lem: TDLayerOrder}
For a $d$-CC $C^d_L(\G)$, its potential vertex set $U^d_L(\G)$ and $j > \max ([l(\G)] - L)$, if $|U^d_{L - \{ j\}}(\G)| < \frac{|\Cov(\R)|}{k} + |\Delta(\R, C^*(\R))|$, any descendant $C^d_{L - \{ j \}}(\G)$ of $C^d_L(\G)$ cannot satisfy Eq.~\eqref{Eqn: RUpdate}.
\end{lemma}

%Let us continue the previous example. Let $L' = L - \{ j \}$. We have $|C^d_{L'}(\G) \cup U^d_{L'}(\G)| = 10, 7,  5,  4,  4, 4$ for $j = 1, 2, 3, 4, 5, 6$, respectively. Thus, the order of layer numbers is $1, 2, 3, 4, 5, 6$. For $j = 2, 3, 4, 5, 6$, we have $|C^d_{L'}(\G) \cup U^d_{L'}(\G)| < \tfrac{1}{k}|\Cov(\R)| + |\Delta(\R, C^*(\R))| = 10$. It can be easily verified that none of the descendants of $C_{\{1, 3, 4, 5, 6\}}^d(\G)$, $C_{\{1, 2, 4, 5, 6\}}^d(\G)$, $C_{\{1, 2, 3, 5, 6\}}^d(\G)$, $C_{\{1, 2, 3, 4, 6\}}^d(\G)$ and $C_{\{1, 2, 3, 4, 5\}}^d(\G)$ can satisfy Eq.~\eqref{Eqn: RUpdate}.

More interestingly, for Case~4, in some optimistic cases, we need not to search the descendants of $C^{d}_{L}(\G)$. Instead, we can randomly select a descendant $C^{d}_{S}(\G)$ of $C^{d}_{L}(\G)$ with $|S| = s$ to update $\R$ (lines~25--27).  The correctness is ensured by the following lemma.

\begin{lemma}[Potential Set Pruning]
\label{Lem: TDUPPrune}
For a $d$-CC $C^d_L(\G)$ and its potential vertex set $U^d_L(\G)$, where $|L| > s$, if $C^{d}_{L}(\G)$ satisfies Eq.~\eqref{Eqn: RUpdate}, and $U^d_L(\G)$ satisfies
\begin{equation}
%\small
\label{Eqn: UCondition}
|U^{d}_{L}(\G)| <  (\tfrac{1}{k} + \tfrac{1}{k^2}) |\Cov(\R)|+ (1 + \tfrac{1}{k})|\Delta(\R, C^{*}(\R))|,
\end{equation}
the following proposition holds: For any two distinct descendants $C^d_{S_1}(\G)$ and $C^d_{S_2}(\G)$ of $C^{d}_{L}(\G)$ such that $|S_1| = |S_2| = s$, if $|\R| = k$ and $\R$ has already been updated by $C^{d}_{S_1}(\G)$, then $C^{d}_{S_2}(\G)$ cannot update $\R$ any more.
\end{lemma}

%Following the previous example and suppose $\R = \{ \{ a, b, c, d, f, w\}, \{ a, b, c, d, h, y\} \}$. Let $L' = \{2,3,4,5,6\}$. We have $C^d_{L'}(\G) = \{ a, b, c, d\}$ and $U^d_{L'}(\G) = \{ e, g, i, k, x, z \}$. Thus, $|U^d_{L'}(\G)\})| = 6 < \left(\tfrac{1}{k} + \tfrac{1}{k^2}\right) |\Cov(\R)| + \left(1 + \frac{1}{k}\right)|\Delta(\R, C^{*}(\R))| = 6.5$. If we update $\R$ by $C_{\{2, 3, 4, 5\}}^d(\G) = \{ a, b, c, d, x, z\}$, we have $\R = \{ \{ a, b, c, d, x, z\}, \{ a, b, c, d, h, y\} \}$. We can easily verify that other descendants $C_{\{3, 4, 5, 6\}}^d(\G) = \{ a, b, c, d\}$, $C_{\{2, 4, 5, 6\}}^d(\G)  = \{ a, b, c, d\}$ and $C_{\{2, 3, 5, 6\}}^d(\G) = \{ a, b, c, d\}$ of $C^d_{L'}(\G)$ cannot update $\R$.

%% Section 5.2 %%
\subsection{Refinement of Potential Vertex Sets}
\label{Sec:TApproach-3}

Let $C^{d}_{L}(\G)$ be the $d$-CC currently visited by DFS and $C^{d}_{L'}(\G)$ be a child of $C^{d}_{L}(\G)$. To generate $C^{d}_{L'}(\G)$, Procedure \textsf{TD-Gen} first refines $U^{d}_{L}(\G)$ to $U^{d}_{L'}(\G)$ and then generates $C^{d}_{L'}(\G)$ based on $U^{d}_{L'}(\G)$. This subsection introduces how to shrink $U^{d}_{L}(\G)$ to $U^{d}_{L'}(\G)$.

First, we introduce some useful concepts. Given a subset of layer numbers $L \subseteq [l(\G)]$, we can divide all layer numbers in $L$ into two disjoint classes:

%\noindent \underline{\bf Class 1:} For any layer $\ell \in [l(\G)] - L$, $\ell$ has already been removed from $L$. So for any descendant $C^d_S(\G)$ of $C^{d}_{L}(\G)$ with $|S| = s$, $\ell$ must not in $S$.

\noindent \underline{\bf Class 1:}
By the relationship of $d$-CCs in the top-down search tree, for any layer number $\ell \in L$ and $\ell < \max([l(\G)] - L)$,  $\ell$ will not be removed from $L$ in any descendant of $C^{d}_{L}(\G)$. Thus, for any descendant $C^d_S(\G)$ of $C^{d}_{L}(\G)$ with $|S| = s$, we have $l \in S$.

\noindent \underline{\bf Class 2:}
By the relationship of $d$-CCs in the top-down search tree, for any layer number $\ell \in L$ and $\ell > \max([l(\G)] - L)$,  $\ell$ can be removed from $L$ to obtain a descendant of $C^{d}_{L}(\G)$. Thus, for a descendant $C^d_S(\G)$ of $C^{d}_{L}(\G)$ with $|S| = s$, it is undetermined whether $\ell \in S$.

%We illustrate the relationships of the two classes of layers in Fig.~\ref{Fig: LayerClasses}.
%Notably, for a certain $L$, $\max([l(\G)] - L)$ is a fixed number. So the layer numbers of Class~2 is the consecutive sub-sequence located at the end of the sequence of $L$. The remaining layer numbers in $L$ are in Class~1.

Let $M_L$ and $N_L$ denote the Class~1 and Class~2 of layer numbers w.r.t.~$L$, respectively. Procedure \textsf{RefineU} in Fig.~\ref{Fig: URefine} refines $U^{d}_{L}(\G)$ to $U^{d}_{L'}(\G)$. Let $U = U^{d}_{L}(\G)$ (line~1). First, we obtain $M_{L'}$ and $N_{L'}$ w.r.t.~$L'$  (line~2). Then, we apply them to repeat the following two refinement methods to remove irrelevant vertices from $U$ until no vertices can be removed any more (lines~3--8). Finally, $U$ is output as $U_{L'}^{d}(\G)$ (line~9).

\begin{figure}[!t]
    \scriptsize
    \fbox{
    \parbox{\figwidth}{
    \textbf{Procedure} \textsf{RefineU}$(\G,d,s, U^{d}_{L}(\G), L')$
    \begin{algorithmic}[1]
    \STATE $U \gets U^{d}_{L}(\G)$
    \STATE $M_{L'} \gets \{ j | j \in L, j < \max([l(\G)] - L) \}$, $N_{L'} \gets L - M_{L'}$
    \REPEAT
		\WHILE{there exists $v \in U$ and $i \in M_{L'}$ such that $d_{G_i[U]}(v) <  d$}
			\STATE remove $v$ from $U$ and all layers of $\G$
		\ENDWHILE
    		\WHILE{there exists $v \in U$ that occurs in less than $s - |M_{L'}|$ of the $d$-cores $C^{d}(G_j)$ for $j \in N_{L'}$}
			\STATE remove $v$ from $U$ and all layers of $\G$
		\ENDWHILE
    \UNTIL{no vertex in $U$ can be removed}
    \RETURN $U$
    \end{algorithmic}
    }}
    \vspace{-0.5em}
    \caption{The \textsf{\small RefineU} Procedure.}
    \label{Fig: URefine}
    \vspace{-3em}
\end{figure}

%\smallskip

\noindent{\underline{\textbf{Refinement Method~1} (lines~4--5):}}
For each layer number $i \in M_{L'}$, we have $i \in S$ for all descendants $C^{d}_{S}(\G)$ of $C^{d}_{L'}(\G)$ with $|S| = s$. Note that $C^{d}_{S}(\G)$ must be $d$-dense in $G_i$. Thus, if the degree of a vertex $v$ in $G_i[U]$ is less than $d$, we have $v \not\in C^{d}_{S}(\G)$, so we can remove $v$ from $U$ and $\G$.

%\smallskip

\noindent{\underline{\textbf{Refinement Method~2} (lines~6--7):}}
If a vertex $v \in U$ is contained in a descendant $C^{d}_{S}(\G)$ of $C^{d}_{L}(\G)$ with $|S| = s$, $v$ must occur in all the $d$-cores $C^d(G_i)$ for $i \in M_{L'}$ and must occur in at least $s - |M_{L'}|$ of the $d$-cores $C^d(G_j)$ for $j \in N_{L'}$. Therefore, if $v$ occurs in less than $s - |M_{L'}|$ of the $d$-cores $C^d(G_j)$ for $j \in N_{L'}$, we can remove $v$ from $U$ and $\G$.

%% Section 5.3 %%
\subsection{Refinement of d-CCs}
\label{Sec:TApproach-4}

Let $C^{d}_{L}(\G)$ be the $d$-CC currently visited by DFS and $C^{d}_{L'}(\G)$ be a child of $C^{d}_{L}(\G)$, where $|L| > s$. Since $C^{d}_{L'}(\G) \subseteq U^{d}_{L'}(\G)$, Procedure \textsf{dCC} in Section~\ref{Sec:GAlgorithm} can find $C^{d}_{L'}(\G)$ on $\G[U^{d}_{L'}(\G)]$ from scratch. However, this straightforward method is not efficient. In this subsection, we propose an more efficient algorithm to construct $C^{d}_{L'}(\G)$ by adopting two techniques: 1) An index structure that helps eliminate more vertices in $U^{d}_{L'}(\G)$ irrelevant to $C^{d}_{L'}(\G)$. 2) A search strategy with early termination to find $C^{d}_{L'}(\G)$ efficiently.

\noindent{\underline{\bf Index Structure.}}
First, we introduce an index structure that organizes all vertices of $\G$ hierarchically and helps filter out the vertices irrelevant to $C^{d}_{L'}(\G)$ efficiently. Recall that $\Sup(v)$ is the number of layers whose $d$-cores contain $v$. Values $\Sup(v)$ are used to determine the vertices in $U^{d}_{L'}(\G)$ that are not in $C^{d}_{L'}(\G)$. Specifically, for $h \in \mathbb{N}$, let $J_h$ be the set of vertices $v$ iteratively removed from $\G$ due to $\Sup(v) \leq h$. Let $I_h = J_{h} - J_{h-1}$. Obviously, $I_1, I_2, \dots, I_{l(\G)}$ is a disjoint partition of all vertices of $\G$. Based on this partition, we can narrow down the search scope of $C^{d}_{L'}(\G)$ from $U^{d}_{L'}(\G)$ to $U^{d}_{L'}(\G) \cap (\bigcup_{h = |L'|}^{l(\G)} I_{h})$ according to the following lemma.

\begin{lemma}
\label{Lem:CScope}
$C^{d}_{L'}(\G) \subseteq U^{d}_{L'}(\G) \cap \left(\bigcup_{h = |L'|}^{l(\G)} I_{h} \right)$.
\end{lemma}

The index structure is basically the hierarchy of vertices following $I_1, I_2, \dots, I_{l(\G)}$, that is, the vertices in $I_i$ are placed on a lower level than those in $I_{i + 1}$. Internally, the vertices in $I_i$ are also placed on a stack of levels, which is determined as follows. Suppose the vertices in $I_1, I_2, \ldots, I_{i - 1}$ have been removed from $\G$. Although the vertices $v \in I_i$ are iteratively removed from $\G$ due to $\Sup(v) \le i$, they are actually removed in different batches. In each batch, we select all the vertices $v$ with $\Sup(v) \le i$ and remove them together. After a batch, some vertices $v$ originally satisfying $\Sup(v) > i$ may have $\Sup(v) \le i$ and thus will be removed in next batch. Therefore, in $I_i$, the vertices removed in the same batch are place on the same level, and the vertices removed in a later batch are placed on a higher level than the vertices removed in an early batch. In addition, let $L(v)$ be the set of layer numbers on which $v$ is contained in the $d$-core just before $v$ is removed from $\G$ in batch. We associate each vertex $v$ in the index with $L(v)$. Moreover, we add an edge between vertices $u$ and $v$ in the index if $(u, v)$ is an edge on a layer of $\G$.

By Lemma~\ref{Lem:CScope}, we have narrowed down the search scope of $C^{d}_{L'}(\G)$ from $U^{d}_{L'}(\G)$ to $Z = U^{d}_{L'}(\G) \cap (\bigcup_{h = |L'|}^{l(\G)} I_{h})$. By exploiting the index, we can further narrow down the search scope. If there is no sequence of vertices $w_0, w_1, \dots, w_n$ in the index such that $L' \subseteq L(w_0)$, $w_n = v$, $w_i$ is on a higher level than $w_{i + 1}$, and $(w_i, w_{i + 1})$ is an edge in the index, then $v$ must not be contained in $C^{d}_{L'}(\G)$. The correctness of this method is guaranteed by the following lemma.

\begin{lemma}
\label{Lem:CFilter}
For each vertex $v\!\in\!C^{d}_{L'}(\G)$, there exists a sequence of vertices $w_0, w_1, \dots, w_n$ in the index such that $L' \subseteq L(w_0)$, $w_n = v$, $w_{i+1}$ is placed on a higher level than $w_{i}$, and $(w_i, w_{i + 1})$ is an edge in the index.
\end{lemma}

\begin{figure}[!t]
    \scriptsize
    \fbox{
    \parbox{\figwidth}{
    \textbf{Procedure} \textsf{RefineC}$(\G,d,s, U^{d}_{L'}(\G),L')$\!\!\!\!\!
    \begin{algorithmic}[1]
    \STATE $Z = U^{d}_{L'}(\G) \cap (\bigcup_{h = |L'|}^{l(\G)} I_{h} )$
    \STATE removed all vertices not in $Z$ from the index
    \FOR{each vertex $v \in Z$}
        \STATE set all vertices in $Z$ as unexplored
        \STATE compute $d^{+}_{i}(v)$ of all $i \in L'$
    \ENDFOR
    \FOR{each level of the index}
        \IF{all vertices are unexplored or discarded on the level}
            \FOR{each unexplored vertex $v$ on the level}
                \IF{$L'  \not\subseteq L(v)$}
                    \STATE set $v$ as discarded
                    \STATE \textsf{CascadeD}$(\G, v, d, L')$
                \ELSE
                    \IF{$v$ is not discarded}
                        \STATE set $v$ as undetermined
                        \FOR{each unexplored neighbor $u$ of $v$ on a higher level}
                            \STATE set $u$ as undetermined
                        \ENDFOR
                    \ENDIF
                \ENDIF
            \ENDFOR
        \ELSE
            \FOR{each undetermined vertex $v$ on the level}
                \IF{$d^{+}_{i}(v) < d$ for some $i \in L'$}
                    \STATE set $v$ as discarded
                    \STATE \textsf{CascadeD}$(\G, v, d, L')$
                \ELSE
                    \FOR{each unexplored neighbor $u$ of $v$ on a higher level}
                        \STATE set $v$ as undetermined
                    \ENDFOR
                \ENDIF
             \ENDFOR
            \FOR{each unexplored vertex $v$ on the level}
                 \STATE set $v$ as discarded
                 \STATE \textsf{CascadeD}$(\G, v, d, L')$
            \ENDFOR
        \ENDIF
    \ENDFOR
   \STATE $C_{L'}^{d}(\G) \gets \{  \text{all undetermined vertices in } Z \}$
    \RETURN $C_{L'}^{d}(\G)$
    \end{algorithmic}

    \textbf{Procedure} \textsf{CascadeD}$(\G, v, d, L')$
    \begin{algorithmic}[1]
        \FOR{each undetermined neighbor $u$ of $v$}
            \STATE $d_{i}^{+}(u) \gets d_{i}^{+}(u) - 1$ for each $i \in L'$ and $(u, v) \in E_{i}(\G)$
            \IF{$d_{i}^{+}(u) < d$ for some $i \in L'$}
                \STATE set $u$ as discarded
                \STATE \textsf{CascadeD}$(\G, u, d, L')$
            \ENDIF
        \ENDFOR
    \end{algorithmic}
    }}
     \vspace{-0.5em}
   \caption{The \textsf{\small RefineC} Procedure.}
    \label{Fig: CRefine}
    \vspace{-3.5em}
\end{figure}

\noindent{\underline{\bf Fast Search with Early Termination.}}
Based on the index, Procedure \textsf{RefineC} in Fig.~\ref{Fig: CRefine} searches for the exact $C^{d}_{L'}(\G)$. First, we obtain the search scope $Z = U^{d}_{L'}(\G) \cap (\bigcup_{h = |L'|}^{l(\G)} I_{h} )$ based on the index (line~1). By Lemma~\ref{Lem:CScope}, we only need to consider the vertices in $Z$. Thus, before the search begins, we can remove all the vertices not in $Z$ from the index (line~2).

Unlike Procedure \textsf{dCC} that only removes irrelevant vertices from $\G$, Procedure \textsf{RefineC} can find $C^{d}_{L'}(\G)$ much faster by using two strategies:
1) Identify some vertices not in $C^{d}_{L'}(\G)$ early;
2) Skip searching some vertices not in $C^{d}_{L'}(\G)$.
To this end, we set each vertex $v \in Z$ to one of the following three states:
1) $v$ is \emph{discarded} if it has been determined that $v \not\in C^{d}_{L'}(\G)$;
2) $v$ is \emph{undetermined} if $v$ has been checked, but it has not be determined whether $v \in C^{d}_{L'}(\G)$;
3) $v$ is \emph{unexplored} if it has not been checked by the search process.
During the search process, a discarded vertex will not be involved in the following computation, and an undetermined vertex may become discarded due to the deletion of some edges. Initially, all vertices in $Z$ are set to be unexplored (line~3).

For $i \in L'$, let $d_{i}^{+}(v)$ be the number of undetermined and unexplored vertices adjacent to $v$ in $G_i[Z]$. Clearly, $d_{i}^{+}(v)$ is an upper bound on the degree of $v$ in $G_i[Z]$. If $d_{i}^{+}(v) < d$ on some layer $i \in L'$, we must have $v \not\in C^{d}_{L'}(\G)$, so we can set $v$ as discarded. Notably, the removal of $v$ may trigger the removal of other vertices. The details are described in the \textsf{CascadeD} procedure. Specifically, if $v$ is discarded, for each undetermined vertex $u \in Z$ that is adjacent to $v$, we decrease $d_{i}^{+}(u)$ by $1$ if $(u, v)$ is an edge on a layer $i \in L'$. If $d_{i}^{+}(u) < d$ for some $i \in L'$, we also set $u$ as discarded and recursively invoke the \textsf{CascadeD} procedure to search for more discarded vertices starting from $u$.

In the main search process, we check the vertices in $Z$ in a level-by-level fashion. In each iteration (lines~6--27), we fetch all vertices on a level of the index and process them according to the following two cases:

\noindent\underline{{\bf Case~1} (lines~7--16):}
If there are only unexplored and discarded vertices on the current level, none of the vertices on this level has been checked before by the search process. At this point, we can check each unexplored vertex on this level. Specifically, for each unexplored vertex $v$, if $L' \not\subseteq L(v)$, we have $v \not\in C^{d}_{L'}(\G)$ by Lemma~\ref{Lem:CFilter}. Thus, we can immediately set $v$ as discarded and invoke Procedure \textsf{CascadeD} to explore more discarded vertices starting from $v$ (lines~10--11). Otherwise, if $v$ is not discarded, we set $v$ as undetermined (line~14). For each unexplored neighbor $u \in Z$ of $v$ placed on a higher level than $v$ in the index, we also set $u$ as undetermined since $u$ is possible to be contained in $C^{d}_{L'}(\G)$ (line~16).

\noindent\underline{{\bf Case~2} (lines~17--27):}
If there is some undetermined vertices on the current level, we carry out the following steps. For each undetermined vertex $v$ on this level, we check if $d_{i}^{+}(v)<d$ for some $i \in L'$ (line~19). If it is true, we have $v \not\in C^{d}_{L'}(\G)$. At this point, we set $v$ to be discarded and invoke Procedure \textsf{CascadeD} to explore more discarded vertices starting from $v$ (lines~20--21). Otherwise, $v$ remains to be undetermined. For each unexplored neighbor $u \in Z$ of $v$ placed on a higher level than $v$ in the index, we also set $u$ as undetermined since $u$ is possible to be contained in $C^{d}_{L'}(\G)$ (line~24).

For each vertex $v$ that is still unexplored on the current level, none of the vertices in $C^{d}_{L'}(\G)$ on lower levels than $v$ in the index is adjacent to $v$. By Lemma~\ref{Lem:CFilter}, we have $v \not\in C^{d}_{L'}(\G)$. Thus, we can directly set $v$ to be discarded and invoke Procedure \textsf{CascadeD} to explore more discarded vertices starting from $v$ (lines~26--27).

After examining all levels in the index, $C^{d}_{L'}(\G)$ is exactly the set of all undetermined vertices in $Z$ (lines~28--29).

\noindent{\underline{\bf Time Complexity.}}
Let $l' = |L'|$, $n' = U_{L'}^{d}(\G)$, $m'_i = E_i[U_{L'}^{d}(\G)]$ be the number of edges on layer $i$ of the induced multi-layer graph $\G[U_{L'}^{d}(\G)]$ and $m' = \sum_{i \in L'} m_i$. The following lemma shows that the time cost of the \textsf{RefineC} procedure is $O(n'l' + m')$. Notably, if we apply Procedure \textsf{dCC} on $\G[U_{L'}^{d}(\G)]$ to find $C^{d}_{L'}(\G)$ from scratch, the time cost is $O(n'l' + m''|L'|)$, where $m'' = |\bigcup_{i \in L'} E_i[U_{L'}^{d}(\G)]|$. Since $m' \leq m''l'$ always holds, the time cost of Procedure \textsf{RefineC} is no more than Procedure \textsf{dCC} .

\begin{lemma}
The time complexity of Procedure \textsf{RefineC} is $O(n'l' + m')$.
\end{lemma}

%Following the previous example, we only need to examine vertices $f, k, i, x$ and $z$. First, we consider vertex $f$ in the lowest level. We cannot determine the state of $f$. Therefore, we insert $k$ and $i$ into $Q$ for further exploration. For vertex $k$, we have $d_{3}^{+}(k)  = 1$ and $d_{3}^{-}(k) = 0$, so $k$ is marked as discarded. Then, we cascade to update $k$'s neighbors. For vertex $i$, we have $d_{4}^{+}(i) = 1$ and $d_{4}^{-}(i) = 1$, so $i$ is also marked as discarded. Next, we have $d_{3}^{+}(f) = 1$ and $d_{3}^{-}(f) = 1$. Therefore, $f$ is marked as discarded. After that, we insert vertices $x$ and $z$ into $Q$. The states of $x$ and $z$ remain undetermined after examination. Finally, $x$ and $z$ are marked as selected. We obtain $C^{d}_{\{3, 4, 5\}} = \{ a, b, c, d, x, z\}$.

%% Section 5.4 %%
\subsection{Top-Down Algorithm}
\label{Sec:TApproach-5}

\begin{figure}[!t]
    \scriptsize
    \fbox{
    \parbox{\figwidth}{
    \textbf{Algorithm} \textsf{TD-DCCS}$(\G, d, s, k)$
    \begin{algorithmic}[1]
	\STATE execute lines~1--8 of the \textsf{BU-DCCS} algorithm
	\STATE sort all layer numbers $i$ in ascending order of $|C^{d}(G_i)|$, where $i \in [l(\G)]$
   \STATE construct the index of $\G$
   \STATE $C_{[l(\G)]}^{d} \gets $ \textsf{dCC}$(\G, [l(\G)], d)$
   \STATE \textsf{TD-Gen}$(\G, d, s, k, [l(\G)], C_{[l(\G)]}^{d}, V(\G), \R)$
   \RETURN $\R$
    \end{algorithmic}
    }}
    \vspace{-0.5em}
    \caption{The \textsf{\small TD-DCCS} Algorithm.}
    \label{Fig:ATopDown}
    \vspace{-1.5em}
\end{figure}

The preprocessing methods proposed in Section~\ref{Sec:BApproach-5} can also be applied to the top-down DCCS algorithm. The method of vertex deletion and the method of initializing $\R$ can be directly applied. For the method of sorting layers, we sort all layers $i$ of $\G$ in ascending order of $|C^d(G_i)|$ since a layer whose $d$-core is small is less likely to support a large $d$-CC.

We present the complete top-down DCCS algorithm called \textsf{TD-DCCS} in Fig.~\ref{Fig:ATopDown}. The input is a multi-layer graph $\G$ and parameters $d, s, k \in \mathbb{N}$. First, we execute lines~1--8 of the bottom-up algorithm \textsf{BU-DCCS} to remove irreverent vertices and initialize $\R$. Then, we sort all layers $i$ of $\G$ in ascending order of $|C^d(G_i)|$ at line~2. We construct the index for $\G$ (line~3). Next, we invoke recursive Procedure \textsf{TD-Gen} to generate candidate $d$-CCs and update the result set $\R$ (line~5). Finally, $\R$ is returned as the result (line~6).

\begin{theorem}
The approximation ratio of \textsf{TD-DCCS} is $1/4$.
\end{theorem}

%%% Section 6 %%%
\section{Performance Evaluation}
\label{Sec:PEvaluation}

\begin{figure}
    \centering
    \scriptsize
    \begin{tabular}{lrrrr}
    	\hline
        \rowcolor{mygray}
    	Graph $\G$	& $|V(\G)|$ & $\sum_{i = 1}^{l(\G)} |E(G_i)|$ & $|\bigcup_{i = 1}^{l(\G)} E(G_i)|$ & $l(\G)$ \\
    	\hline
    	{\it PPI}	& 328 & 4,745 &  3,101 & 8 \\
    	{\it Author}	& 1,017 & 15,065 & 11,069 & 10 \\
        %{\it MaxFlow}	& 24,818 & 506,550 & 239,978 & 8 \\
        %{\it FaceBook} & 22,754 & 1,359,366 & 264,000 & 10\\
        {\it German}     & 519,365 & 7,205,624 & 1,653,621 & 14\\
        {\it Wiki}    & 1,140,149 & 7,833,140 & 3,309,592 & 24\\
        %{\it DBLP}	& 1,282,458 & 16,382,162 & 9,271,094 & 15\\
    	{\it English}	& 1,749,651 & 18,951,428 & 5,956,877  & 15\\
        {\it Stack}	& 2,601,977 & 63,497,050 & 36,233,450 & 24\\
        \hline
    \end{tabular}
    \caption{Statistics of Graph Datasets Used in Experiments.}
    \label{Tab: Datasets}
    \vspace{-1em}
\end{figure}

\begin{figure}[!t]
    \centering
    \scriptsize
    \resizebox{0.95\columnwidth}{!}{
    \begin{tabular}{lcc}
    	\hline
        \rowcolor{mygray}
    	Parameter	& Range & Default Value\\
    	\hline
    	$k$	& $\{5, 10, 15, 20, 25\}$ & $10$ \\
    	$d$	& $\{2, 3, 4, 5, 6\}$ & $4$ \\
        $s$ (small)	& $\{1, 2, 3, 4, 5\}$ & $3$ \\
        $s$ (large) & $\{l(\G)-4, l(\G)-3, l(\G)-2, l(\G)-1, l(\G)\}$ & $l(\G)-2$ \\
        $p$	& $\{0.2, 0.4, 0.6, 0.8, 1.0\}$ & $1.0$ \\
        $q$	& $\{0.2, 0.4, 0.6, 0.8, 1.0\}$ & $1.0$ \\
        \hline
    \end{tabular}}
    \vspace{-0.1em}
    \caption{Parameter Configuration.}
     \label{Tab: Parameters}
    \vspace{-3em}
\end{figure}

\begin{figure*}
   \addtolength{\subfigcapskip}{-1.2ex}
    \begin{minipage}[!t]{0.33\linewidth}
       % \subfigure[\st{German (Vary $s$)}]{\includegraphics[width=\figswidth]{./ExpFig/Exp1T3.pdf}}
        \subfigure[\st{English (Vary $s$)}]{\includegraphics[width=0.49\linewidth]{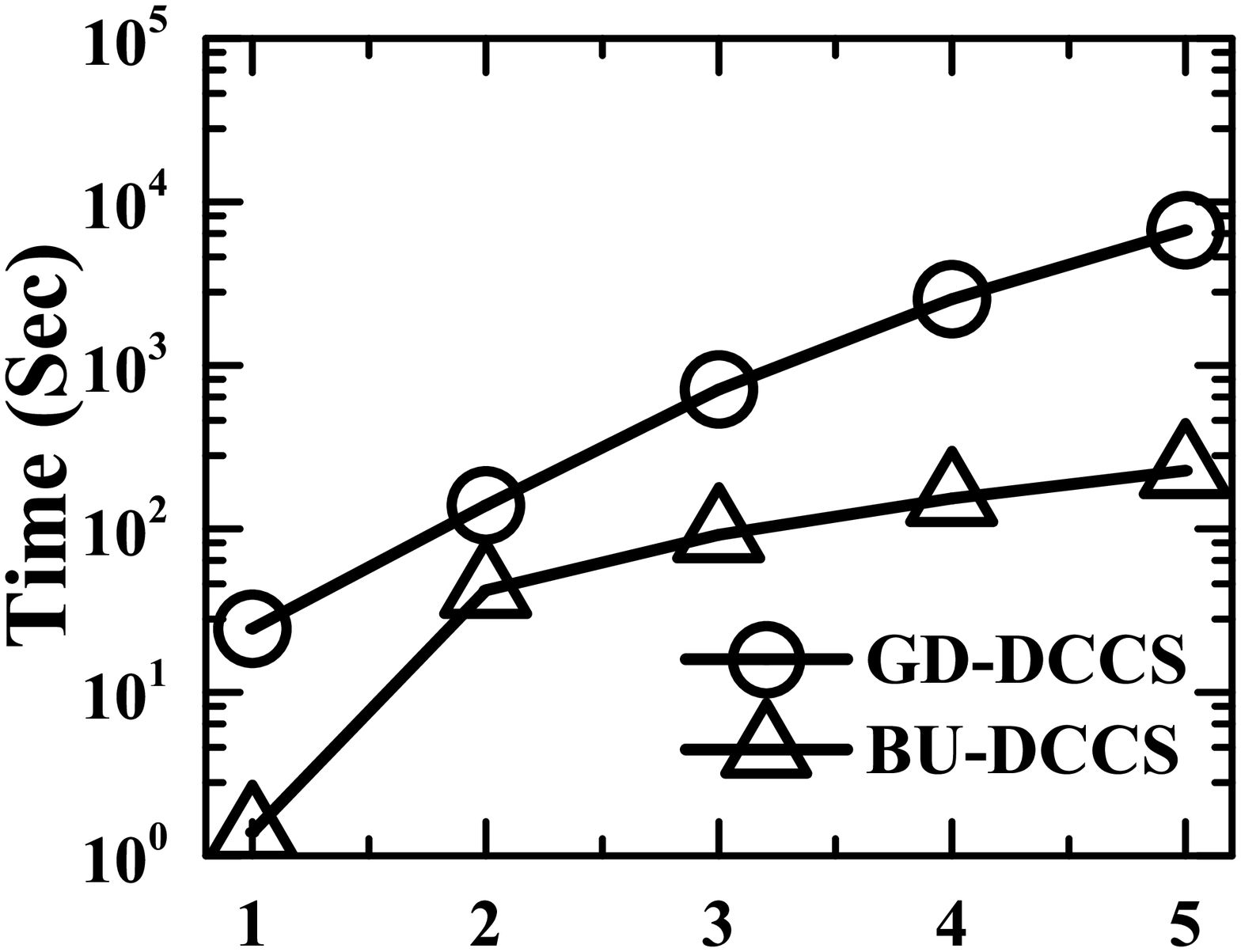}}
        \subfigure[\st{Stack (Vary $s$)}]{\includegraphics[width=0.49\linewidth]{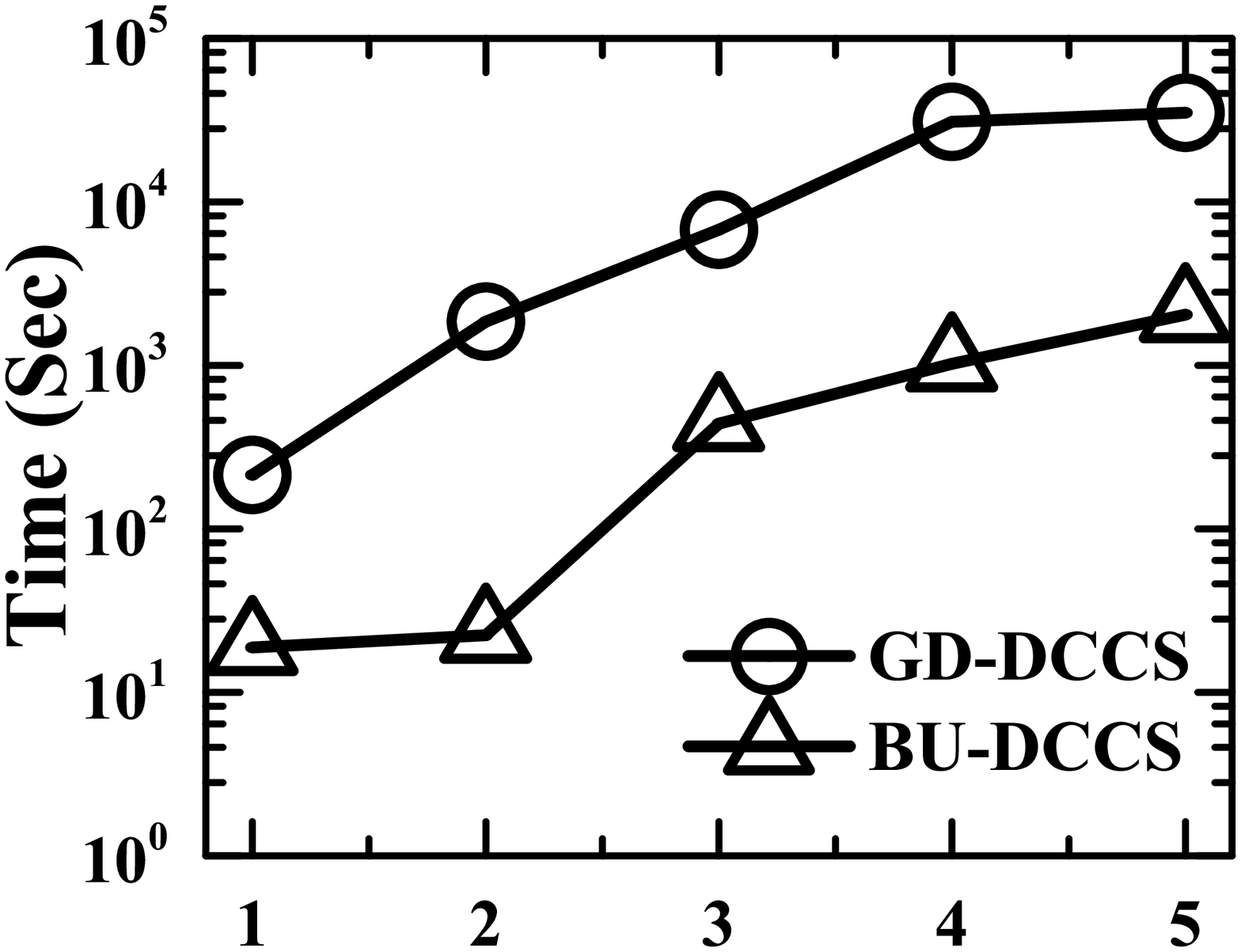}}
        \vspace{-1.5em}
        \caption{Execution Time vs Small $s$.}
        \label{Fig: Exp1T}
        \vspace{-1.2em}
    \end{minipage}
    \begin{minipage}[!t]{0.33\linewidth}
        %\subfigure[\st{German (Vary $s$)}]{\includegraphics[width=\figswidth]{./ExpFig/Exp1TL3.pdf}}
        \subfigure[\st{English (Vary $s$)}]{\includegraphics[width=0.49\linewidth]{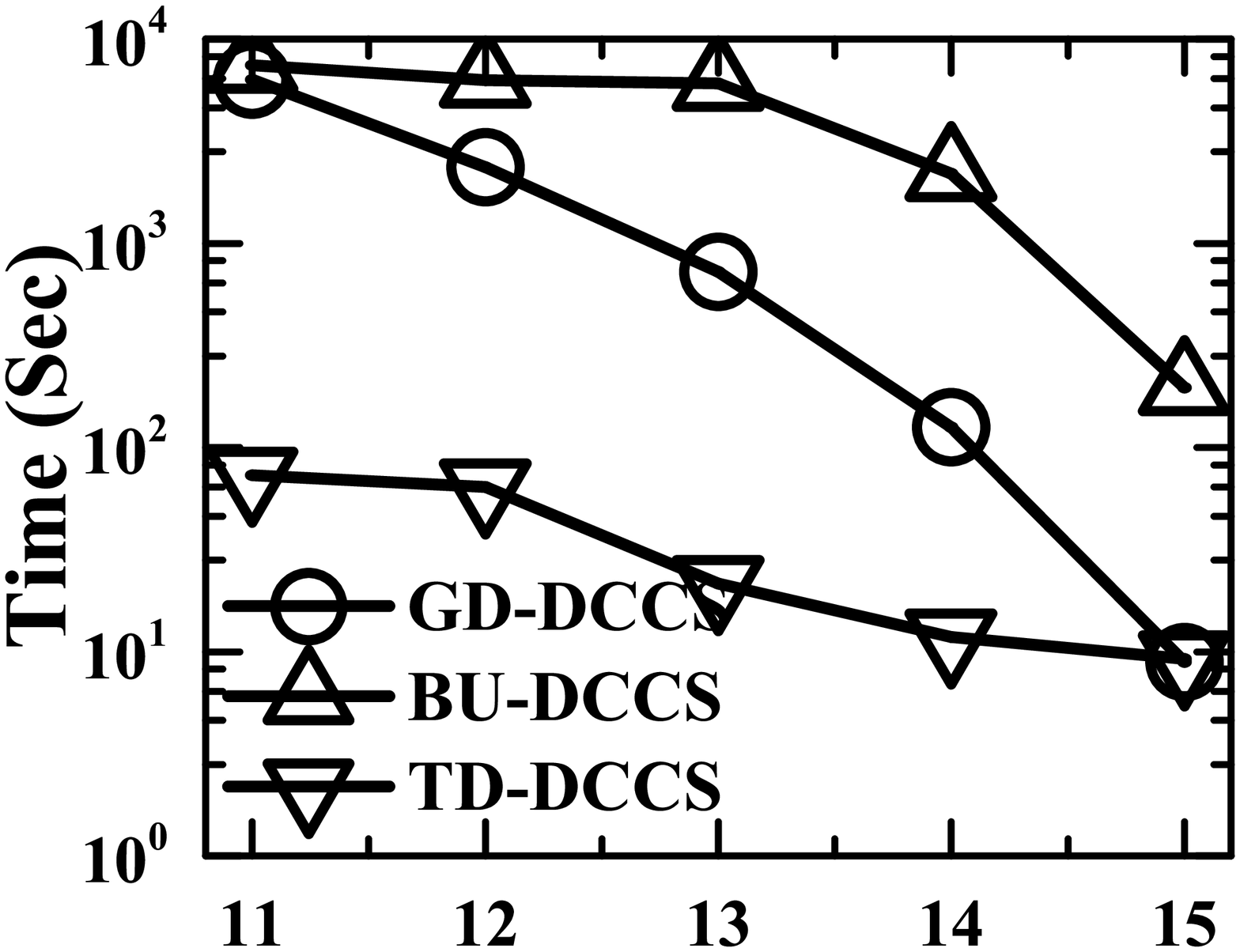}}
        \subfigure[\st{Stack (Vary $s$)}]{\includegraphics[width=0.49\linewidth]{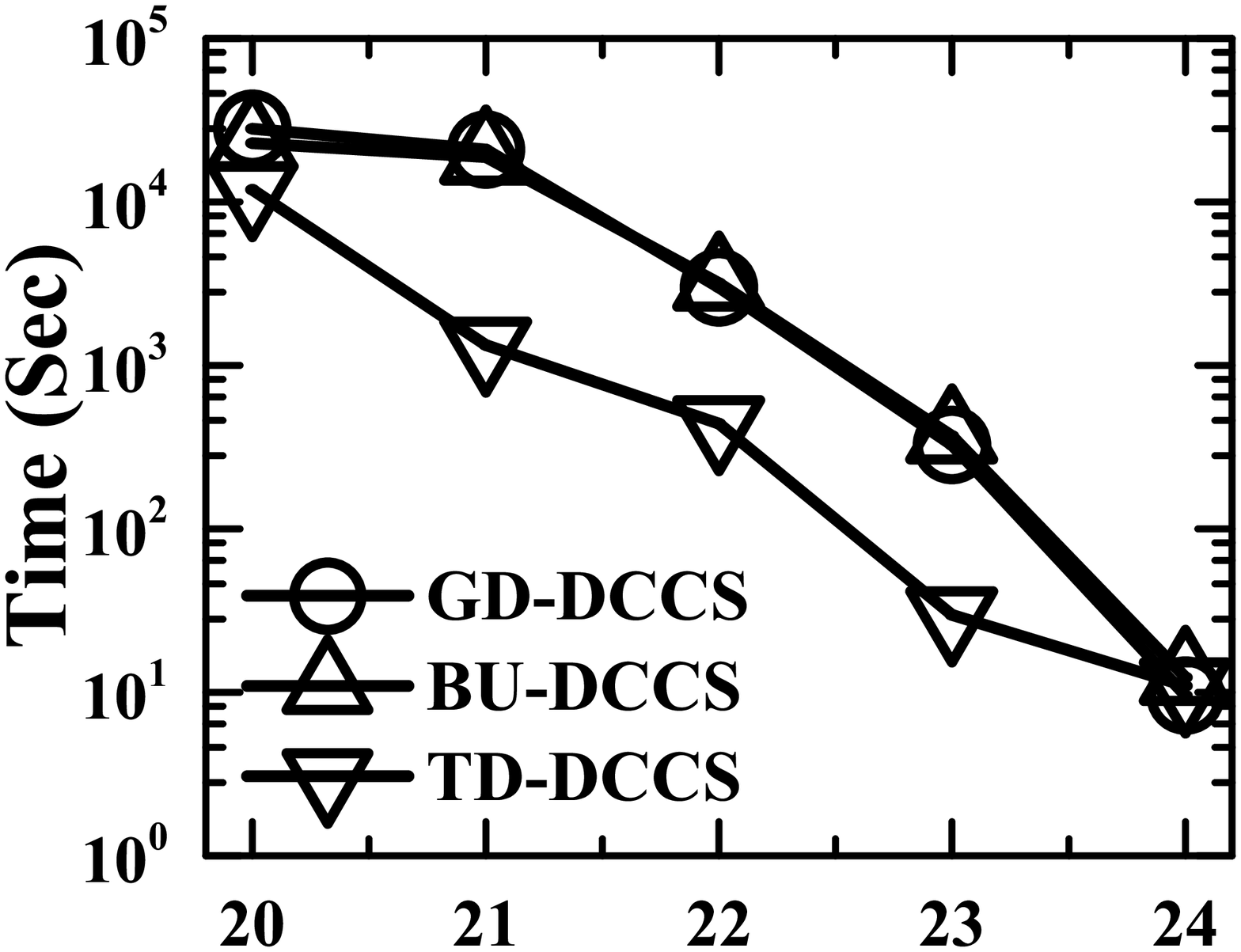}}
        \vspace{-1.5em}
        \caption{Execution Time vs Large $s$.}
        \label{Fig: Exp1TL}
        \vspace{-1.2em}
    \end{minipage}
        \begin{minipage}[!t]{0.33\linewidth}
        %\subfigure[\st{German (Vary $s$)}]{\includegraphics[width=\figswidth]{./ExpFig/Exp1C3.pdf}}
        \subfigure[\st{English (Vary $s$)}]{\includegraphics[width=0.49\linewidth]{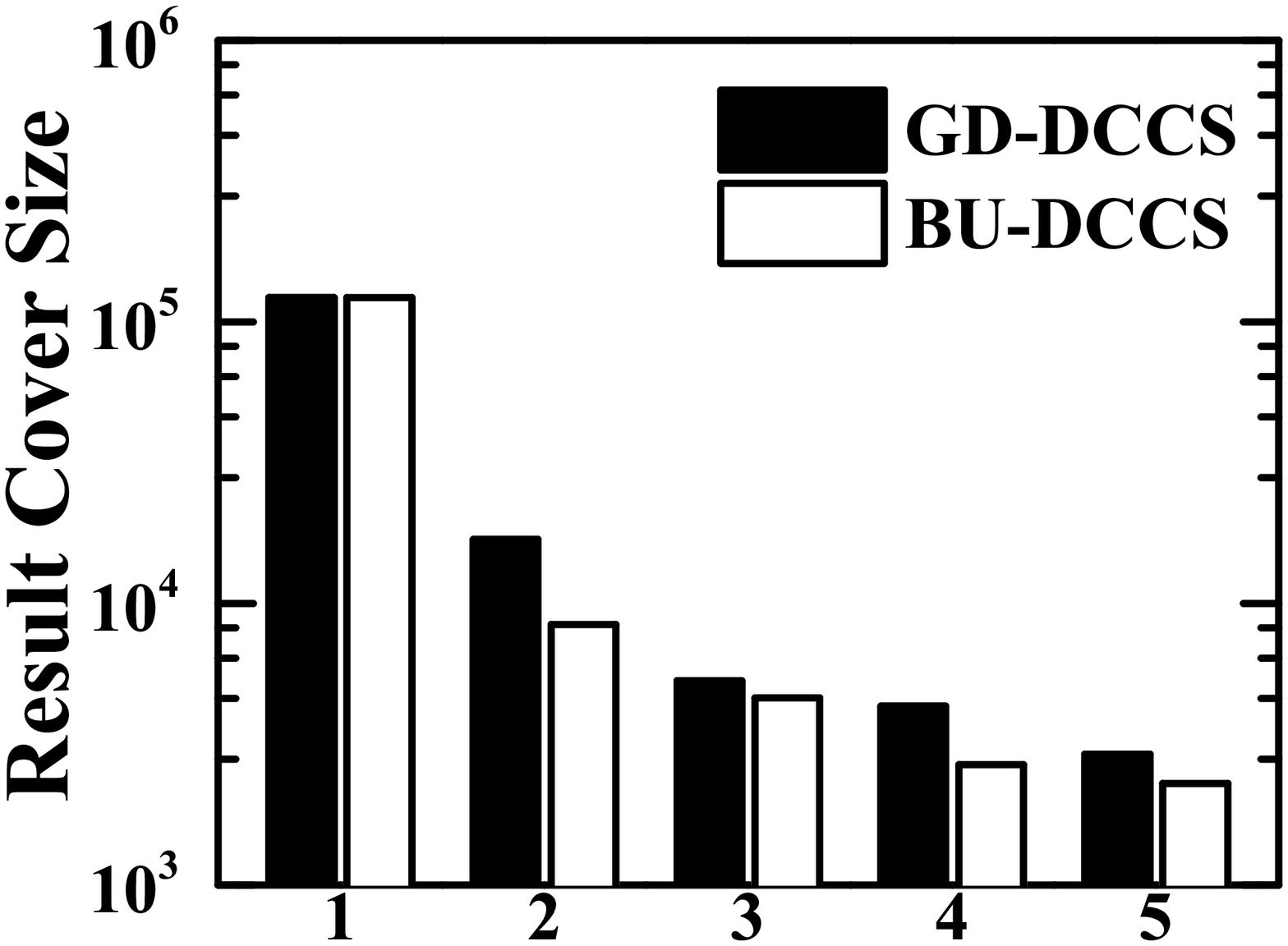}}
        \subfigure[\st{Stack (Vary $s$)}]{\includegraphics[width=0.49\linewidth]{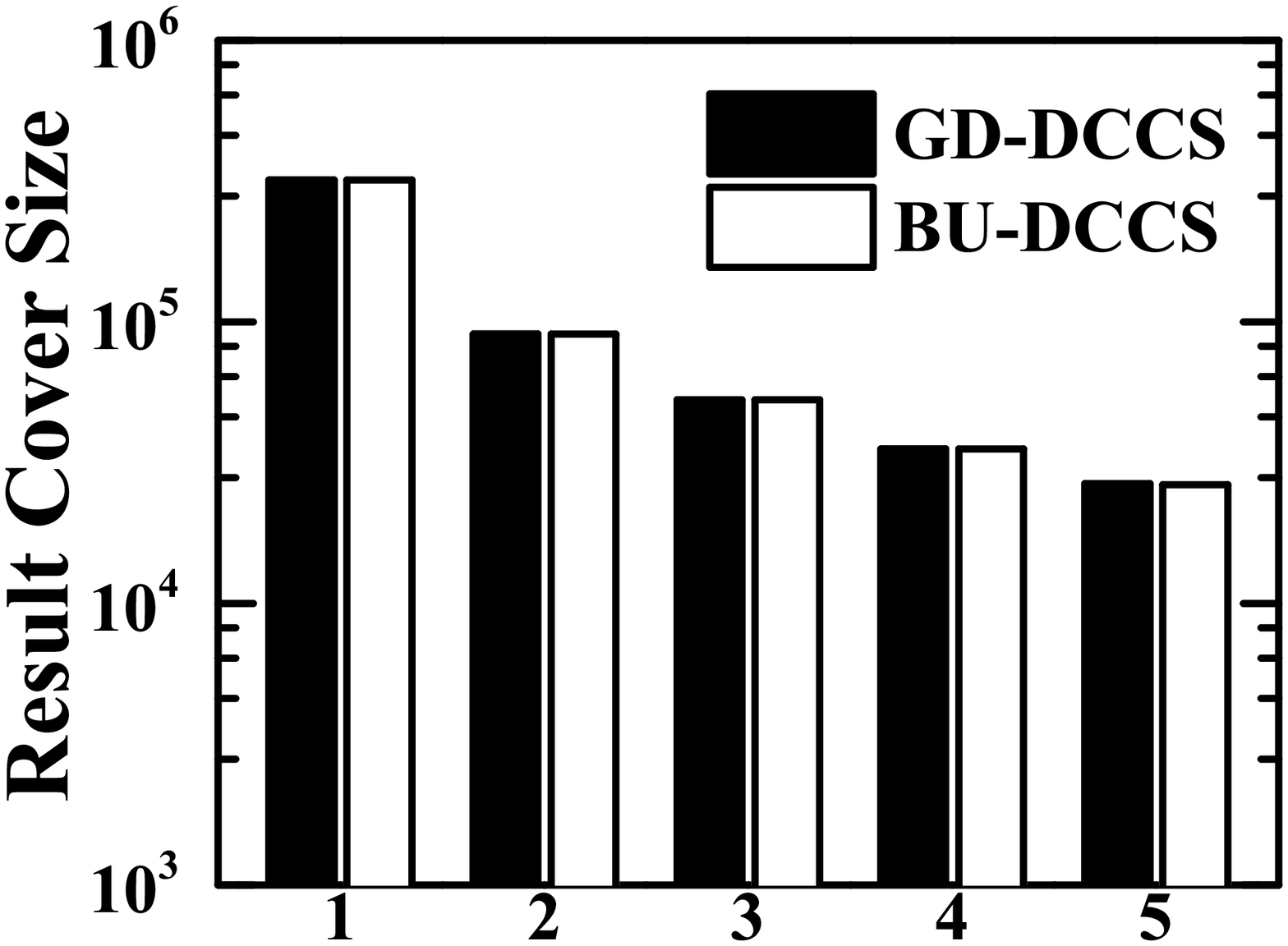}}
        \vspace{-1.5em}
        \caption{Result Cover Size vs Small $s$.}
        \label{Fig: Exp1C}
        \vspace{-1.2em}
    \end{minipage}
\end{figure*}

\begin{figure*}
\addtolength{\subfigcapskip}{-1.2ex}
    \begin{minipage}[!t]{0.33\linewidth}
       % \subfigure[\st{German (Vary $s$)}]{\includegraphics[width=\figswidth]{./ExpFig/Exp1CL3.pdf}}
        \subfigure[\st{English (Vary $s$)}]{\includegraphics[width=0.49\linewidth]{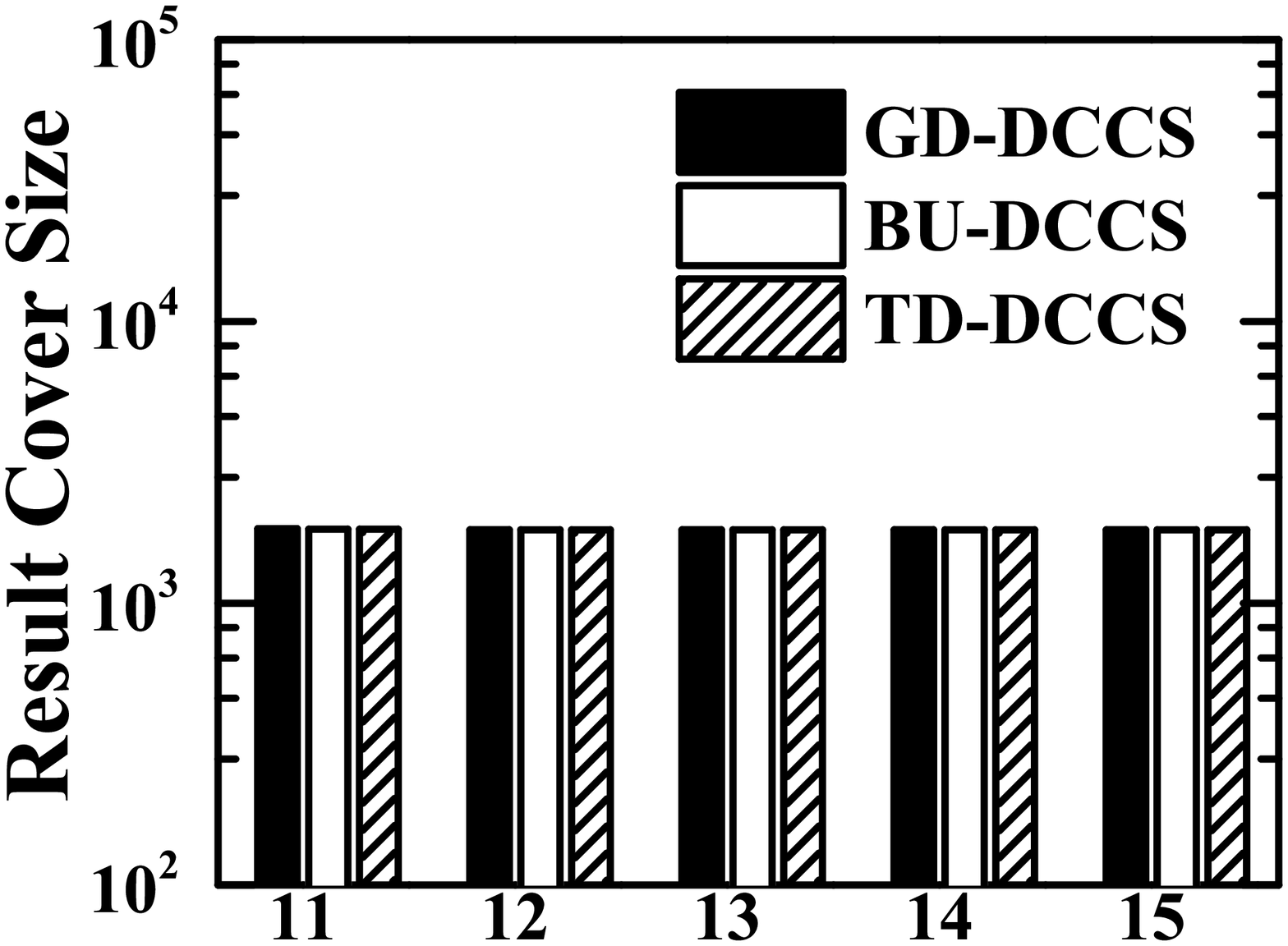}}
        \subfigure[\st{Stack (Vary $s$)}]{\includegraphics[width=0.49\linewidth]{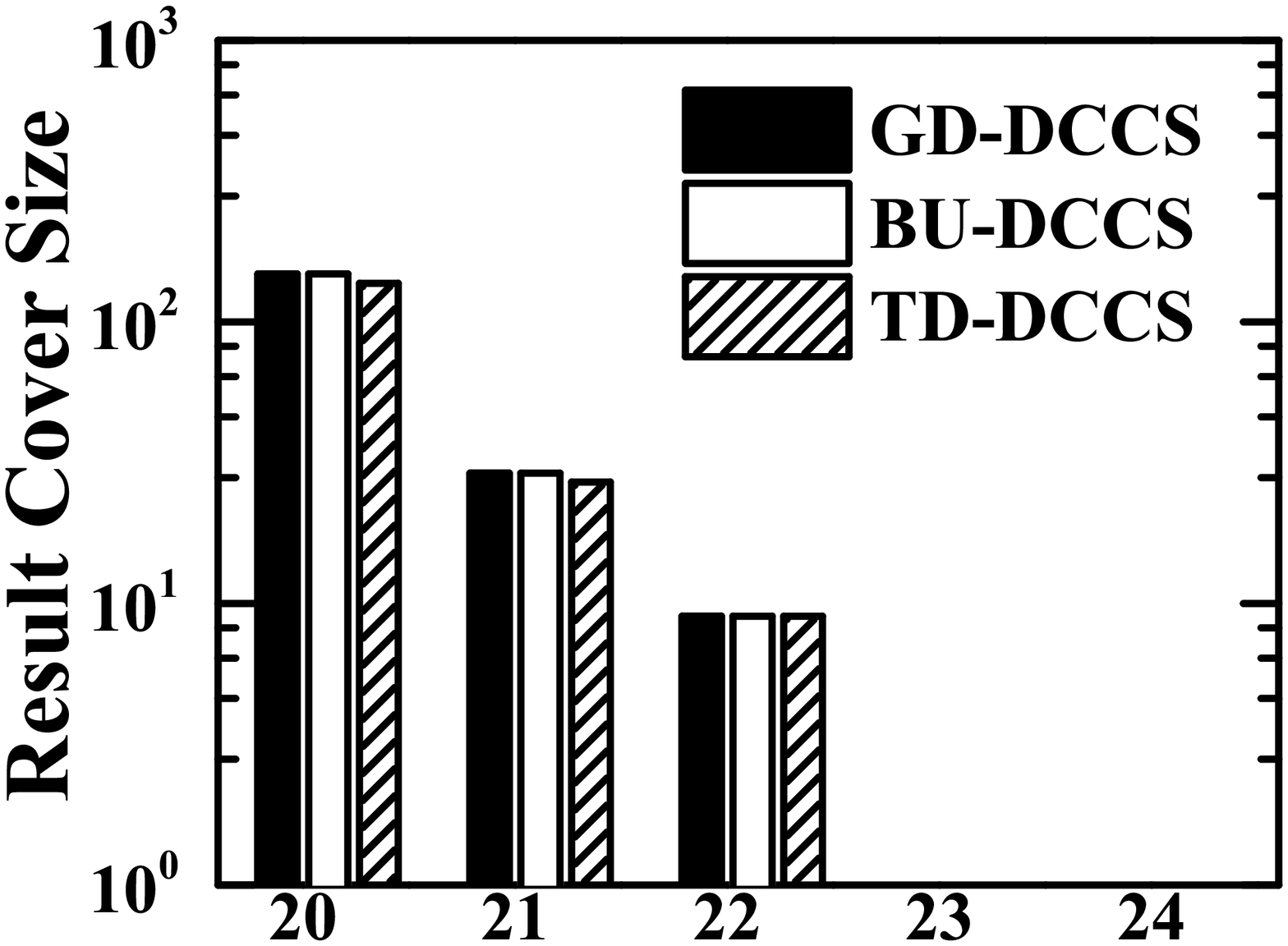}}
        \vspace{-1.5em}
        \caption{Result Cover Size vs Large $s$.}
        \label{Fig: Exp1CL}
        %\vspace{-1em}
    \end{minipage}
        \begin{minipage}[!t]{0.33\linewidth}
        \subfigure[\st{German (Vary $d$)}]{\includegraphics[width=0.49\linewidth]{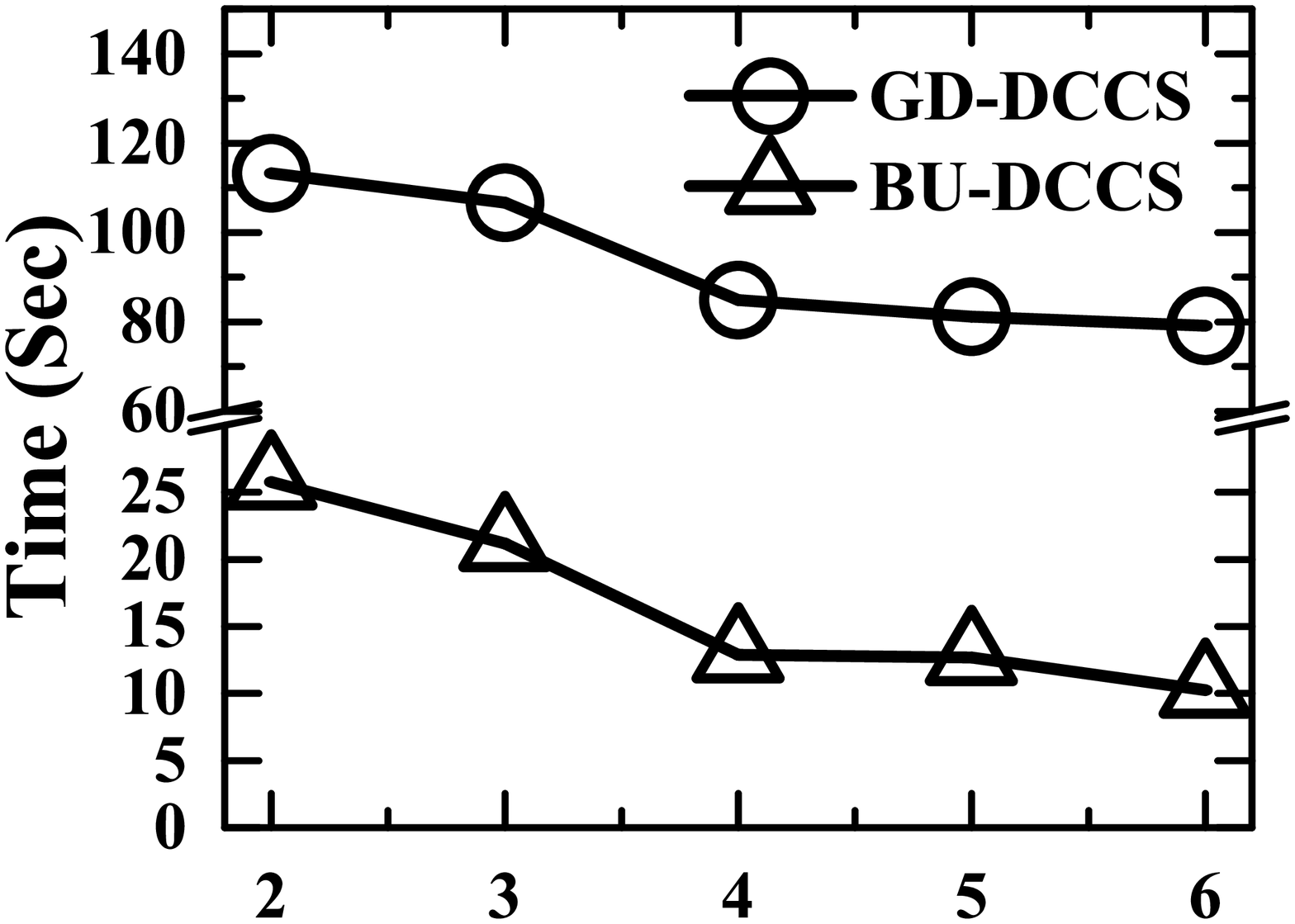}}
        \subfigure[\st{English (Vary $d$)}]{\includegraphics[width=0.49\linewidth]{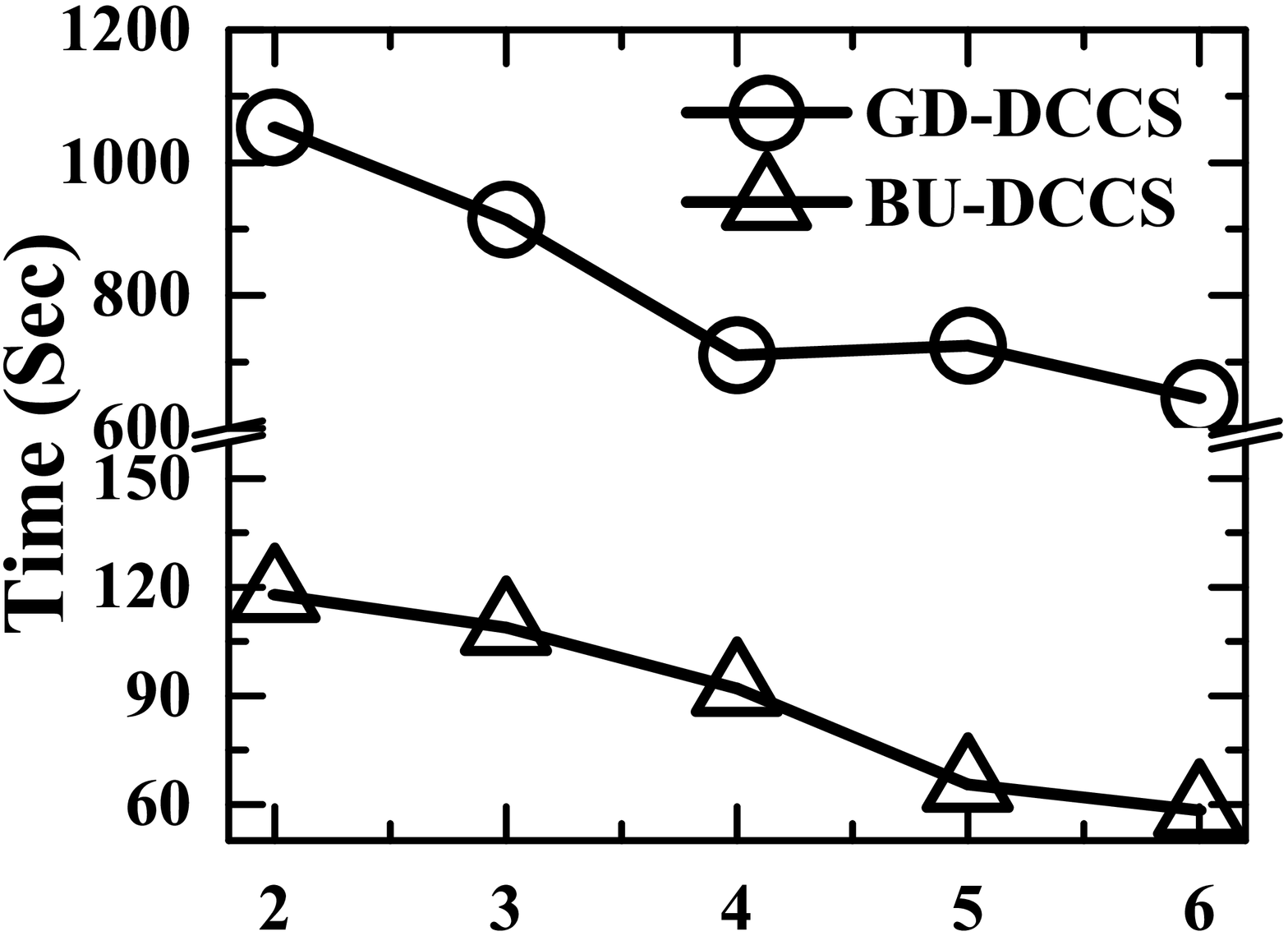}}
        \vspace{-1.5em}
        \caption{Execution Time vs $d$ (Small $s$).}
        \label{Fig: Exp2TB}
    \end{minipage}
    %\vspace{-1em}
    \begin{minipage}[!t]{0.33\linewidth}
        \subfigure[\st{German (Vary $d$)}]{\includegraphics[width=0.49\linewidth]{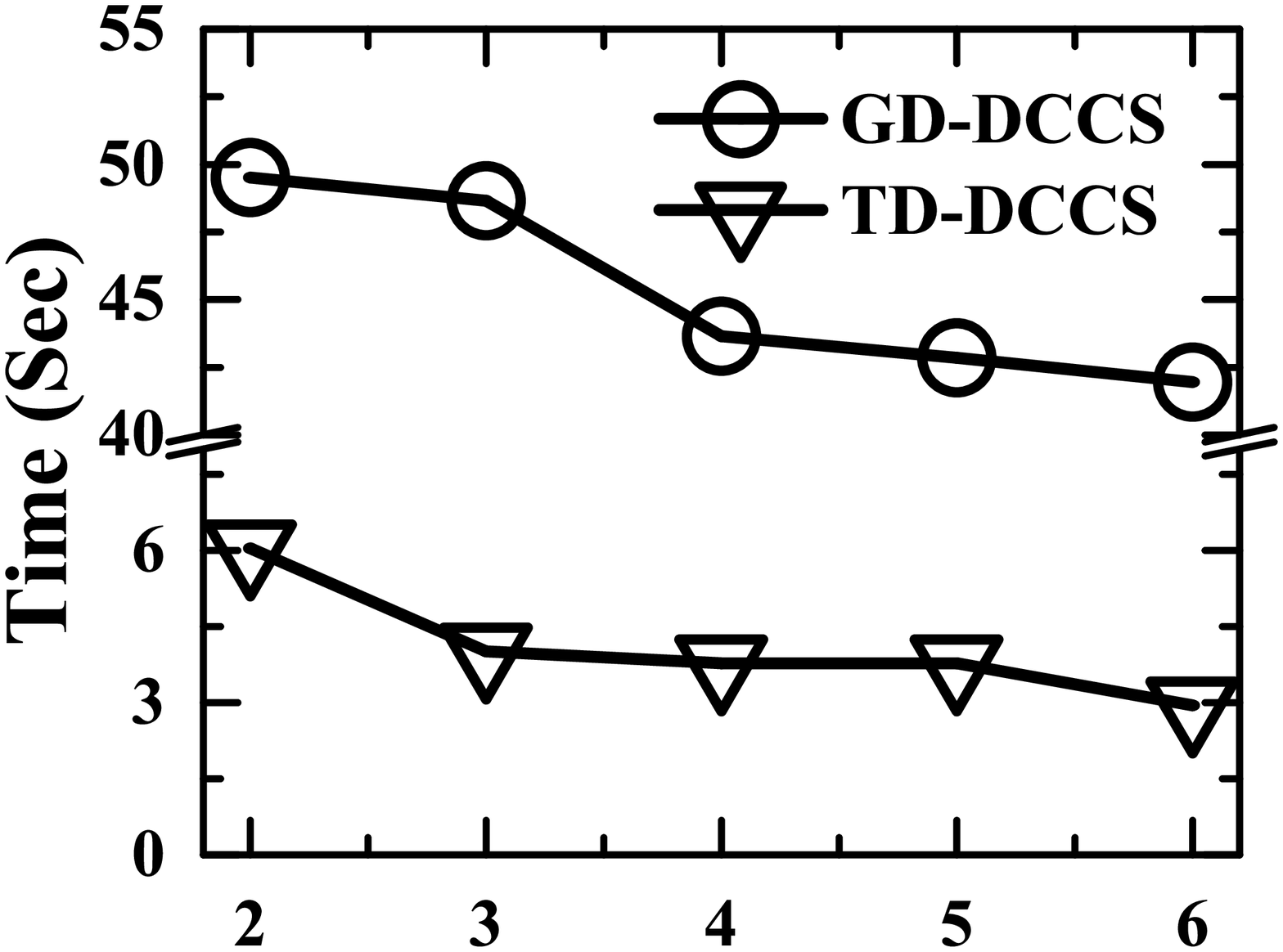}}
        \subfigure[\st{English (Vary $d$)}]{\includegraphics[width=0.49\linewidth]{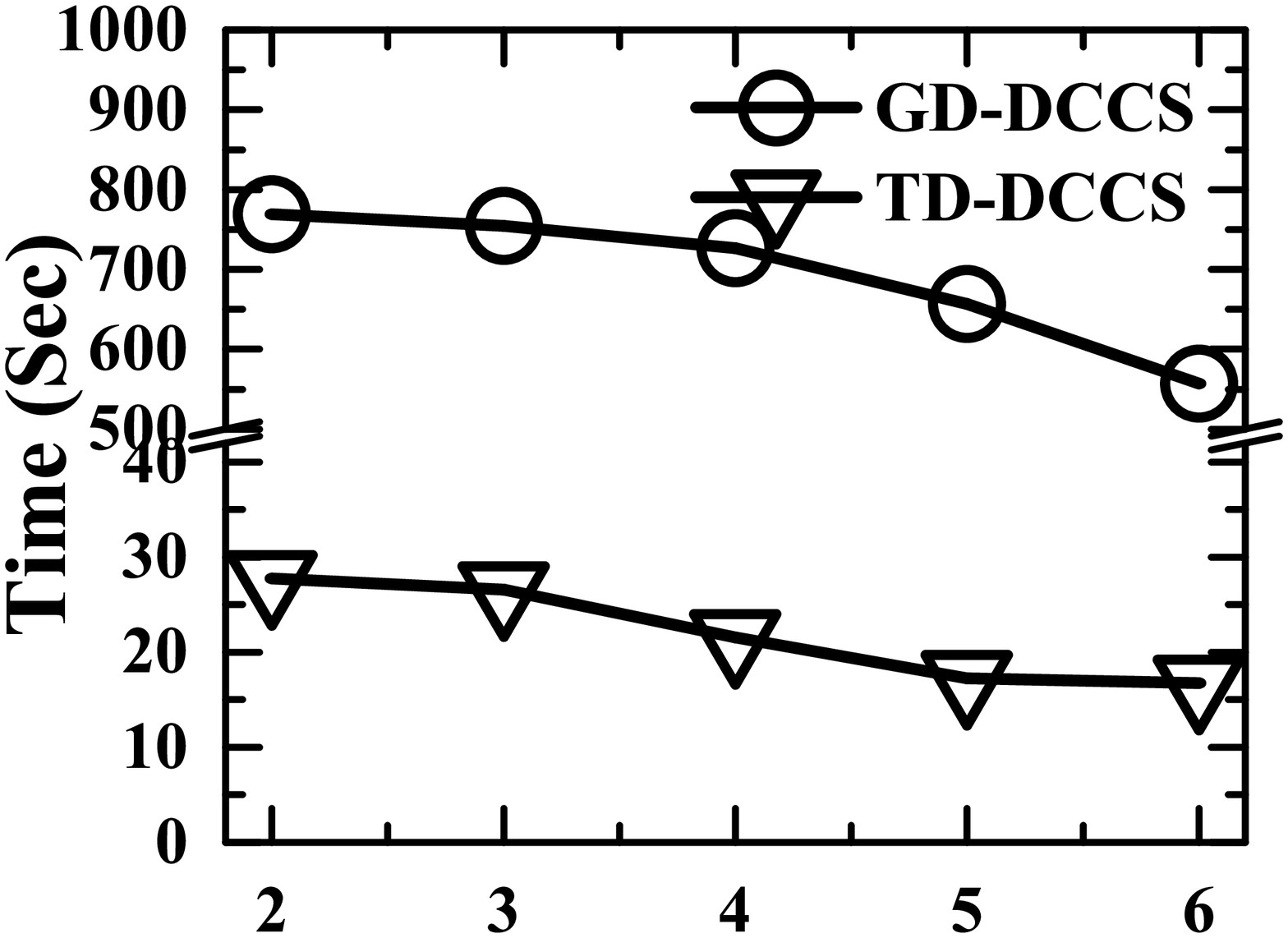}}
        \vspace{-1.5em}
        \caption{Execution Time vs $d$ (Large $s$).}
        \label{Fig: Exp2TT}
    \end{minipage}
    \vspace{-0.5em}
\end{figure*}

This section experimentally evaluates of the proposed algorithms \textsf{GD-DCCS}, \textsf{BU-DCCS} and \textsf{TD-DCCS}. We implemented these algorithms in C++. We did not implement the brute-force exact algorithm mentioned in the beginning of Section~\ref{Sec:GAlgorithm} since it cannot terminate in reasonable time on the graph datasets used in the experiments. For fairness, all the algorithms exploit the preprocessing methods given in Section~\ref{Sec:BApproach-5}. In the experiments, we designate \textsf{GD-DCCS} as the baseline. Every algorithm is evaluated by its execution time (efficiency) and the cover size $|\Cov(\R)|$ of the result $\R$ (accuracy). All the experiments were run on a machine installed with an Intel Core i5-2400 CPU (3.1GHz and 4 cores) and 22GB of RAM, running 64-bit Ubuntu 14.04.

%\smallskip

\noindent{\underline{\bf Datasets.}}
We use 6 real-world graph datasets of various types and sizes in the experiments. The statistics of the graph datasets are summarized in Fig.~\ref{Tab: Datasets}. \textit{PPI} is a protein-protein interaction network extracted from the STRING database (\texttt{\small{http://string-db.org}}). It contains 8 layers representing the interactions between proteins detected by different methods. \textit{Author} is a co-authorship network obtained from AMiner (\texttt{\small{http://cn.aminer.org}}). It contains 10 layers representing the collaboration between authors in 10 different years. \textit{PPI} and \textit{Author} are very small datasets. They are used in the comparisons between the notions of $d$-CC and quasi-clique. The other datasets were obtained from KONECT (\texttt{\small{http://konect.uni-koblenz.de}}) and SNAP (\texttt{\small{http://snap.stanford.edu}}), where each layer contains the connections generated in a specific time period. Specifically, in \textit{German} and \textit{English}, each layer consists of the interactions between users in a year; in \textit{Wiki} and \textit{Stack}, each layer contains the connections generated in an hour.

\noindent{\underline{\bf Parameters.}}
We set 5 parameters in the experiments, namely $k$, $d$ and $s$ in the DCCS problem and $p, q \in [0, 1]$. Parameters $p$ and $q$ are varied in the scalability test of the algorithms. Specifically, $p$ and $q$ controls the proportion of vertices and layers extracted from the graphs, respectively. The ranges and the default values of the parameters are shown in Fig.~\ref{Tab: Parameters}. We adopt two configurations for parameter $s$. When testing for small $s$, we select $s$ from $\{1,2,3,4,5\}$; when testing for large $s$, we  select $s$ from $\{l(\G)-4, l(\G)-3, l(\G)-2, l(\G)-1, l(\G)\}$. Without otherwise stated, when varying a parameter, other parameters are set to their default values.

\noindent{\underline{\bf Execution Time w.r.t.~Parameter \textit{s}.}}
We evaluate the execution time of the algorithms w.r.t.~$s$. First, we experiment for small $s$. Since the \textsf{TD-DCCS} algorithm is not applicable when $s < l(\G)/2$, we only test the other three algorithms for small $s$. Fig.~\ref{Fig: Exp1T} shows the execution time of the algorithms on the datasets \textit{English} and \textit{Stack}. We have two observations: 1) The execution time of all the algorithms substantially increases with $s$. This is simply because the search space of the \textsc{DCCS} problem fast grows with $s$ when $s < l(\G)/2$. 2) The \textsf{BU-DCCS} algorithm outperforms \textsf{GD-DCCS} by $1$--$2$ orders of magnitude. For example, when $s = 4$, \textsf{BU-DCCS} is 39X and 30X faster than \textsf{GD-DCCS} on \textit{English} and \textit{Stack}, respectively. The main reason is that the pruning techniques adopted by \textsf{BU-DCCS} reduce the search space of the DCCS by 80\%--90\%.

We also examine the algorithms for large $s$ and show results in Fig.~\ref{Fig: Exp1TL}. At this time, we also test the \textsf{TD-DCCS} algorithm. We have the the following observations: 1) The execution time of all the algorithms decreases when $s$ grows. This is because the search space of the \textsc{DCCS} problem decreases with $s$ when $s \geq l(\G)/2$. 2) \textsf{BU-DCCS} is not efficient for large $s$. Sometimes, it is even worse than \textsf{GD-DCCS}. When $s$ is large, the sizes of the $d$-CCs significantly decreases. \textsf{BU-DCCS} has to search down deep the search tree until the pruning techniques start to take effects. In some cases, \textsf{BU-DCCS} searches even more $d$-CCs than \textsf{GD-DCCS}. 3) \textsf{TD-DCCS} runs much faster than all the others. For example, when $s = 13$, \textsf{TD-DCCS} is 50X faster than \textsf{GD-DCCS} on \textit{English}. This is because $d$-CCs are generated in a top-down manner in \textsf{TD-DCCS}, so the number of $d$-CCs searched by \textsf{TD-DCCS} must be less than \textsf{BU-DCCS}. Moreover, many unpromising candidates $d$-CCs are pruned earlier in \textsf{TD-DCCS}.

%\smallskip

\noindent{\underline{\bf Cover Size of Result w.r.t.~Parameter \textit{s}.}}
We evaluate the cover size $|\Cov(\R)|$ of result $\R$ w.r.t.~parameter $s$. Fig.~\ref{Fig: Exp1C} and Fig.~\ref{Fig: Exp1CL} show the experimental results for small $s$ and large $s$, respectively. We have two observations: 1) For all the algorithms, $|\Cov(\R)|$ decreases with $s$. This is because while $s$ increases, the size of $d$-CCs never increases due to Property~\ref{Lem:PHierarchy}, so $\R$ cannot cover more vertices. 2) In most cases, the results of the algorithms cover similar amount of vertices for either small $s$ or large $s$. Sometimes, the result of \textsf{GD-DCCS} covers slightly more vertices than the results of \textsf{BU-DCCS} and \textsf{TD-DCCS}.  This is because \textsf{GD-DCCS} is $(1 - 1/e)$-approximate; while \textsf{BU-DCCS} and \textsf{TD-DCCS} are $1/4$-approximate. It verifies that the practical approximation quality of \textsf{BU-DCCS} and \textsf{TD-DCCS} is close to \textsf{GD-DCCS}.

\begin{figure*}[!t]
\addtolength{\subfigcapskip}{-1.2ex}
    \begin{minipage}[!t]{0.33\linewidth}
        \subfigure[\st{German (Vary $d$)}]{\includegraphics[width=0.49\linewidth]{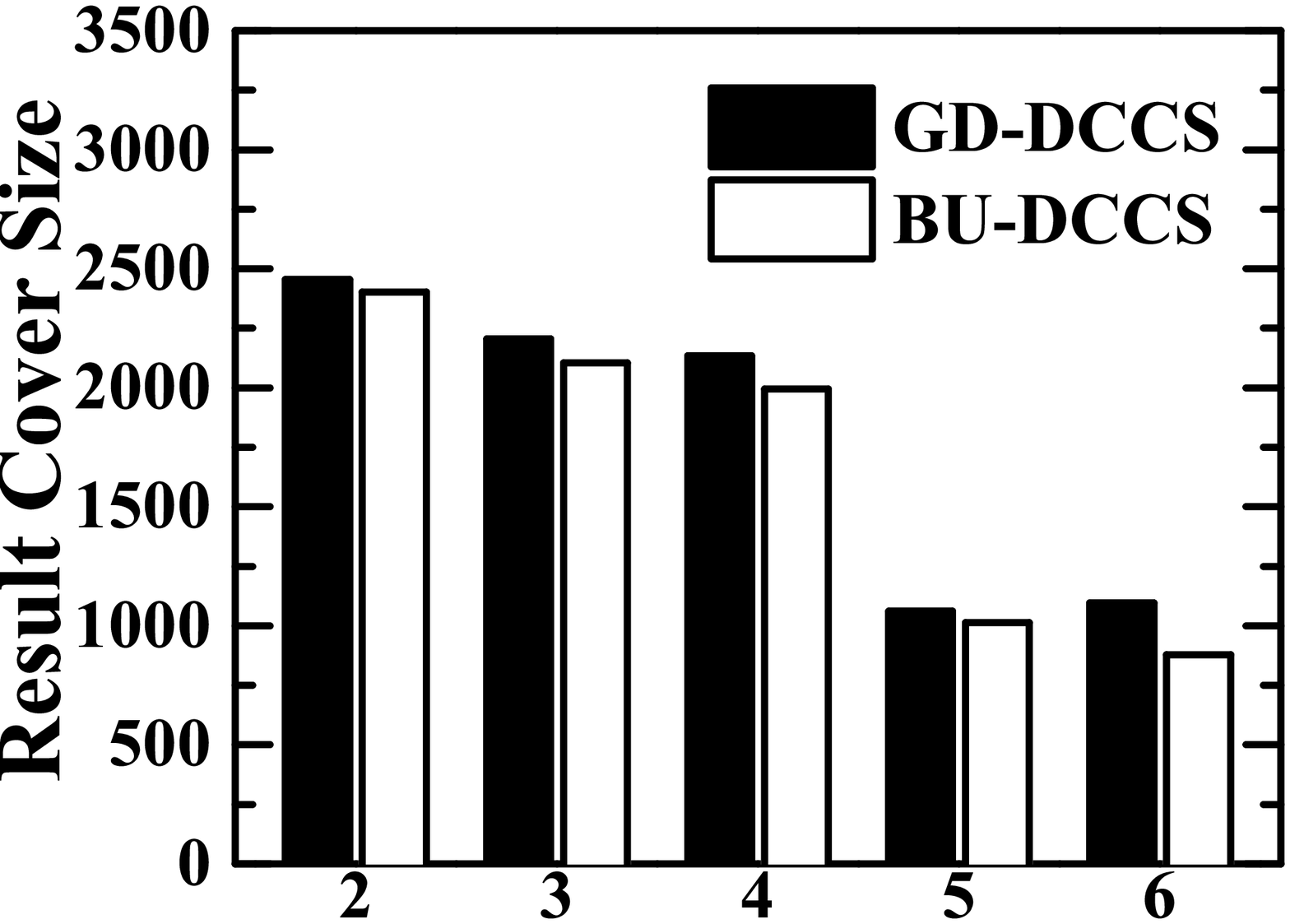}}
        \subfigure[\st{English (Vary $d$)}]{\includegraphics[width=0.49\linewidth]{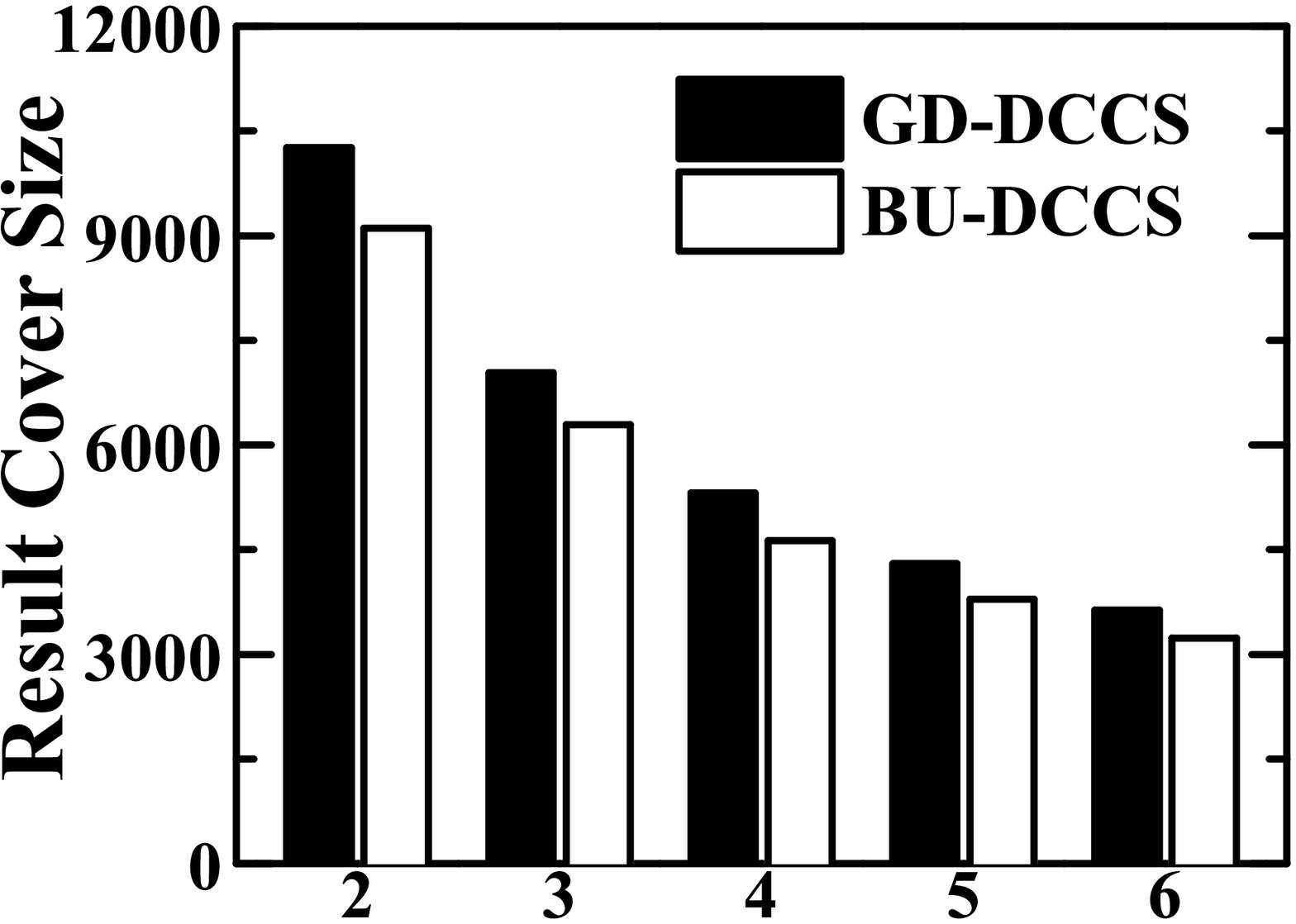}}
        \vspace{-1.5em}
        \caption{Result Cover Size vs~$d$ (Small $s$).}
        \label{Fig: Exp2CB}
    \end{minipage}
    \vspace{-1.2em}
        \begin{minipage}[!t]{0.33\linewidth}
        \subfigure[\st{German (Vary $d$)}]{\includegraphics[width=0.49\linewidth]{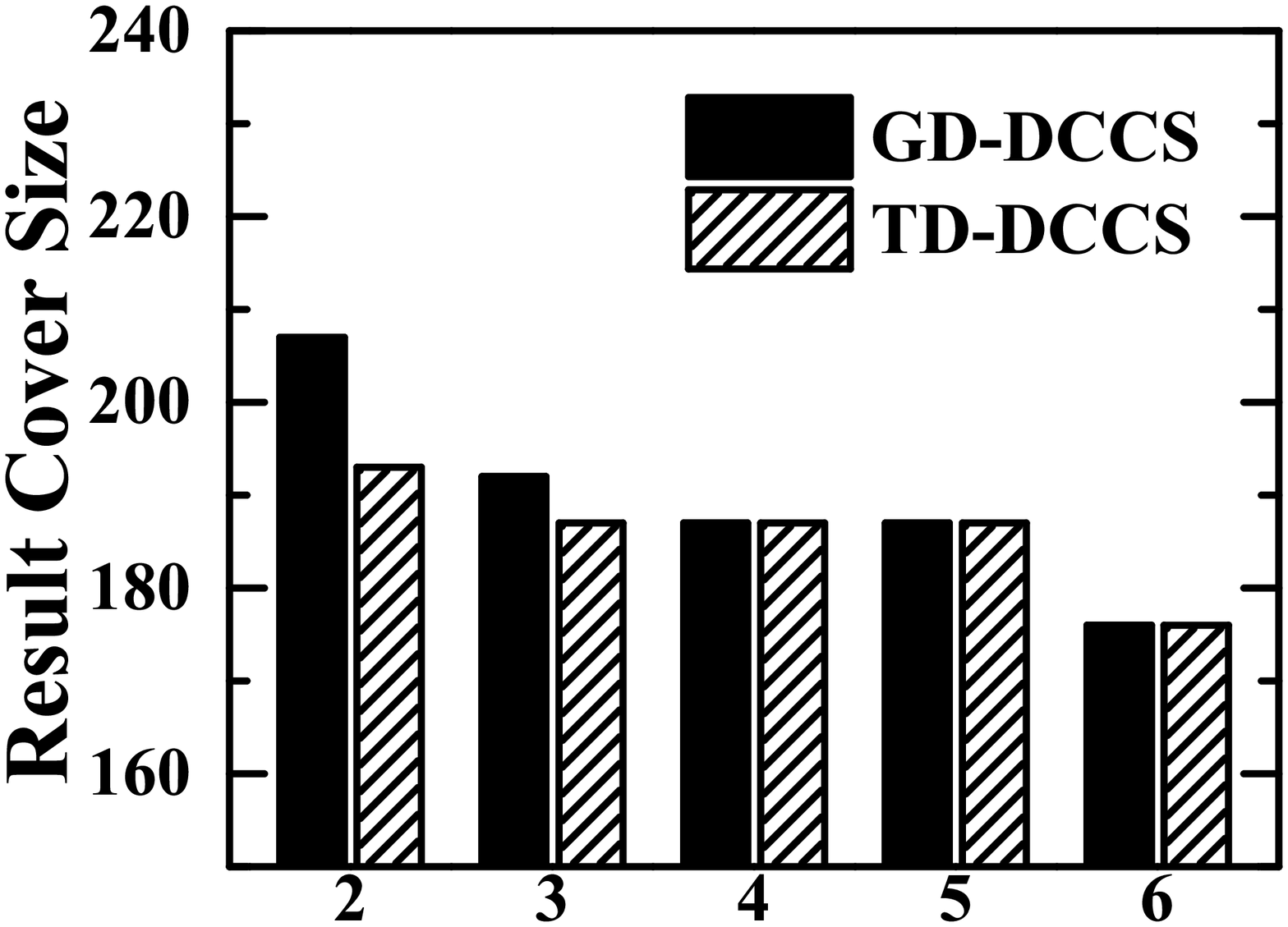}}
        \subfigure[\st{English (Vary $d$)}]{\includegraphics[width=0.49\linewidth]{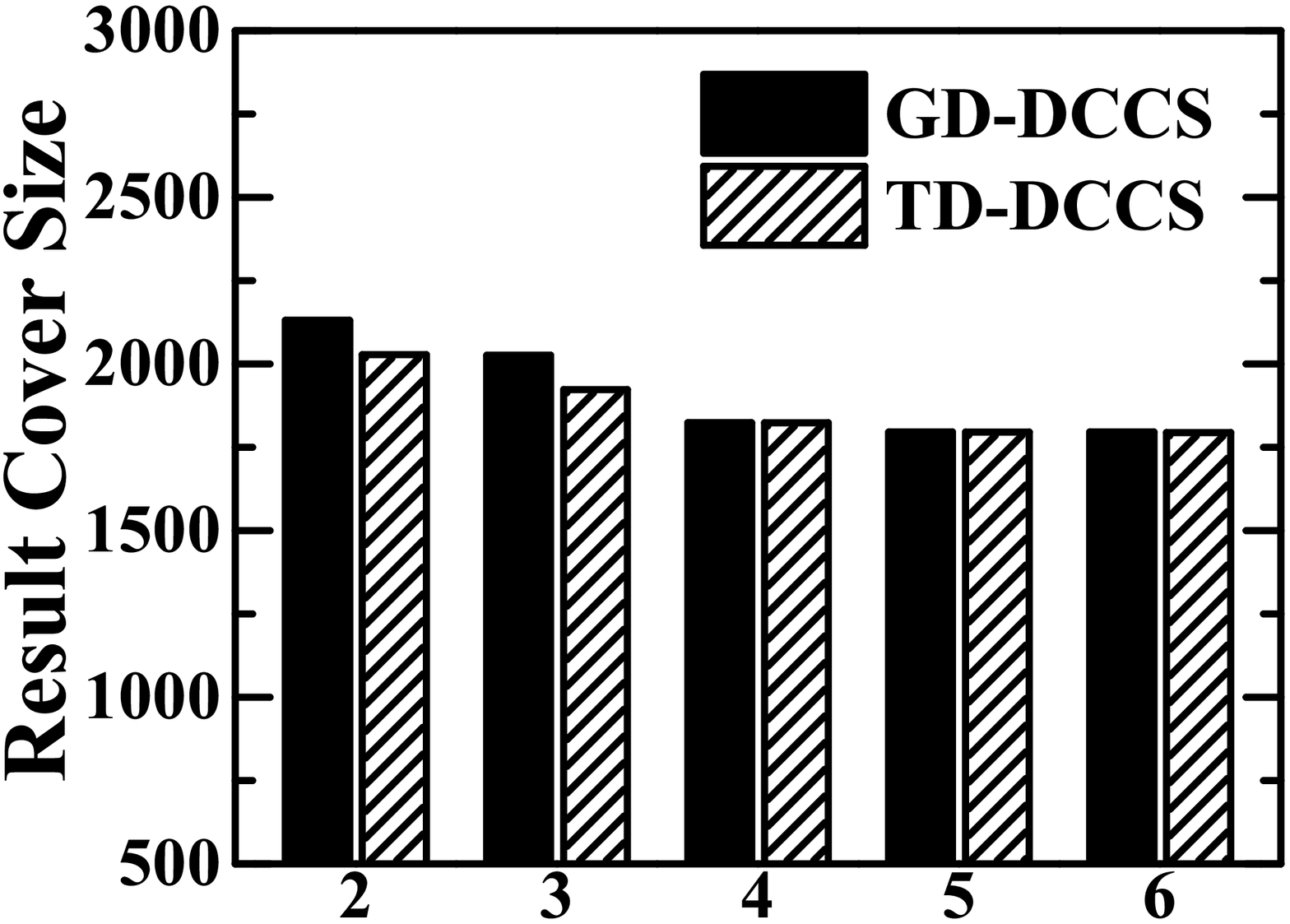}}
        \vspace{-1.5em}
        \caption{Result Cover Size vs~$d$ (Large $s$).}
        \label{Fig: Exp2CT}
    \end{minipage}
    \begin{minipage}[!t]{0.33\linewidth}
        \subfigure[\st{Wiki (Vary $k$)}]{\includegraphics[width=0.49\linewidth]{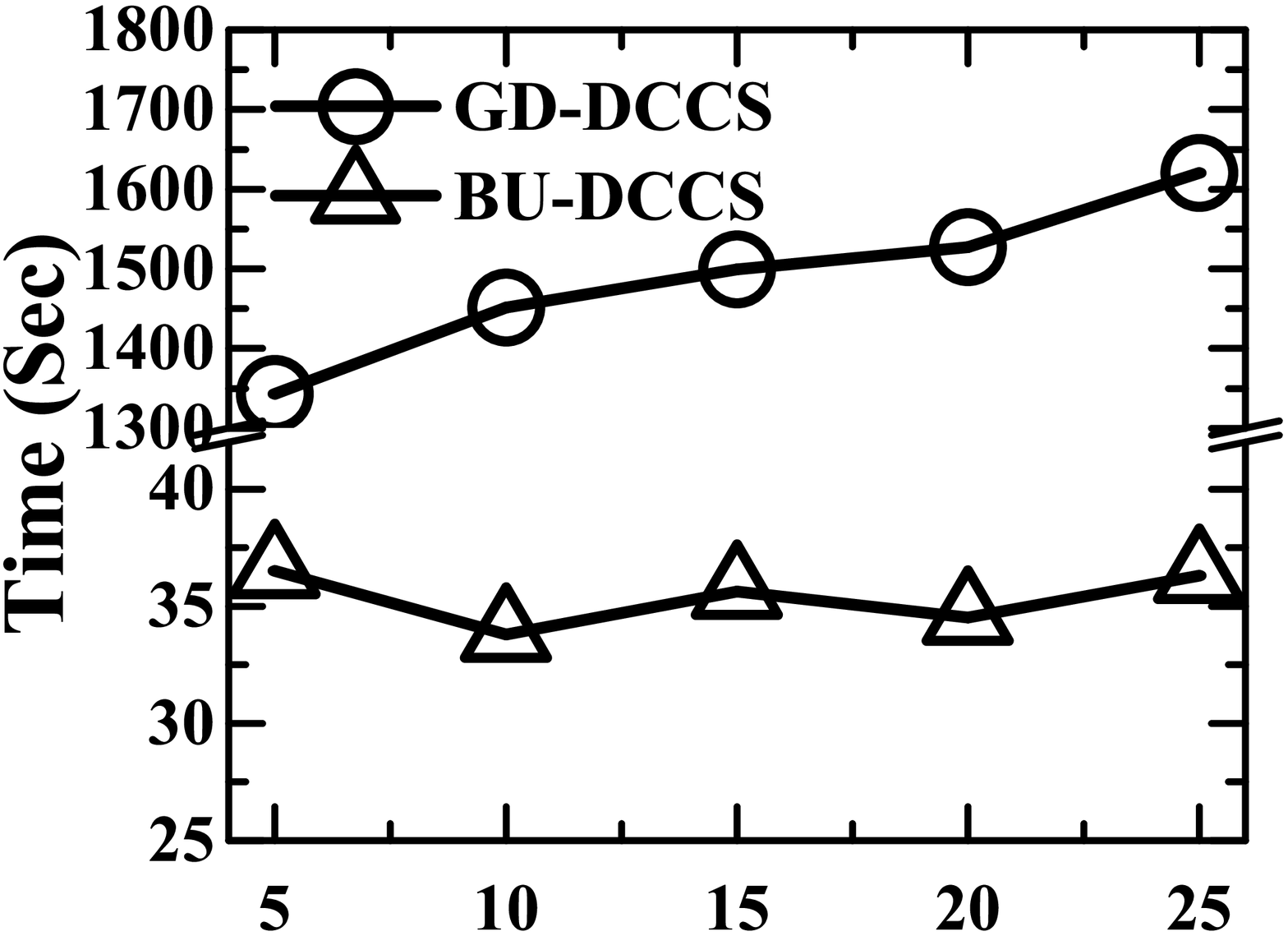}}
        \subfigure[\st{English (Vary $k$)}]{\includegraphics[width=0.49\linewidth]{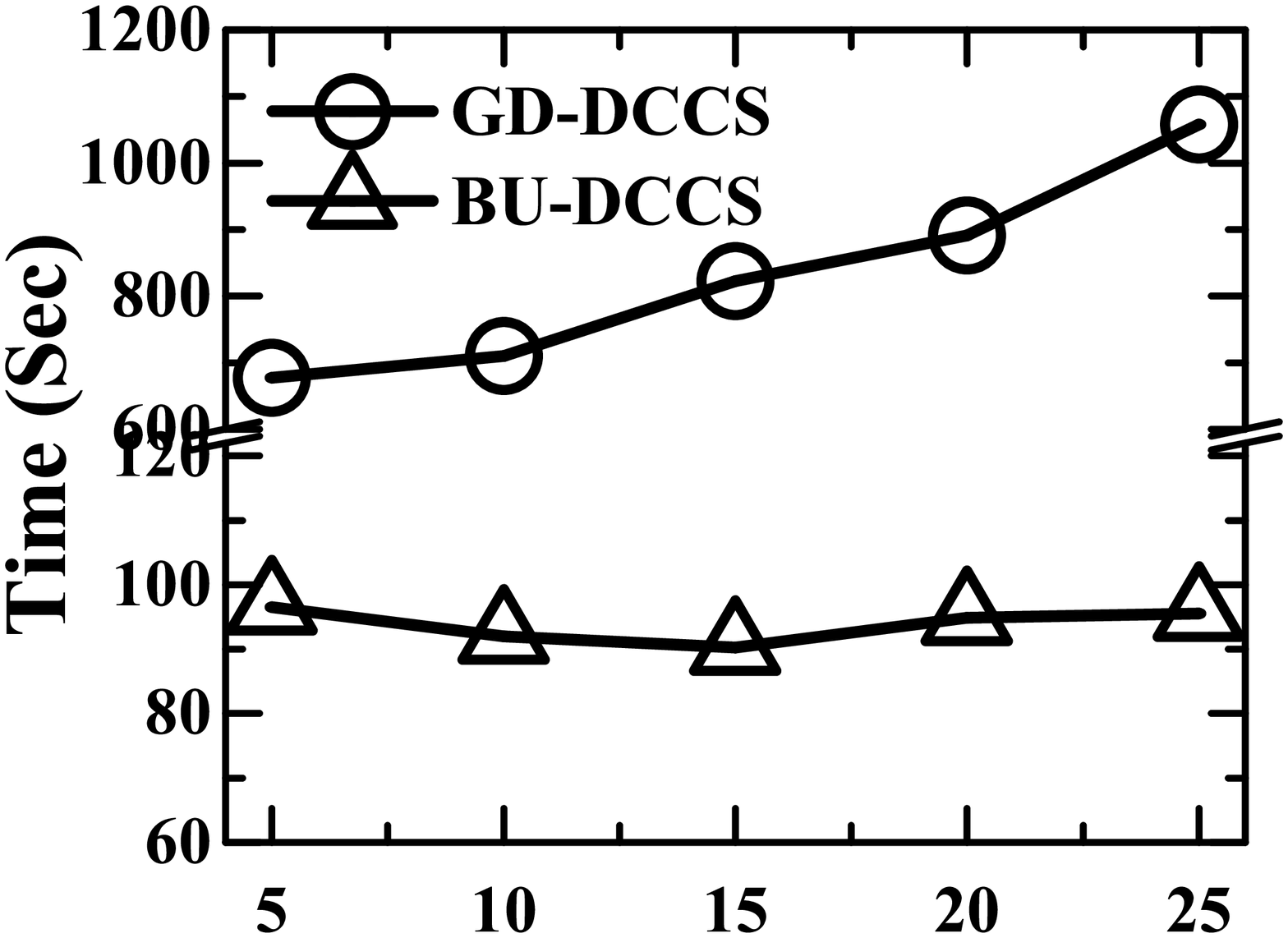}}
        \vspace{-1.5em}
        \caption{Execution Time vs $k$ (Small $s$).}
        \label{Fig: Exp3TB}
    \end{minipage}
\end{figure*}

%\smallskip
\begin{figure*}[!t]
\addtolength{\subfigcapskip}{-1.2ex}
    \begin{minipage}[!t]{0.33\linewidth}
        \subfigure[\st{Wiki (Vary $k$)}]{\includegraphics[width=0.49\linewidth]{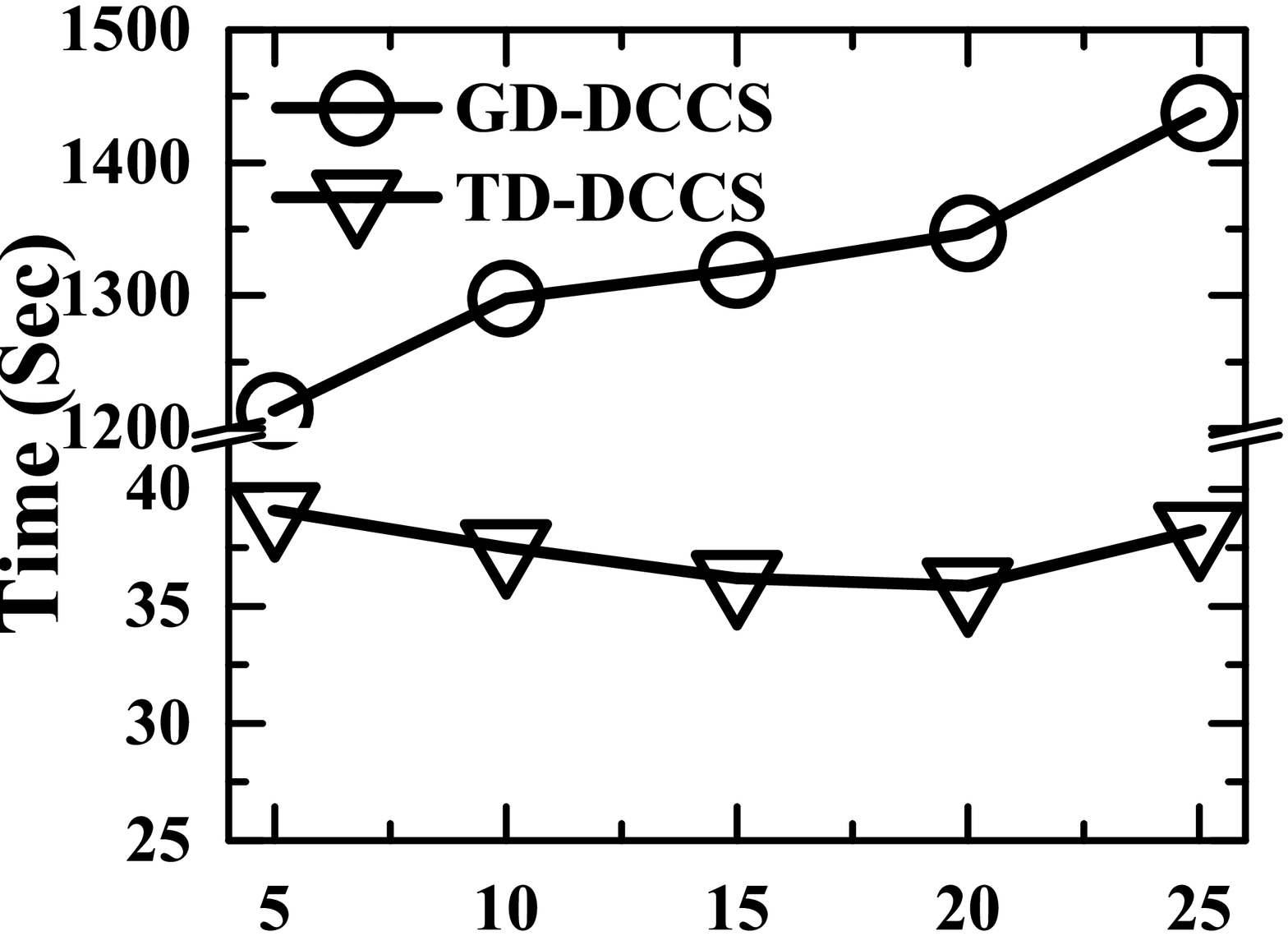}}
        \subfigure[\st{English (Vary $k$)}]{\includegraphics[width=0.49\linewidth]{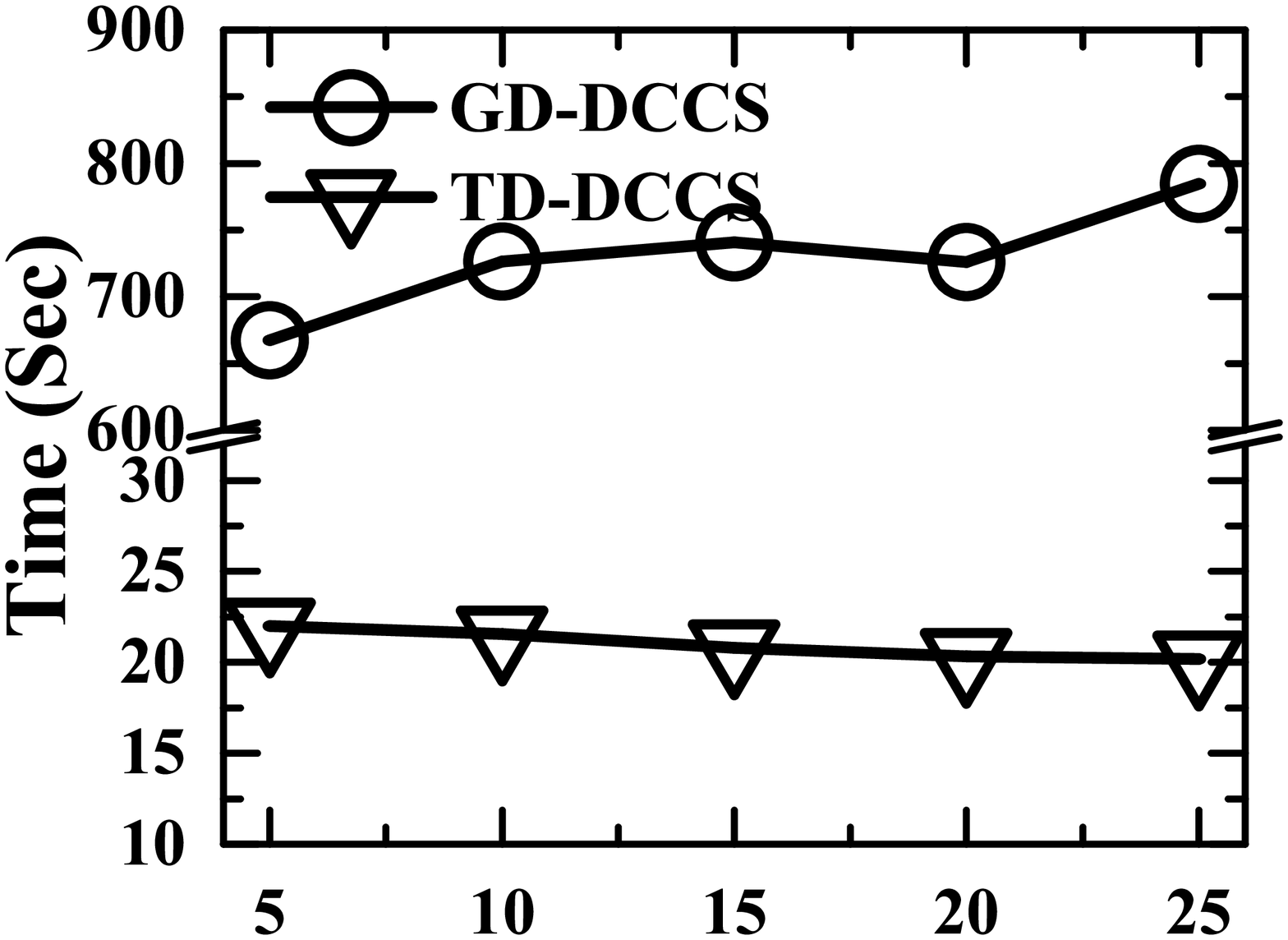}}
        \vspace{-1.5em}
        \caption{Execution Time vs $k$ (Large $s$).}
        \label{Fig: Exp3TT}
    \end{minipage}
    \vspace{-1.2em}
        \begin{minipage}[!t]{0.33\linewidth}
        \subfigure[\st{Wiki (Vary $k$)}]{\includegraphics[width=0.49\linewidth]{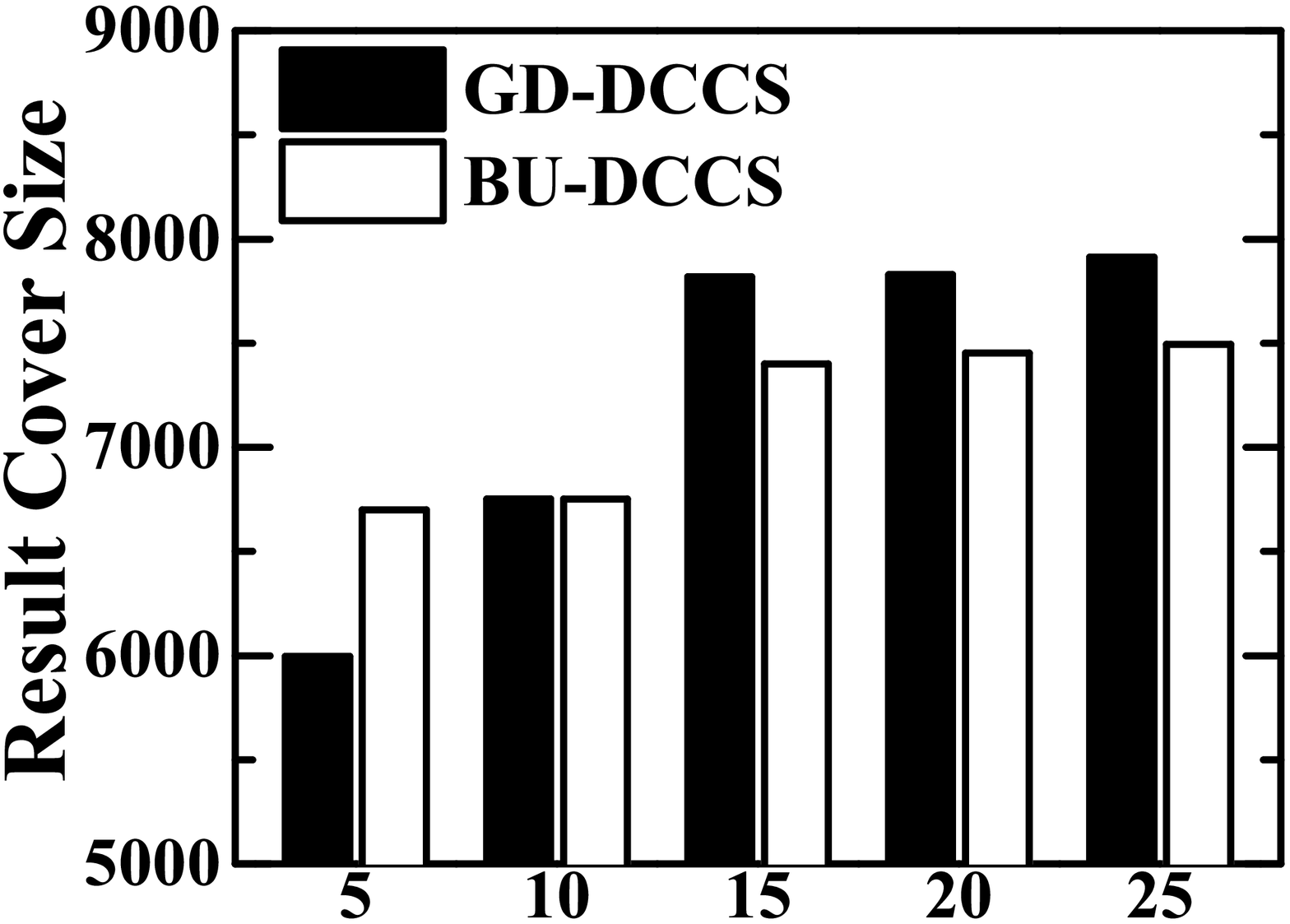}}
        \subfigure[\st{English (Vary $k$)}]{\includegraphics[width=0.49\linewidth]{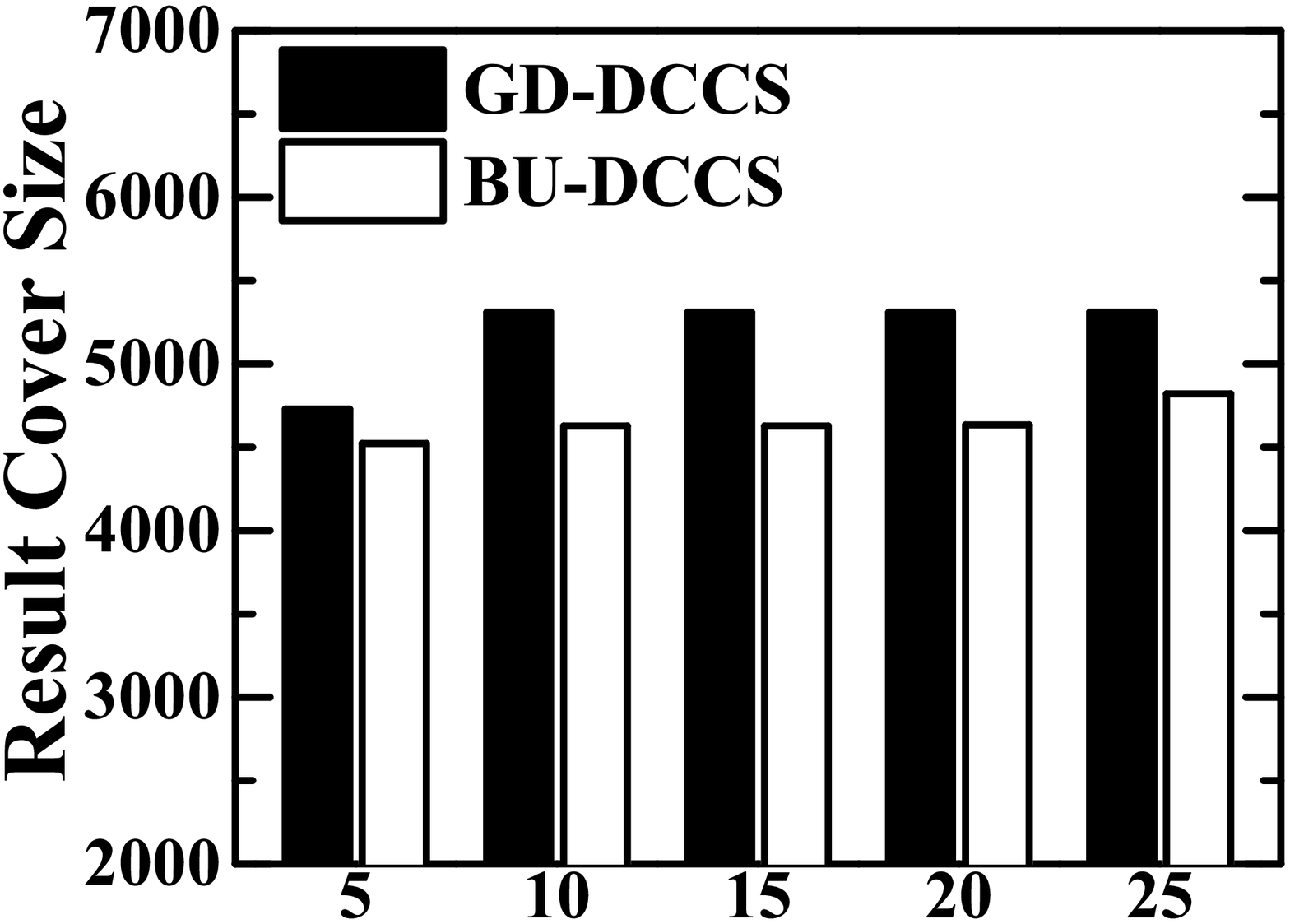}}
        \vspace{-1.5em}
        \caption{Result Cover Size vs $k$ (Small $s$).}
        \label{Fig: Exp3CB}
    \end{minipage}
    \begin{minipage}[!t]{0.33\linewidth}
        \subfigure[\st{Wiki (Vary $k$)}]{\includegraphics[width=0.49\linewidth]{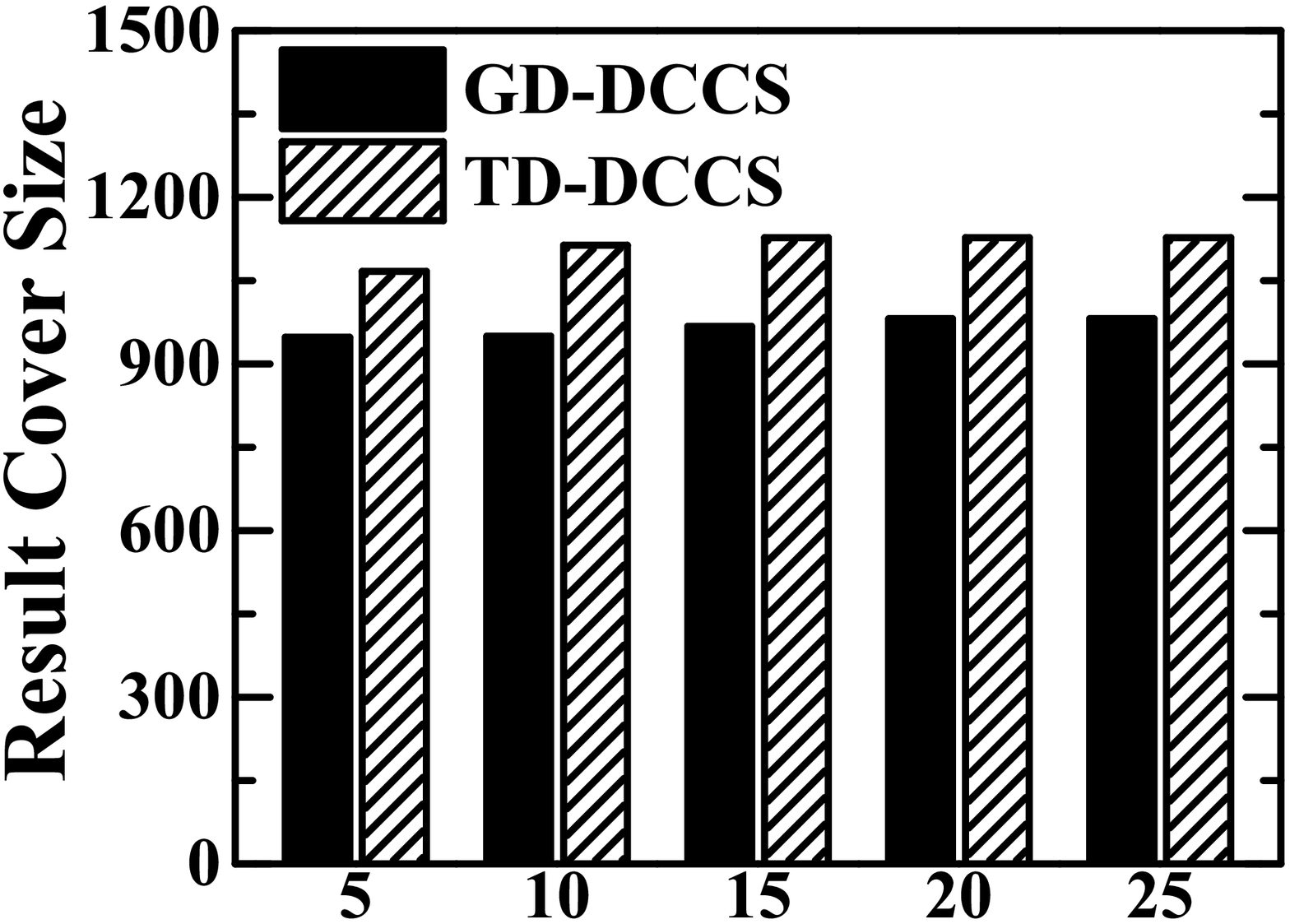}}
        \subfigure[\st{English (Vary $k$)}]{\includegraphics[width=0.49\linewidth]{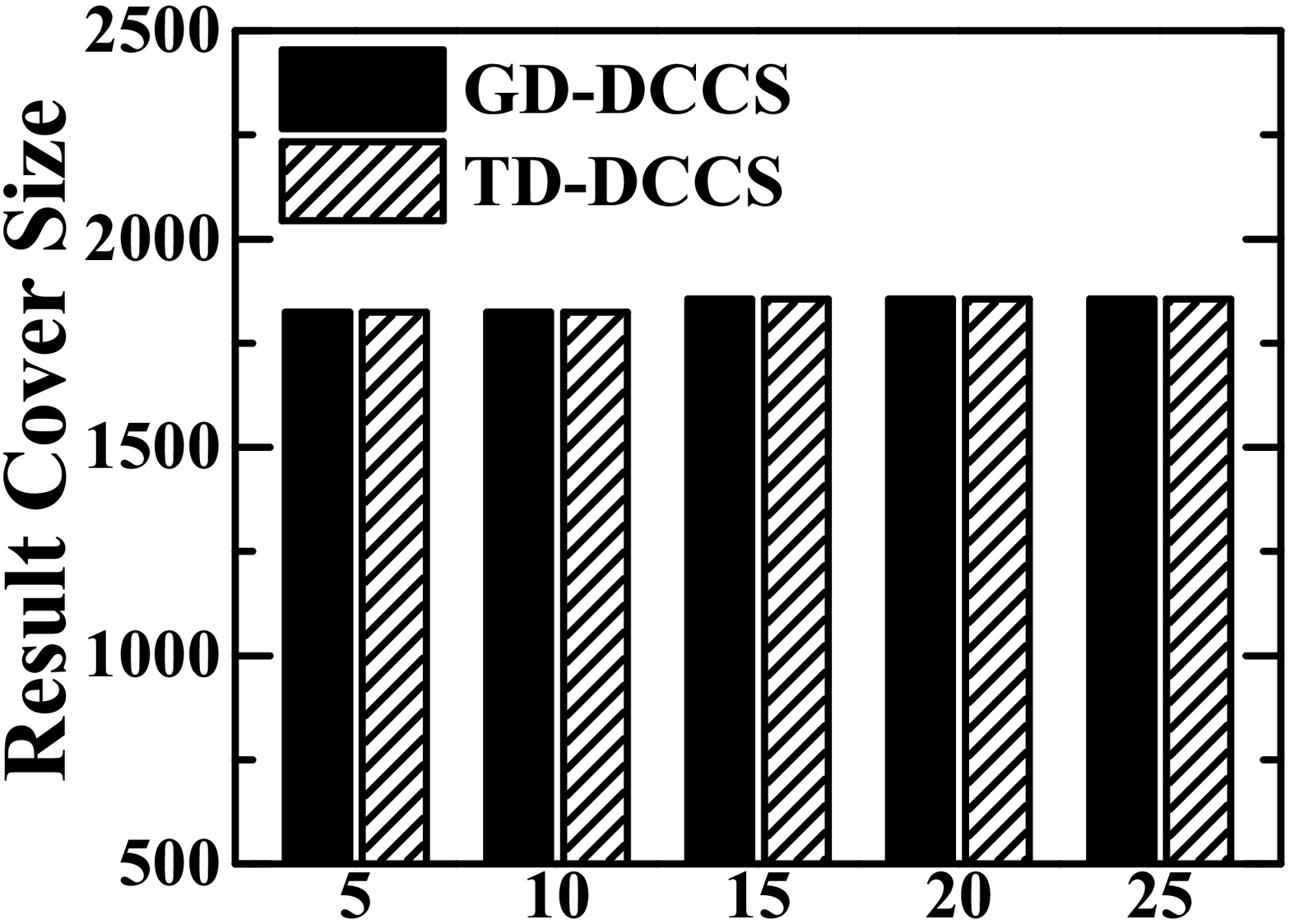}}
        \vspace{-1.5em}
        \caption{Result Cover Size vs $k$ (Large $s$).}
        \label{Fig: Exp3CT}
    \end{minipage}
\end{figure*}

\begin{figure*}[!t]
\addtolength{\subfigcapskip}{-1.2ex}
    \begin{minipage}[!t]{0.33\linewidth}
        \subfigure[\st{Stack (Vary $p$)}]{\includegraphics[width=0.49\linewidth]{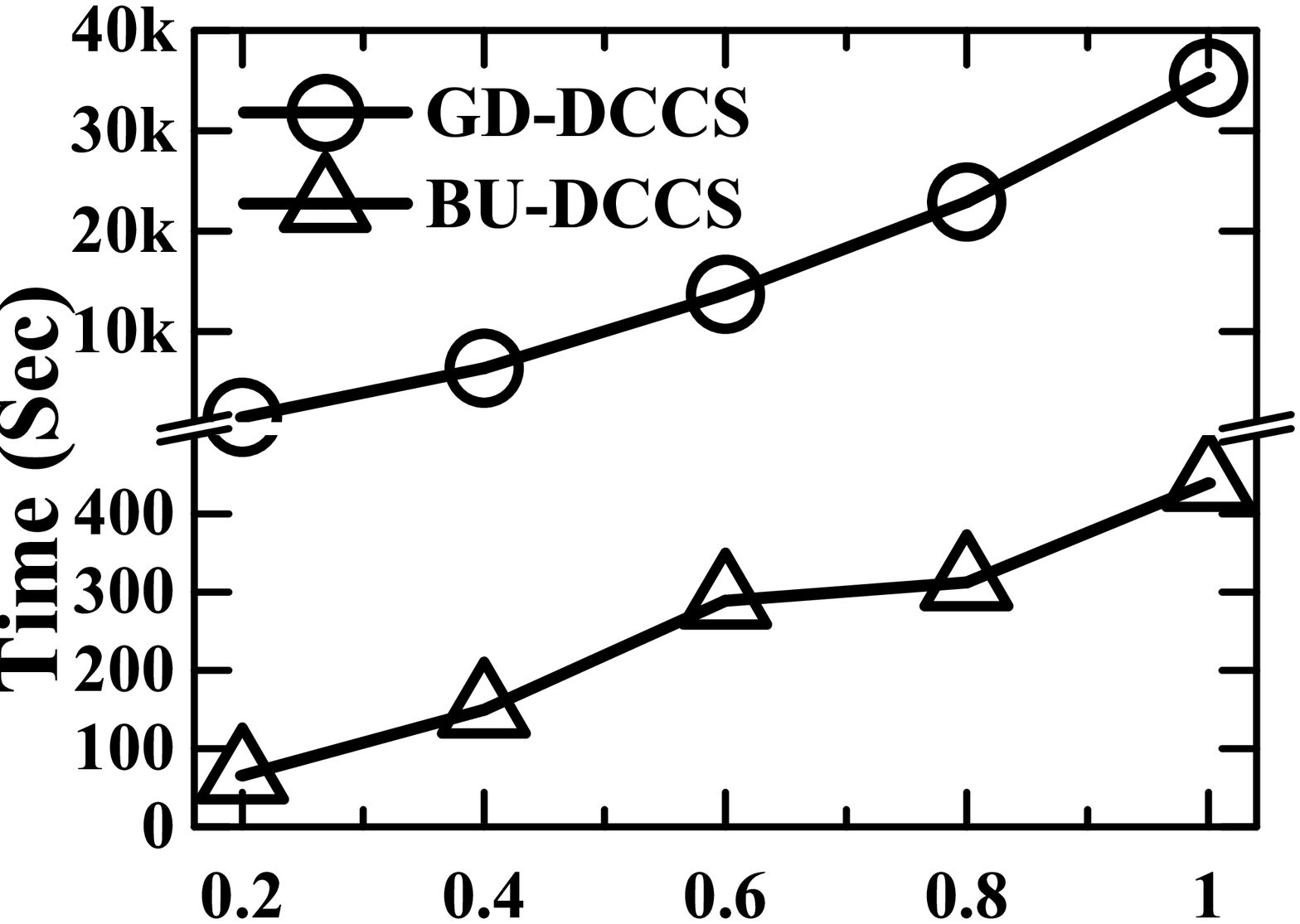}}
        \subfigure[\st{Stack (Vary $p$)}]{\includegraphics[width=0.49\linewidth]{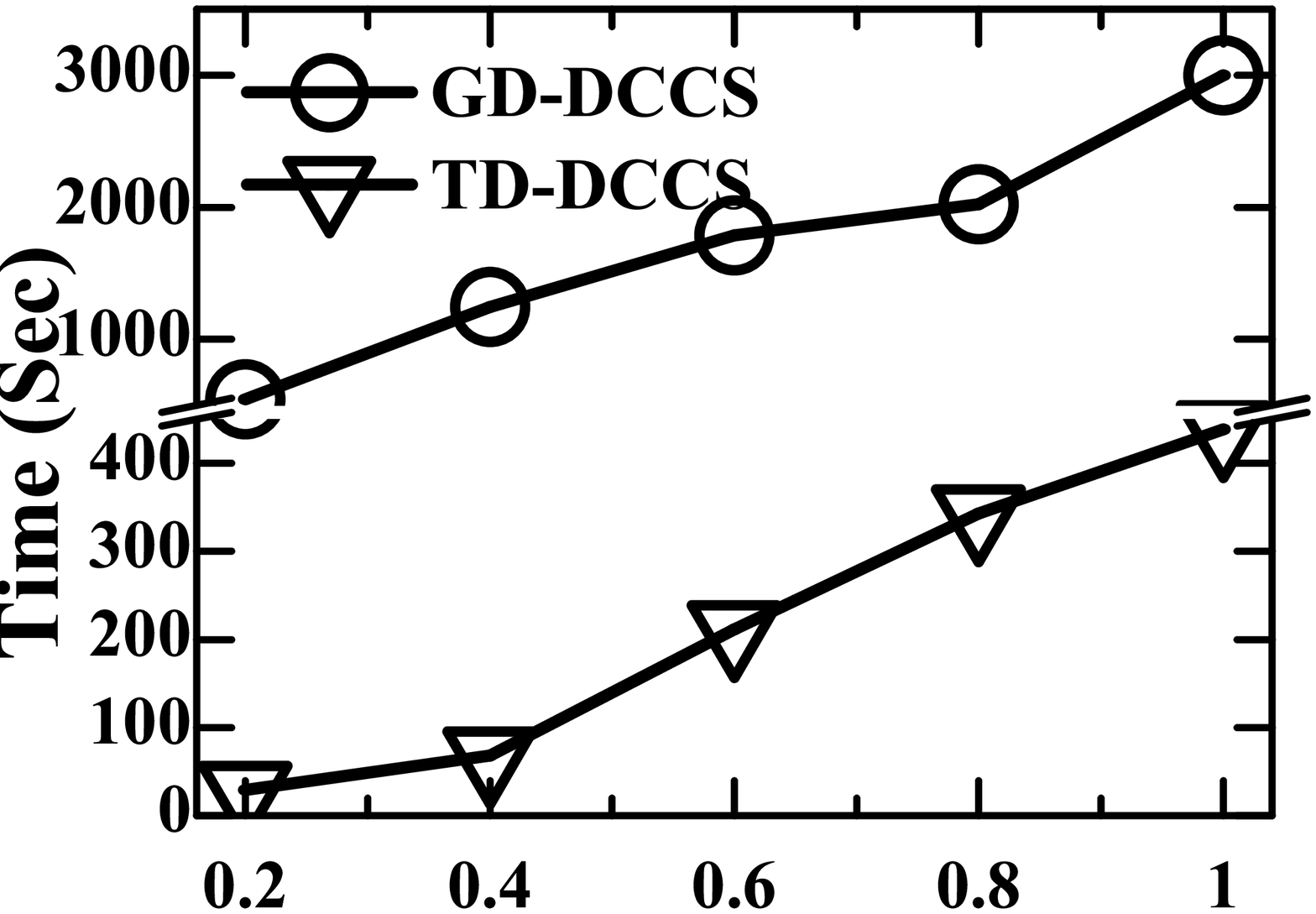}}
        \vspace{-1.5em}
        \caption{Execution Time vs $p$.}
        \label{Fig: Exp4V}
    \end{minipage}
    \vspace{-0.5em}
        \begin{minipage}[!t]{0.33\linewidth}
        \subfigure[\st{Stack (Vary $q$)}]{\includegraphics[width=0.49\linewidth]{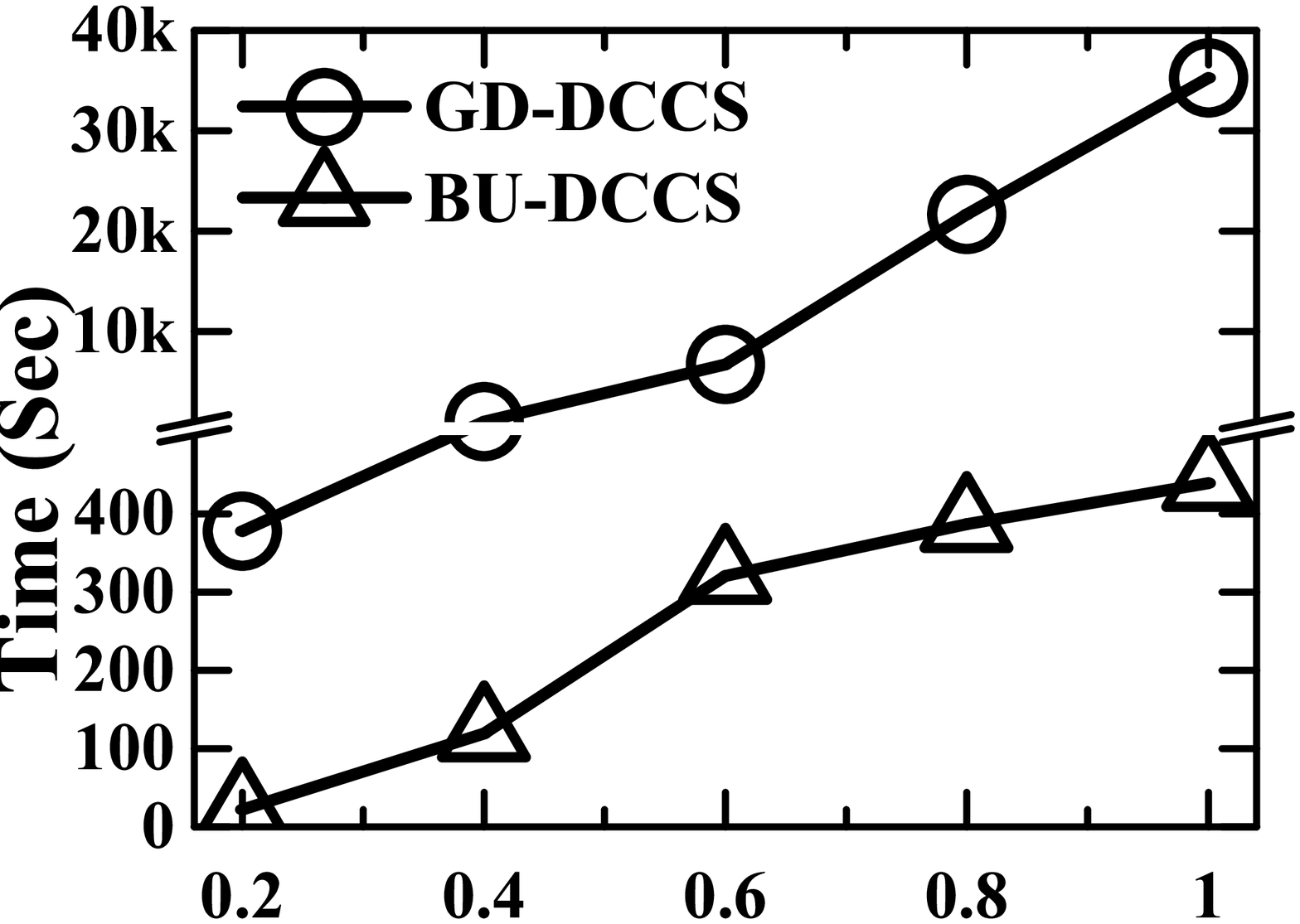}}
        \subfigure[\st{Stack (Vary $q$)}]{\includegraphics[width=0.49\linewidth]{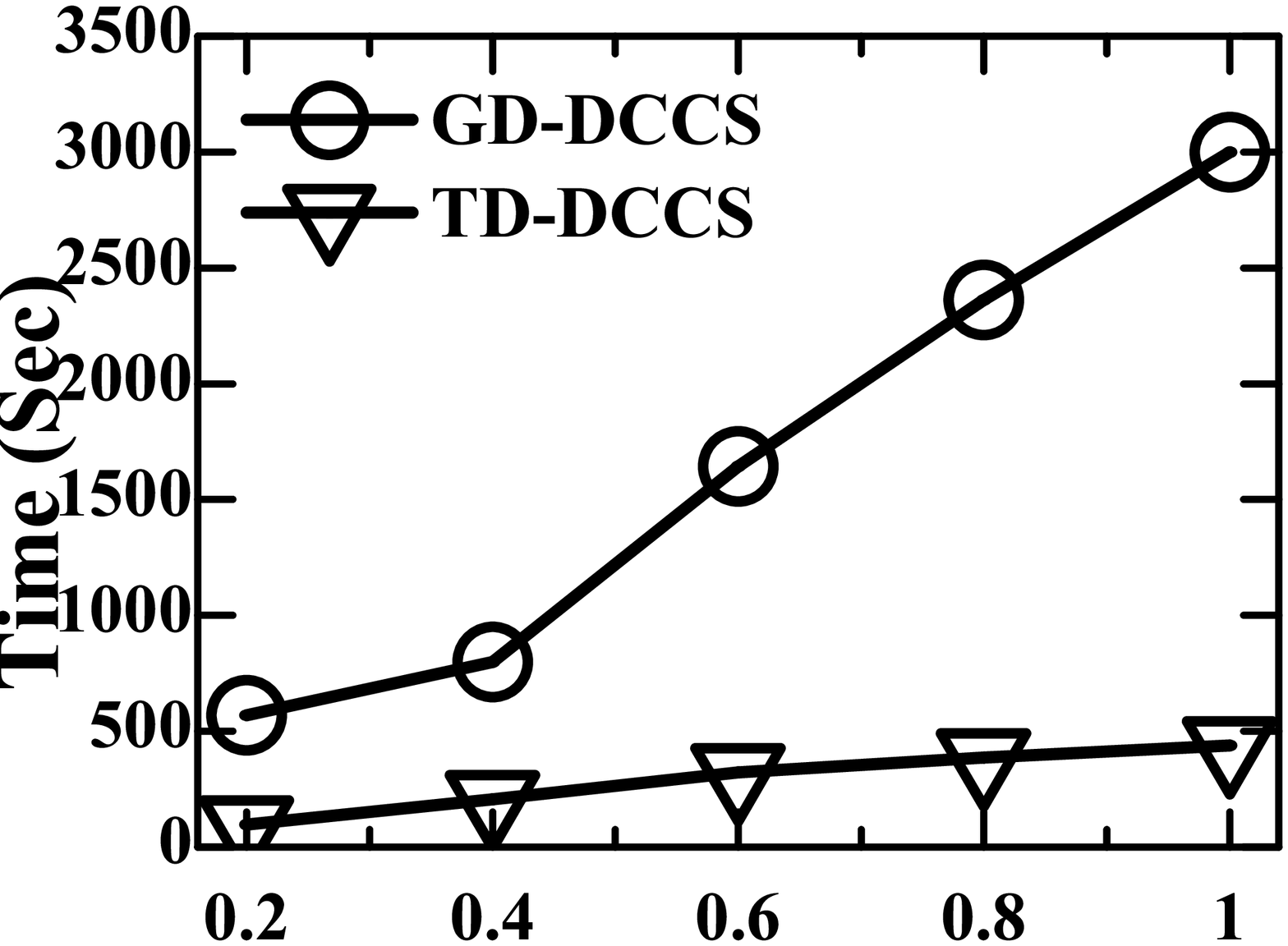}}
        \vspace{-1.5em}
        \caption{Execution Time vs $q$.}
        \label{Fig: Exp4L}
    \end{minipage}
     \vspace{-1.5em}
        \begin{minipage}[!t]{0.33\linewidth}
        %\subfigure[\st{Wiki (Small $s$)}]{\includegraphics[width=0.245\linewidth]{./ExpFig/Exp5TB4.pdf}}
        \subfigure[\st{Small $s$}]{\includegraphics[width=0.49\linewidth]{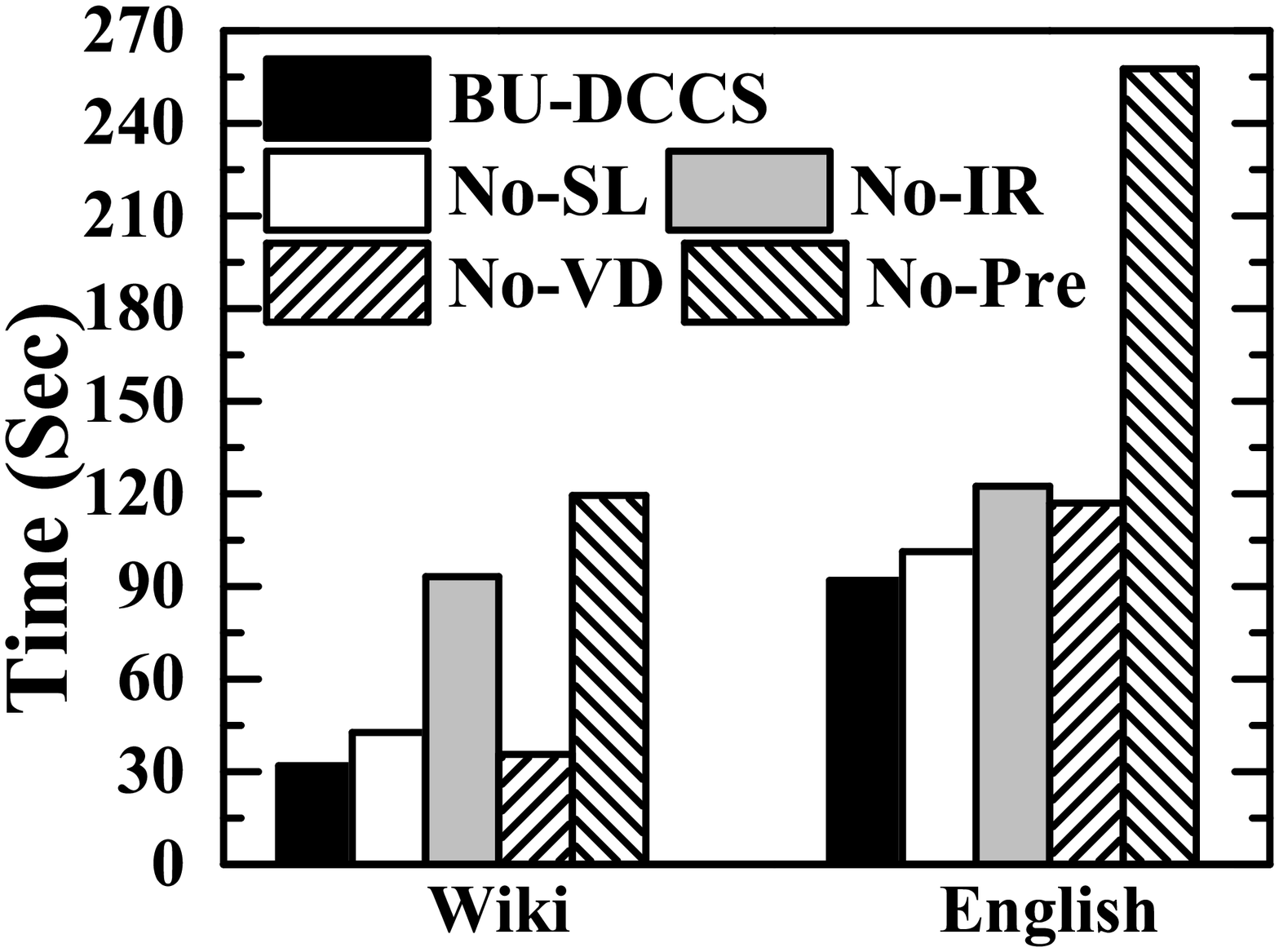}}
        %\subfigure[\st{Wiki (Large $s$)}]{\includegraphics[width=0.245\linewidth]{./ExpFig/Exp5TT4.pdf}}
        \subfigure[\st{Large $s$}]{\includegraphics[width=0.49\linewidth]{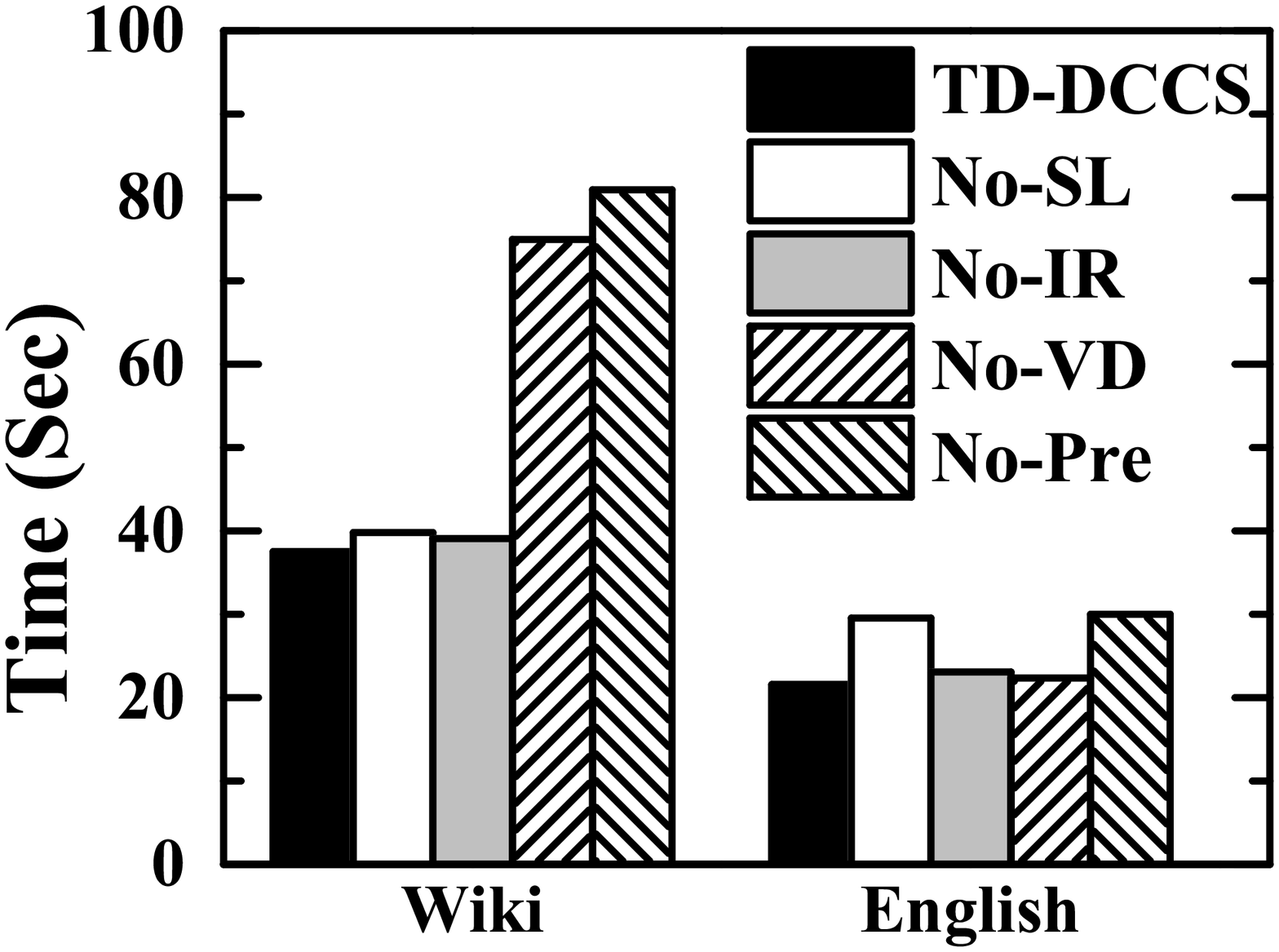}}
        \vspace{-1.5em}
        \caption{Effects of Preprocessing.}
        \label{Fig: Exp5TB}
    \end{minipage}
\end{figure*}

\noindent{\underline{\bf Effects of Parameter \textit{d}.}}
We examine the effects of parameter $d$ on the performance of the algorithms. By varying $d$, Fig.~\ref{Fig: Exp2TB} shows the execution time of \textsf{BU-DCCS} and \textsf{GD-DCCS} on datasets \textit{German} and \textit{English} for $s = 3$, and Fig.~\ref{Fig: Exp2TT} shows the execution time of \textsf{TD-DCCS} and \textsf{GD-DCCS} on \textit{German} and \textit{English} for $s = l(\G) - 2$. We observe that the execution time of all the algorithms decreases as $d$ grows. The reasons are as follows: 1) Due to Property~\ref{Lem:kHierarchy}, the size of $d$-CCs decreases as $d$ grows. Thus, \textsf{GD-DCCS} takes less time in selecting $d$-CCs, and \textsf{BU-DCCS} and \textsf{TD-DCCS} take less time in updating temporary results. 2) While $d$ increases, the size of the $d$-core on each layer decreases. By Lemma~\ref{Lem:kInjection}, the algorithms spend less time on $d$-CC computation. Moreover, both \textsf{BU-DCCS} and \textsf{TD-DCCS} are much faster than \textsf{GD-DCCS}.

Fig.~\ref{Fig: Exp2CB} and Fig.~\ref{Fig: Exp2CT} show the effects of $d$ on the cover size of the results of \textsf{BU-DCCS}, \textsf{TD-DCCS} and \textsf{GD-DCCS} for small $s$ and large $s$, respectively. We find that the cover size of the results decreases w.r.t.~$d$ for all the algorithms. This is simply because that the size of $d$-CCs decreases as $d$ increases. Therefore, the results cover less vertices for larger $d$. Moreover, the practical approximation quality of \textsf{BU-DCCS} and \textsf{TD-DCCS} is close to \textsf{GD-DCCS}.

%\smallskip

\noindent{\underline{\bf Effects of Parameter \textit{k}.}}
We examine the effects of parameter $k$ on the performance of the algorithms. By varying $k$, Fig.~\ref{Fig: Exp3TB} shows the execution time of \textsf{BU-DCCS} and \textsf{GD-DCCS} on datasets \textit{Wiki} and \textit{English} for $s = 3$, and Fig.~\ref{Fig: Exp3TT} shows the execution time of \textsf{TD-DCCS} and \textsf{GD-DCCS} on \textit{Wiki} and \textit{English} for $s = l(\G) - 2$. We have the following observations: 1) The execution time of \textsf{GD-DCCS} increases with $k$ because the time cost for selecting $d$-CCs in \textsf{GD-DCCS} is proportional to $k$. 2) Both \textsf{BU-DCCS} and \textsf{TD-DCCS} run much faster than \textsf{GD-DCCS}. 3) The execution time of \textsf{BU-DCCS} and \textsf{TD-DCCS} is insensitive to $k$. This is because the power of the pruning techniques in \textsf{BU-DCCS} and \textsf{TD-DCCS} relies on $|\Cov(\R)|$ according to Eq.~\eqref{Eqn: RUpdate}. As $k$ grows, $|\Cov(\R)|$ increases insignificantly, so $k$ has little effects on the execution time of \textsf{BU-DCCS} and \textsf{TD-DCCS}.

Fig.~\ref{Fig: Exp3CB} and Fig.~\ref{Fig: Exp3CT} show the effects of $k$ on the cover size of the results of \textsf{BU-DCCS}, \textsf{TD-DCCS} and \textsf{GD-DCCS} for small $s$ and large $s$, respectively. We find that the cover size grows w.r.t.~$k$; however, insignificantly for $k \geq 20$. From another perspective, it shows that there exists substantial overlaps among $d$-CCs. To reduce redundancy, it is meaningful to find top-$k$ diversified $d$-CCs on a multi-layer graph.

%\smallskip

\noindent{\underline{\bf Scalability w.r.t.~Parameters \textit{p} and \textit{q}.}}
We evaluate the scalability of the algorithms w.r.t.~the input multi-layer graph size. We control the graph size by randomly selecting a fraction $p$ of vertices or a fraction $q$ of layers from the original graph. Fig.~\ref{Fig: Exp4V} shows the execution time of \textsf{BU-DCCS}, \textsf{TD-DCCS} and \textsf{GD-DCCS} on the largest dataset \textit{Stack} by varying $p$ from $0.2$ to $1.0$. All the algorithms scale linearly w.r.t.~$p$ because the time cost of computing $d$-CCs is linear to the vertex count.

Fig.~\ref{Fig: Exp4L} shows the execution time of \textsf{BU-DCCS}, \textsf{TD-DCCS} and \textsf{GD-DCCS} on \textit{Stack} w.r.t.~$q$. We observed that: 1) The execution time of all algorithms grows with $q$. This is simply because the search space of the \textsc{DCCS} problem increases when the input multi-layer graph contains more layers. 2) The execution time of \textsf{GD-DCCS} grows much faster than \textsf{BU-DCCS} and \textsf{TD-DCCS}. The main reason is that both \textsf{BU-DCCS} and \textsf{TD-DCCS} adopt the effective pruning techniques to significantly reduce the search space. The number of candidate $d$-CCs examined by \textsf{GD-DCCS} grows much faster than those examined by \textsf{BU-DCCS} and \textsf{TD-DCCS}.

\noindent{\underline{\bf Effects of Preprocessing Methods.}}
We evaluate the effects of the preprocessing methods by disabling each (or all) of them in \textsf{BU-DCCS} (or \textsf{TD-DCCS}) and compare the execution time. Fig.~\ref{Fig: Exp5TB} shows the comparison results for \textsf{BU-DCCS} and \textsf{TD-DCCS}, respectively, where \textsf{No-VD} means ``vertex deletion is disabled'', \textsf{No-SL} means ``sorting layers is disabled'', \textsf{No-IR} means ``result initialization is disabled'', and \textsf{No-Pre} means ``all the preprocessing methods are disabled''. We have the following observations: 1) Every preprocessing method can improve the efficiency of \textsf{BU-DCCS} and \textsf{TD-DCCS}. It verifies that the preprocessing methods can reduce the size of the input graph (by vertex deletion) and enhance the pruning power of the algorithms (by sorting layers and result initialization). 2) A preprocessing method may have different effects for different algorithms. For example, the result initialization method has more significant effects in \textsf{BU-DCCS} than in \textsf{TD-DCCS}. This is because for smaller $s$, the cover size of the result is much larger according to Property~3. By Eq.~\eqref{Eqn: RUpdate}, the initial result can eliminate more candidates $d$-CCs in \textsf{BU-DCCS}.

\begin{figure}[!t]
    \centering
    \scriptsize
    \resizebox{0.85\columnwidth}{!}{
    \begin{tabular}{c|c|crrccc}
    	\hline
        \rowcolor{mygray}
    	Graph & $d$	& Algorithm & Time (Sec) & Size & Precision & Recall & $F_1$-score\\ \hline
       \multirow{6}{*}{\textit{PPI}} & \multirow{2}{*}{$2$}	& \textsf{MiMAG} & 6.28 & 58 & \multirow{2}{*}{0.598} & \multirow{2}{*}{$1$} & \multirow{2}{*}{0.748}\\
       &	& \textsf{BU-DCCS} & 0.078 & 97 & & &\\  \cline{2-8}
      &\multirow{2}{*}{$3$}	& \textsf{MiMAG} & 5.93 & 59 & \multirow{2}{*}{0.652} & \multirow{2}{*}{0.796} & \multirow{2}{*}{0.718}\\
      &	& \textsf{BU-DCCS} & 0.051 & 72 & & &\\  \cline{2-8}
      &\multirow{2}{*}{$4$}	& \textsf{MiMAG} & 5.16 & 55 & \multirow{2}{*}{0.631} & \multirow{2}{*}{0.745} & \multirow{2}{*}{0.683}\\
      &	& \textsf{BU-DCCS} & 0.02 & 65 & & &\\\hline
        \multirow{6}{*}{\textit{Author}} & \multirow{2}{*}{$2$}	& \textsf{MiMAG} & 13.90 & 122 & \multirow{2}{*}{0.682} & \multirow{2}{*}{$1$} & \multirow{2}{*}{0.811}\\
       &	& \textsf{BU-DCCS} & 0.091 & 179 & & &\\  \cline{2-8}
      & \multirow{2}{*}{$3$}	& \textsf{MiMAG} & 12.83 & 117 & \multirow{2}{*}{0.731} & \multirow{2}{*}{0.838} & \multirow{2}{*}{0.781}\\
      &	& \textsf{BU-DCCS} & 0.081 & 134 & & &\\  \cline{2-8}
      & \multirow{2}{*}{$4$} 	& \textsf{MiMAG} & 12.89 & 72 & \multirow{2}{*}{1} & \multirow{2}{*}{0.828} & \multirow{2}{*}{0.906}\\
      &	& \textsf{BU-DCCS} & 0.035 & 87 & & &\\ \hline
    \end{tabular}}
    \vspace{-0.3em}
    \caption{Comparison between \textsf{MiMAG} and \textsf{BU-DCCS}.}
     \label{Tab: QCCom}
     \vspace{-0.5em}
\end{figure}

\begin{figure}[!t]
    \centering
    \scriptsize
    \resizebox{0.7\columnwidth}{!}{
    \begin{tabular}{c|c|rrrrrr}
    	\hline
       \rowcolor{mygray}
    	 &  & \multicolumn{5}{c}{$|Q \cap \Cov(\R_c)|$} &\\ %\cline{3-8}
      \rowcolor{mygray}
        Graph &	 $|Q|$&  0 & 1 & 2  & 3 & 4 & 5 \\ \hline
       \multirow{3}{*}{\textit{PPI}} & 3 & 0 & 0 & 0 & {\bf 1.0} & --- & ---\\
       &  4 & 0 & 0.0045 & 0 & 0.1216 & {\bf 0.8739} & ---\\
       & 5 & 0 & 0 & 0 & 0 & 0.2759 & {\bf 0.7241}\\
      \hline
      \multirow{3}{*}{\textit{Author}} & 3 & 0 & 0 & 0 & {\bf 1.0} & --- & ---\\
       &  4 & 0 & 0.0045 & 0 & 0.0861 & {\bf 0.9139} & ---\\
       & 5 & 0 & 0.0506 & 0 & 0 & 0.1772 & {\bf 0.7722}\\
      \hline
    \end{tabular}
    }
    \vspace{-0.5em}
     \caption{Distribution of $|Q \cap \Cov(\R_c)|$.}
     \label{Tab: QCDis}
      \vspace{-2.5em}
\end{figure}

\noindent{\underline{\bf Comparison with Quasi-Clique Mining.}}
We compare our DCCS algorithms with the quasi-clique-based algorithm \textsf{MiMAG}~\cite{Boden2012Mining} for mining coherent subgraphs on a multi-layer graph. A set $Q$ of vertices in a graph is a $\gamma$-quasi-clique if each vertex in $Q$ is adjacent to $\gamma (|Q| - 1)$ other vertices in $Q$, where $\gamma \in [0, 1]$. Given a multi-layer graph $\G$ and parameters $\gamma \in [0, 1]$ and $d', s \in \mathbb{N}$, \textsf{MiMAG} finds a set of diversified vertex subsets $Q$ such that $|Q| \geq d'$ and $Q$ is a $\gamma$-quasi-clique on at least $s$ layers of $\G$. Since the datasets in our experiments are unlabelled graphs, the distance function of labels in \textsf{MiMAG} is disabled.

In the experiment, we set the parameters as follows. For the \textsf{MiMAG} algorithm, we set $\gamma = 0.8$ and $s = l(\G)/2$. For the \textsf{BU-DCCS} algorithm, we set $s = l(\G)/2$, $k = 10$. For fairness, \textsf{BU-DCCS} and \textsf{MiMAG} use the same parameter $s$. More over, when comparing \textsf{MiMAG} with \textsf{BU-DCCS}, we set $d' = d + 1$. We vary $d = 2, 3, 4$. Under this setting, the minimum degree constraints of a vertex in a dense subgraph generated by \textsf{BU-DCCS} and \textsf{MiMAG} are $d$ and $\lceil \gamma d \rceil$, which have the same value for $d = 2, 3, 4$ and $\gamma = 0.8$.

Let $\R_{Q}$ and $\R_{C}$ be the output of \textsf{MiMAG} and \textsf{BU-DCCS}, respectively. We compare them by five evaluation metrics: 1) execution time; 2) cover sizes $|\Cov(\R_{Q})|$ and $|\Cov(\R_{C})|$; 3) precision $\frac{| \Cov(\R_{Q}) \cap \Cov(\R_{C}) |}{|\Cov(\R_{C})|}$; 4) recall $\frac{| \Cov(\R_{Q}) \cap \Cov(\R_{C}) |}{|\Cov(\R_{Q})|}$; 5) $F_1$-score, i.e.~the harmonic mean of the precision and recall. The metrics 2--5 assess the similarity between $\R_{Q}$ and $\R_{C}$.

We ran \textsf{MiMAG} and \textsf{BU-DCCS} on datasets \textit{PPI} and \textit{Author}. The experimental results are shown in Fig.~\ref{Tab: QCCom}. We have three observations: 1) \textsf{BU-DCCS} runs much faster than \textsf{MiMAG}. This is because the search tree of \textsf{BU-DCCS} contains $2^{l(\G)}$ vertex subsets; while the search tree of \textsf{MiMAG} contains $2^{|V(\G)|}$ vertex subsets, where $l(\G) \ll |V(\G)|$. 2) The vertices covered by $\R_{Q}$ and $\R_{C}$ are significantly overlapped. Specifically, $\Cov(\R_{Q}) \cap \Cov(\R_{C})$ contains 70\%+ of vertices in $\Cov(\R_{Q})$ and 50\%+ of vertices in $\Cov(\R_{C})$. 3) The quasi-cliques in $\R_Q$ are largely contained in the $d$-CCs in $\R_C$ (entirely contained for most of the quasi-cliques). Fig.~\ref{Tab: QCDis} shows the detailed experimental results.

\begin{figure}[!t]
    \centering
    \includegraphics[width=\columnwidth]{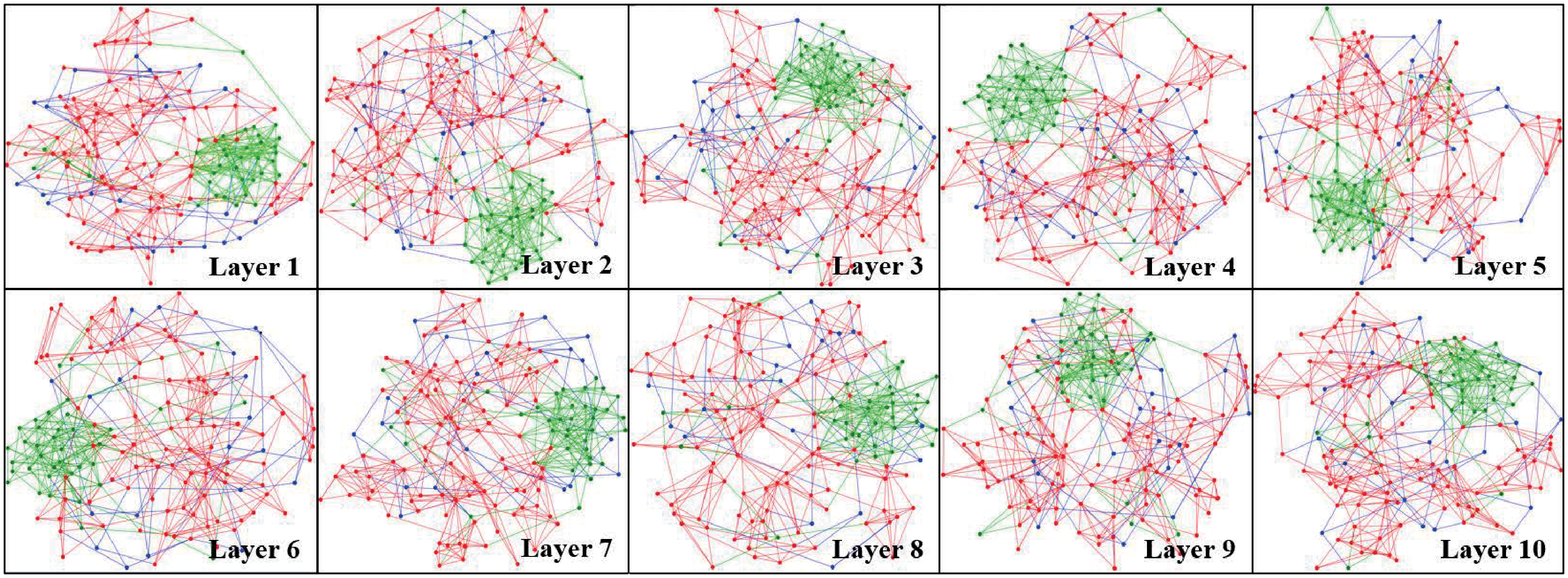}
    \vspace{-2em}
    \caption{Induced Coherent Dense Subgraphs on \textit{Author}.}
    \label{Fig: Exp6T}
    \vspace{-0.5em}
\end{figure}

We also analyze the differences between $\R_{Q}$ and $\R_{C}$. Fig.~\ref{Fig: Exp6T} shows the subgraphs induced by $\Cov(\R_C)$ and $\Cov(R_Q)$ on all layers of the \textit{Author} graph for $d = 3$. The vertices in $\Cov(\R_{C}) \cap \Cov(\R_{Q})$, $\Cov(\R_{C}) - \Cov(\R_{Q})$ and $\Cov(\R_{Q}) - \Cov(\R_{C})$ are colored in red, green and blue, respectively. We have two observations: 1) The vertices in $\Cov(\R_{Q}) - \Cov(\R_{C})$ (blue vertices) are sparsely connected compared with the vertices in $\Cov(\R_{Q}) \cap \Cov(\R_{C})$ (red vertices). 2) The vertices in $\Cov(\R_{C}) - \Cov(\R_{Q})$ (green vertices) are densely connected with themselves and with the vertices in $\Cov(\R_{C}) \cap \Cov(\R_{Q})$ (red vertices). The dense portion constituted by the vertices in $\Cov(\R_{C}) - \Cov(\R_{Q})$ found by \textsf{BU-DCCS} is missing from the result of \textsf{MiMAG}.

\begin{figure}[!t]
    \centering
    \scriptsize
    \resizebox{0.6\columnwidth}{!}{
    \begin{tabular}{c|ccc}
    	\hline
       \rowcolor{mygray}
    	 Algorithm & $d = 2$ & $d = 3$ & $d = 4$\\ \hline
        \textsf{MiMAG} & 69.7\% & 67.2\% & 65.3\% \\
        \textsf{BU-DCCS} & {\bf 83.6\%} & {\bf 80.1\%} & {\bf 77.9\%} \\
      \hline
    \end{tabular}
    }
    \vspace{-0.5em}
     \caption{Proportion of Protein Complexes Found by \textsf{MiMAG} and \textsf{BU-DCCS}.}
     \label{Tab: PPro}
      \vspace{-2.5em}
\end{figure}

Moreover, we compared protein complexes found by \textsf{MiMAG} and \textsf{BU-DCCS} on \textit{PPI}. We use the  MIPS database (\texttt{\small{http://mips.helmholtz-muenchen.de}}) as ground truth. For each protein complex on \textit{PPI}, if it is entirely contained in a dense subgraph, we say this protein complex is found. The proportion of protein complexes found by \textsf{MiMAG} and \textsf{BU-DCCS} with different $d$ is shown in Fig.~\ref{Tab: PPro}. We observe that: 1)  When $d$ increases, the proportion of found protein complexes decreases. This is because when $d$ increases, both the cover sizes $|\Cov(\R_{C})|$ and $|\Cov(\R_{Q})|$ become smaller. Thus, the dense subgraphs cover less number of protein complexes.  2) The proportion of protein complexes found by \textsf{BU-DCCS} is much higher than \textsf{MiMAG}. This is because the dense subgraphs generated by \textsf{BU-DCCS} cover more vertices than \textsf{MiMAG}. As we show before, some dense portions are missing from the result of \textsf{MiMAG}, so some protein complexes cannot be found by \textsf{MiMAG}.
This result verifies that \textsf{BU-DCCS} is more preferable than \textsf{MiMAG} for protein complex detection on biological networks.

In summary, \textsf{BU-DCCS} is much faster than \textsf{MiMAG} and produces larger coherent dense subgraphs than \textsf{MiMAG} (covering most of the quasi-cliques).

%%% Section 7 %%%
\section{Related Work}
\label{Sec: RWork}

Dense subgraph mining is a fundamental graph mining task, which has been extensively studied on single-layer graphs. Recently, mining dense subgraphs on graphs with multiple types of edges has attracted much attention. A detailed survey can be found in~\cite{Kim2015Community}. Basically, existing work can be categorized into two classes: dense subgraph mining on two-layer graphs and dense subgraph mining on general multi-layer graphs.

\noindent{\underline{\bf Dense Subgraph Mining on Two-layer Graphs.}}
Two-layer graph, is a special multi-layer graph. In a two-layer graph, one layer represents physical link structures, and the other represents conceptual connections between vertices derived from physical structures. The dense subgraph mining algorithms on two-layer graphs take both physical and conceptual connections into account. The algorithm in~\cite{Li2008Scalable} finds dense subgraphs by expanding from initial seed vertices. The algorithm~\cite{Qi2012Community} adopts edge-induced matrix factorization. In~\cite{Zhou2009Graph}, structural and attribute information are combined to form a unified distance measure, and a clustering algorithm is applied to detect dense subgraphs. In~\cite{Xu2012A}, structures and attributes are fused by a probabilistic model, and a model-based algorithm is proposed to find dense subgraphs. Other work on two-layer graphs includes the method based on correlation pattern mining~\cite{Silva2012Mining} and graph merging~\cite{Ruan2012Efficient}. All the algorithms are tailored to fit two-layer graphs. They only support the input where one layer represents physical connections, and the other represents conceptual connections. Therefore, they cannot be adapted to process general multi-layer graphs.

\noindent{\underline{\bf Dense Subgraph Mining on General Multi-layer Graphs.}}
A general multi-layer graph is composed by many layers representing different types of edges between vertices. Ref.~\cite{Tang2009Clustering} and~\cite{Dong2011Clustering} study dense subgraph mining using matrix factorization. The goal is to approximate the adjacency matrix and the Laplacian matrix of the graph on each layer. However, the matrix-based methods require huge amount of memory and are not scalable to large graphs. Alternatively, other work~\cite{Boden2012Mining, Pei2005On, Zeng2006Coherent} focus on finding dense subgraph patterns by extending the quasi-clique notion defined on single-layer graphs. In \cite{Zeng2006Coherent} and~\cite{Pei2005On}, the algorithms find cross-graph quasi-cliques. In \cite{Boden2012Mining}, the method is adapted to find diversified result to avoid redundancy. However, all these work has inherent limitations:
1) Quasi-clique-based methods are computationally costly.
2) The diameter of the discovered dense subgraphs are often very small. As verified by the experimental results in Section~\ref{Sec:PEvaluation}, the quasi-clique-based methods tend to miss large dense subgraphs.

We also discuss on some other related work.

\noindent{\underline{\bf Frequent Subgraph Pattern Mining.}}
Given a set $D$ of labelled graphs, frequent subgraph pattern mining discovers all subgraph patterns that are subgraph isomorphic to at least a fraction \emph{minsup} of graphs in $D$ (i.e., frequent)~\cite{Yan2002gSpan}. Our work is different from frequent subgraph pattern mining in two main aspects: 1) The graphs in $D$ are labelled graphs. A vertex in a graph may not be identical to any vertex in other graphs. Hence, the graphs in $D$ usually do not form a multi-layer graph. Inversely, a multi-layer graph is not necessary to be labelled. 2) A frequent subgraph pattern represents a common substructure recurring in many graphs in $D$. However, a $d$-CC is a set of vertices, and they are not required to have the same link structure on different layers of a multi-layer graph.

\noindent{\underline{\bf Clustering on Heterogeneous Information Networks.}}
Heterogeneous Information Network (HIN for short) is a logical network composed by multiple types of links between multiple types of objects. The clustering problem on HINs has been well studied in~\cite{Sun2009Ranking}. This work is different from our work in two aspects:
1) HIN characterizes the relationships between different types of objects. Normally, only one type of edges between two different types of vertices is considered. However, a multi-layer graph models multiple types of relationships between homogenous objects of the same type.
2) HIN is single-layer graph. The clustering algorithm only consider the cohesiveness of a vertex subset rather than its support.

\noindent{\underline{\bf $d$-Cores on Single-Layer Graphs.}}
The notion of $d$-core is widely used to represent dense subgraphs on single-layer graphs. It has many useful properties and has been applied to community detection~\cite{Li2015Influential}. However, the $d$-core notion only considers density of but ignores support. In this paper, we propose the $d$-CC notion, which extends the $d$-core notion by 1) considering both density and support of dense subgraphs and 2) inheriting the elegant properties of $d$-cores.

%%% Section 8 %%%
\section{Conclusions}

This paper addresses the diversified coherent core search (\textsc{DCCS}) problem on multi-layer graphs. The new notion of $d$-coherent core ($d$-CC) has three elegant properties, namely uniqueness, hierarchy and containment. The greedy algorithm is $(1 - 1/e)$-approximate; however, it is not efficient on large multi-layer graphs. The bottom-up and the top-down DCCS algorithms are $1/4$-approximate. For $s < l(\G)/2$, the bottom-up algorithm is faster than the other ones; for $s \ge l(\G)/2$, the top-down algorithm is faster than the other ones. The DCCS algorithms outperform the quasi-clique-based cohesive subgraph mining algorithm in terms of both time efficiency and result quality.

\scriptsize

%%% References %%%
\bibliographystyle{abbrv}
\bibliography{DCCS}

\clearpage

%%% Appendix %%%
%\appendix

\setcounter{property}{0}
\setcounter{lemma}{0}
\setcounter{theorem}{0}
\setcounter{section}{0}

\renewcommand\thesection{\Alph{section}}

\centerline{\Large{\textsc{\textbf{Appendix}}}}

\normalsize

\section{Proofs}

\label{Sec:PProperty}

\noindent{\textit{1. Proof of Property~1}}

\begin{property}[Uniqueness]
\label{Lem:kUnique}
Given a multi-layer graph $\G$ and a subset $L \subseteq [l(\G)]$, $C^{d}_{L}(\G)$ is unique for $d \in \mathbb{N}$.
\end{property}

\begin{proof}
Suppose $C^{d}_{L}(\G)$ is not unique. Let $C_1, C_2, \dots, C_n$ be the distinct instances of $C^{d}_{L}(\G)$. Due to the maximality of $d$-CC, we have $C_i \not\subseteq C_j$ for $i \ne j$. Let $C = \bigcup_{j = 1}^{n} C_j$. For each layer number $l \in L$, $G_l[C_i]$ is a subgraph of $G_l[C]$ for all $1 \leq i \leq n$. Thus, for each vertex $v \in C$, we have
\begin{equation*}
	d_{G_l[C]}(v) \geq \max_{1 \leq i \leq n} d_{G_l[C_i]}(v) \geq d
\end{equation*}
for every layer number $l \in L$. By definition, $C$ is also a $d$-CC of $\G$ w.r.t.~$L$. Due to the maximality of $d$-CC, none of $C_1, C_2, \dots, C_n$ is a $d$-CC of $\G$ w.r.t.~$L$. It leads to contradiction. Hence, $C^{d}_{L}(\G)$ is unique.
\end{proof}

\medskip
\noindent{\textit{2. Proof of Property~2}}

\begin{property}[Hierarchy]
\label{Lem:kHierarchy}
Given a multi-layer graph $\G$ and a subset $L \subseteq [l(\G)]$, we have $C^{d}_{L}(\G) \subseteq C^{d - 1}_{L}(\G) \subseteq \dots \subseteq C^{1}_{L}(\G) \subseteq C^{0}_L(\G)$ for $d \in \mathbb{N}$.
\end{property}

\begin{proof}
Let $d_1, d_2 \in \mathbb{N}$ and $d_1 > d_2$. For each vertex $v \in C^{d_1}_{L}(\G)$, we have
\begin{equation*}
	d_{G_l[C^{d_1}_{L}(\G)]}(v) \geq d_1 > d_2
\end{equation*}
for every layer number $l \in L$. By the definition of $d$-CC, $C^{d_1}_{L}(\G) \subseteq C^{d_2}_{L}(\G)$. Thus, the property holds.
\end{proof}

\medskip
\noindent{\textit{3. Proof of Property~3}}

\begin{property}[Containment]
\label{Lem:PHierarchy}
Given a multi-layer graph $\G$ and two subsets $L, L' \subseteq [l(\G)]$, if $L \subseteq L'$, we have $C^{d}_{L'}(\G) \subseteq C^{d}_{L}(\G)$ for $d \in \mathbb{N}$.
\end{property}

\begin{proof}
For each vertex $v \in C^{d}_{L'}(\G)$, we have
\begin{equation*}
	d_{G_l[C^{d}_{L'}(\G)]}(v) \geq d
\end{equation*}
for each layer number $l \in L$. Based on the definition of $d$-CC, we have $C^{d}_{L'}(\G) \subseteq C^{d}_{L}(\G)$. Hence, the property holds.
\end{proof}

\medskip
\noindent{\textit{4. Proof of Lemma~1}}

\begin{lemma}[Intersection Bound]
\label{Lem:kInjection}
Given a multi-layer graph $\G$ and two subsets $L_1, L_2 \subseteq [l(\G)]$, we have $C_{L_1 \cup L_2}^{d} (\G) \subseteq C_{L_1}^{d} (\G) \cap C_{L_2}^{d}(\G)$ for $d \in \mathbb{N}$.
\end{lemma}

\begin{proof}
First, we have $L_1 \cup L_2 \subseteq L_1$ and $L_1 \cup L_2 \subseteq L_2$. By Property~3, we have $C_{L_1 \cup L_2}^{d} (\G) \subseteq C_{L_1}^{d} (\G)$ and $C_{L_1 \cup L_2}^{d} (\G) \subseteq C_{L_2}^{d} (\G)$.
Thus, $C_{L_1 \cup L_2}^{d} (\G) \subseteq C_{L_1}^{d} (\G) \cap C_{L_2}^{d}(\G)$.
\end{proof}

\medskip
\noindent{\textit{5. Proof of Lemma~2}}

\begin{lemma}[Search Tree Pruning]
\label{Lem: dCCPrune}
For a $d$-CC $C^d_{L}(\G)$, if $C^d_{L}(\G)$ does not satisfy Eq.~\eqref{Eqn: RUpdate}, none of the descendants of $C^d_{L}(\G)$ can satisfy Eq.~\eqref{Eqn: RUpdate}.
\end{lemma}

\begin{proof}
For any descendant $C^d_{L'}(\G)$ of $C^d_{L}(\G)$, we have $L \subseteq L'$. By Property~3, we have $C^d_{L'}(\G) \subseteq C^d_{L}(\G)$. Thus,
\begin{align*}
	~ & \Cov((\R - \{C^{*}(\R)\}) \cup \{C^d_{L'}(\G)\})  \\
	&  \subseteq \Cov((\R - \{C^{*}(\R)\}) \cup \{C^d_{L}(\G)\}).
\end{align*}

Obviously, if
\begin{equation*}
|\Cov((\R - \{C^{*}(\R)\}) \cup \{C^d_{L}(\G)\})| < \left(1 + \frac{1}{k}\right) |\Cov(\R)|,
\end{equation*}
we have
\begin{equation*}
|\Cov((\R - \{C^{*}(\R)\}) \cup \{C^d_{L'}(\G)\})| < \left(1 + \frac{1}{k}\right) |\Cov(\R)|.
\end{equation*}
Thus, $C^d_{L'}(\G)$ cannot satisfy Eq.~\eqref{Eqn: RUpdate}.
\end{proof}

\medskip
\noindent{\textit{6. Proof of Lemma~3}}

\begin{lemma}[Order-based Pruning]
\label{Lem: LayerOrder}
For a $d$-CC $C^d_L(\G)$ and $j > \max(L)$, if $|C_{L}^d(\G) \cap C^d(G_j)| < \frac{1}{k}|\Cov(\R)| + |\Delta(\R, C^*(\R))|$, then $C_{L \cup \{j\}}^d(\G)$ cannot satisfy Eq.~\eqref{Eqn: RUpdate}.
\end{lemma}

\begin{figure}[t]
	\centering
	\includegraphics[width=0.6\columnwidth]{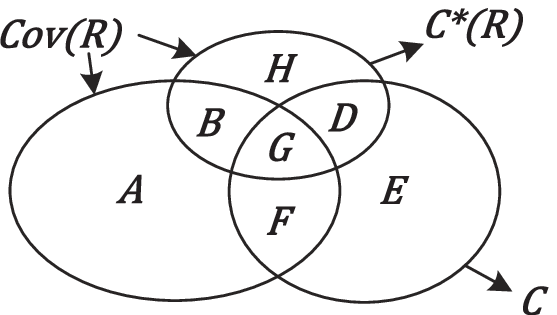}
	\vspace{-1em}
	\caption{Relationships between $\Cov(\R)$, $C^{*}(\R)$ and $C$.}
	\label{Fig: Venn}
\vspace{-2em}
\end{figure}

\begin{proof}
According to the definitions of $d$-CC and $d$-core, we have $C^d(G_j) = C^d_{\{j\}}(\G)$. For ease of presentation, let $C = C_{L}^d(\G) \cap C^d(G_j)$. We illustrate the relationships between $\Cov(\R)$, $C^{*}(\R)$ and $C$ in Fig.~\ref{Fig: Venn} with 7 disjoint subsets $A, B, D, E, F, G$ and $H$. We have
\begin{align*}
	~ & |\Cov(\R)| = |A| + |B| + |D| + |F| + |G| + |H|,\\
	~ & |C^{*}(\R)| = |B| + |D| + |G| + |H|,\\
	~ & |C| = |D| + |E| + |F| + |G|,\\
	~ & |\Delta(\R, C^*(\R))| = |D| + |H|.
\end{align*}
Since $|C| < \frac{1}{k}|\Cov(\R)| + |\Delta(\R, C^*(\R))|$, we have
\begin{align*}
	~ & |D| +|E| + |F| + |G| \\
	&  < \frac{1}{k} (|A| + |B| + |D| + |F| + |G| + |H|) + |D| + |H|.
\end{align*}
Thus,
\begin{equation*}
\begin{split}
	~ &\, |\Cov((\R - \{C^{*}(\R)\}) \cup \{C\})| \\
	= &\, |A| + |B| + |D| + |E| + |F| + |G| \\
	< &\, \frac{1}{k} (|A| + |B| + |D| + |G| + |F| + |H|) \\
    + & |A| + |B| + |D| + |H| \\
	\le &\, \left(1 + \frac{1}{k}\right) (|A| + |B| + |D| + |G| + |F| + |H|) \\
	= &\, \left(1 + \frac{1}{k}\right) |\Cov(\R)|.
\end{split}
\end{equation*}
By Lemma~1, we have $C^{d}_{L \cup \{j\}}(\G) \subseteq C$, so
\begin{align*}
	~ & \Cov((\R - \{C^{*}(\R)\}) \cup \{C^{d}_{L}(\G)\})  \\
	& \subseteq \Cov((\R - \{C^{*}(\R)\}) \cup \{C\}).
\end{align*}

Then, we have
\begin{align*}
	~ & |\Cov((\R - \{C^{*}(\R)\}) \cup \{C^{d}_{L}(\G)\})|\\
	\le & |\Cov((\R - \{C^{*}(\R)\}) \cup \{C\})| < \left(1 + \frac{1}{k}\right) |\Cov(\R)|.
\end{align*}
The lemma thus holds.
\end{proof}

\medskip
\noindent{\textit{7. Proof of Lemma~4}}

\begin{lemma}[Layer Pruning]
\label{Lem: LayerPrune}
For a $d$-CC $C^d_{L}(\G)$ and $j > \max(L)$, if $C^d_{L \cup \{j\}}(\G)$ does not satisfy Eq.~\eqref{Eqn: RUpdate}, then $C^d_{L' \cup \{j\}}(\G)$ cannot satisfy Eq.~\eqref{Eqn: RUpdate} for all $L'$ such that $L \subseteq L' \subseteq [l(\G)]$.
\end{lemma}

\begin{proof}
Since $L \subseteq L'$, we have $L \cup \{j\} \subseteq L' \cup \{ j\}$. According to Property~3, we have $C^d_{L' \cup \{j\}}(\G)  \subseteq C^d_{L \cup \{j\}}(\G)$.
Therefore,
\begin{align*}
	~ & \Cov((\R - \{C^{*}(\R)\}) \cup \{C^d_{L' \cup \{j\}}(\G)\})  \\
	& \subseteq \Cov((\R - \{C^{*}(\R)\}) \cup \{C^d_{L \cup \{j\}}(\G)\}).
\end{align*}

Since $C^d_{L \cup \{j\}}(\G)$ does not satisfy Eq.~\eqref{Eqn: RUpdate}, we have
\begin{align*}
	~ & |\Cov((\R - \{C^{*}(\R)\}) \cup \{C^d_{L' \cup \{j\}}(\G)\})|\\
	\le & |\Cov((\R - \{C^{*}(\R)\}) \cup \{C^d_{L \cup \{j\}}(\G)\})|\\
	< & \left(1 + \frac{1}{k}\right) |\Cov(\R)|.
\end{align*}
Thus, the lemma holds.
\end{proof}

\medskip
\noindent{\textit{8. Proof of Lemma~5}}

\begin{lemma}[Search Tree Pruning]
\label{Lem: TDdCCPrune}
For a $d$-CC $C^d_L(\G)$ and its potential vertex set $U^d_L(\G)$, where $|L| > s$, if $U^d_{L}(\G)$ does not satisfy Eq.~\eqref{Eqn: RUpdate}, any descendant $C^d_{L'}(\G)$ of $C^d_L(\G)$ with $|L'| = s$ cannot satisfy Eq.~\eqref{Eqn: RUpdate}.
\end{lemma}

\begin{proof}
According to the usage of potential sets, for any descendant $C^d_{L'}(\G)$ of $C^d_L(\G)$ with $|L'| = s$, we have $C^d_{L'}(\G) \subseteq U^d_{L}(\G)$.
Thus, we have
\begin{align*}
	~ & \Cov((\R - \{C^{*}(\R)\}) \cup \{C^d_{L'}(\G)\})  \\
	& \subseteq \Cov((\R - \{C^{*}(\R)\}) \cup \{U^d_{L}(\G)\}).
\end{align*}

Since $U^d_{L}(\G)$ does not satisfy Eq.~\eqref{Eqn: RUpdate}, we have
\begin{align*}
	~ & |\Cov((\R - \{C^{*}(\R)\}) \cup \{C^d_{L'}(\G)\})|\\
	< & |\Cov((\R - \{C^{*}(\R)\}) \cup \{U^d_{L}(\G)\})|< \left(1 + \frac{1}{k}\right) |\Cov(\R)|.
\end{align*}
The lemma thus holds.
\end{proof}

\medskip
\noindent{\textit{9. Proof of Lemma~6}}

\begin{lemma}[Order-based Pruning]
\label{Lem: TDLayerOrder}
For a $d$-CC $C^d_L(\G)$, its potential vertex set $U^d_L(\G)$ and $j > \max ([l(\G)] - L)$, if $|U^d_{L - \{ j\}}(\G)| < \frac{|\Cov(\R)|}{k} + |\Delta(\R, C^*(\R))|$, any descendant $C^d_{L - \{ j \}}(\G)$ of $C^d_L(\G)$ cannot satisfy Eq.~\eqref{Eqn: RUpdate}.
\end{lemma}

\begin{proof}
Similar to the proof of Lemma~\ref{Lem: LayerOrder}, if $|U^d_{L - \{ j\}}(\G)| < \frac{1}{k}|\Cov(\R)| + |\Delta(\R, C^*(\R))|$, we have
\begin{align*}
	~ & |\Cov((\R - \{C^{*}(\R)\}) \cup \{U^d_{L - \{ j\}}(\G)\})|  < \left(1 + \frac{1}{k}\right)|\Cov(\R)|.
\end{align*}

According to the usage of potential sets, for any descendant $C^d_{L'}(\G)$ of $C^d_L(\G)$ with $|L'| = s$, we have $C^d_{L'}(\G) \subseteq U^d_{L}(\G)$. Thus, we have
\begin{align*}
	~ & |\Cov((\R - \{C^{*}(\R)\}) \cup \{C^d_{L'}(\G)\})|\\
	\le & |\Cov((\R - \{C^{*}(\R)\}) \cup \{U^d_{L - \{ j\}}(\G)\})|\\
	< & \left(1 + \frac{1}{k}\right)|\Cov(\R)|.
\end{align*}
Therefore, the lemma holds.
\end{proof}

\setcounter{equation}{1}

\medskip
\noindent{\textit{10. Proof of Lemma~7}}

\begin{lemma}[Potential Set Pruning]
\label{Lem: TDUPPrune}
For a $d$-CC $C^d_L(\G)$ and its potential vertex set $U^d_L(\G)$, where $|L| > s$, if $C^{d}_{L}(\G)$ satisfies Eq.~\eqref{Eqn: RUpdate}, and $U^d_L(\G)$ satisfies
\begin{equation}
%\small
\label{Eqn: UCondition}
|U^{d}_{L}(\G)| <  (\tfrac{1}{k} + \tfrac{1}{k^2}) |\Cov(\R)|+ (1 + \tfrac{1}{k})|\Delta(\R, C^{*}(\R))|,
\end{equation}
the following proposition holds: For any two distinct descendants $C^d_{S_1}(\G)$ and $C^d_{S_2}(\G)$ of $C^{d}_{L}(\G)$ such that $|S_1| = |S_2| = s$, if $|\R| = k$ and $\R$ has already been updated by $C^{d}_{S_1}(\G)$, then $C^{d}_{S_2}(\G)$ cannot update $\R$ any more.
\end{lemma}

\begin{figure}[t]
	\centering
	\includegraphics[width=0.6\columnwidth]{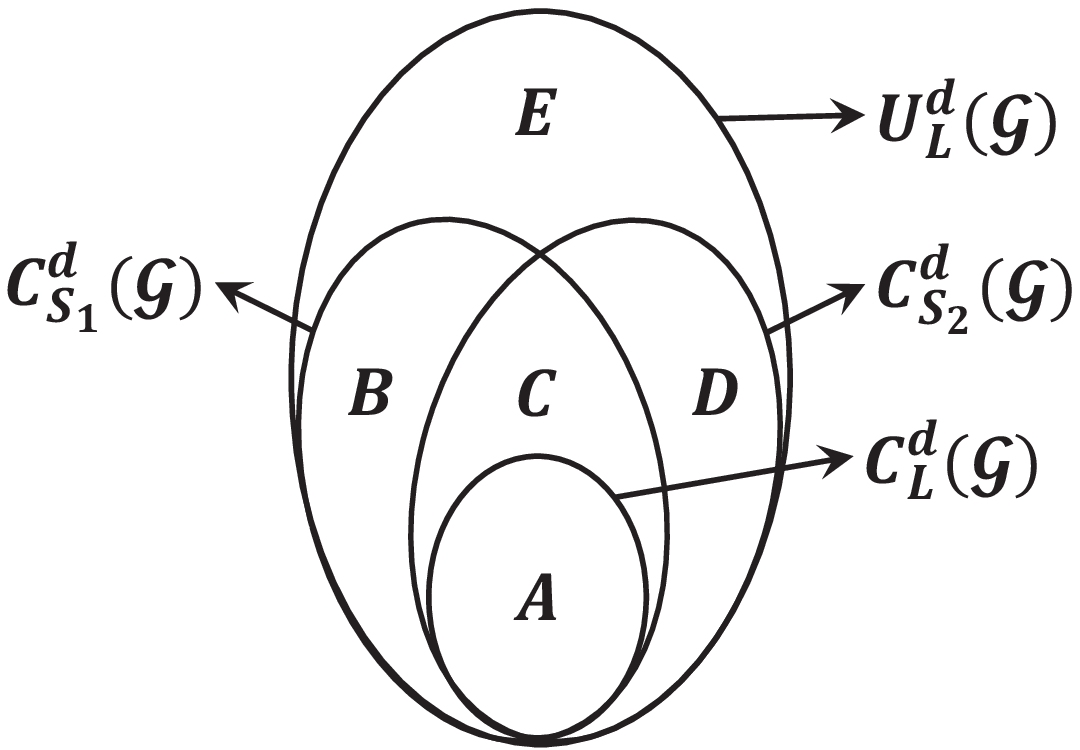}
	\caption{Relationships between $C^d_{L}(\G)$, $C^d_{L_1}(\G)$, $C^d_{L_2}(\G)$ and $U^d_{L}(\G)$.}
	\label{Fig: Venn2}
\vspace{-2em}
\end{figure}

\begin{proof}
We illustrate the relationships between $C^d_{L}(\G)$, $C^d_{S_1}(\G)$, $C^d_{S_2}(\G)$ and $U^d_{L}(\G)$ in Fig.~\ref{Fig: Venn2} with five disjoint subset $A$, $B$, $C$, $D$ and $E$.  We have
\begin{align*}
	~ & |C^d_{S_1}(\G)| = |A| + |B| + |C|,\\
	~ & |C^d_{S_2}(\G)| = |A| + |C| + |D|,\\
	~ & |U^{d}_{L}(\G)| = |A| + |B| + |C| + |D| + |E|,\\
	~ & |C^d_{S_1}(\G) \cap C^d_{S_2}(\G)| = |A|.
\end{align*}

Since $C^d_{S_1}(\G)$ can update $\R$, Lemma~\ref{Lem: LayerOrder} implies that
\begin{equation*}
	|C^d_{S_1}(\G)| \geq \frac{1}{k}|\Cov(\R)| + |\Delta(\R, C^*(\R))|.
\end{equation*}
Let $\R'$ be the resulting $\R$ after updating $\R$ with $C^d_{S_1}(\G)$. We have
\begin{equation*}
	|\Cov(\R')| \geq \left(1 + \frac{1}{k}\right) |\Cov(\R)|.
\end{equation*}

Suppose that $C^d_{S_2}(\G)$ can update $\R'$ again. We have
\begin{equation*}
	\begin{split}
		~ & |\Cov((\R' - \{C^{*}(\R')\}) \cup \{C^d_{S_2}(\G)\})| \\
		& \geq \left(1 + \frac{1}{k}\right)|\Cov(\R')| \geq \left(\frac{1}{k} + \frac{1}{k^2}\right) |\Cov(\R)|.
	\end{split}
\end{equation*}
Since $A \cup C \subset C^d_{S_2}(\G)$, we have
\begin{align*}
	~ & \Cov((\R' - \{C^{*}(\R')\}) \cup \{C^d_{S_2}(\G)\}) \\
	= & \Cov(\R') - \Delta(\R', C^{*}(\R')) + D \subseteq  \Cov(\R') + D.
\end{align*}

Putting the discussions together, we have
\begin{equation*}
	\begin{split}
		|\Cov(\R')| + |D|  & \geq |\Cov((\R' - \{C^{*}(\R')\}) \cup \{C^d_{S_2}(\G)\})| \\
		& \geq \left(1 + \frac{1}{k}\right) |\Cov(\R')|,
	\end{split}
\end{equation*}
that is, $|D| \geq \frac{1}{k} |\Cov(\R')|$. Thus, for $U^d_{L}(\G)$, we have
\begin{align*}
	~ & |U^d_{L}(\G)| = |A| + |B| + |C| + |D| + |E| \\
	& \geq |C^d_{S_1}(\G)| + |D|  \\
	& \geq \frac{1}{k}|\Cov(\R)| + |\Delta(\R, C^*(\R))| + \frac{1}{k} |\Cov(\R')| \\
	& = \left(\frac{1}{k} + \frac{1}{k^2}\right) |\Cov(\R)| + |\Delta(\R, C^*(\R))| + \frac{1}{k}|\Cov(\R)|\\
	& \geq \left(\frac{1}{k} + \frac{1}{k^2}\right) |\Cov(\R)| + \left(1 + \frac{1}{k}\right)|\Delta(\R, C^*(\R))|.
\end{align*}
The last equation holds due to the pigeonhole principle. For each $C' \in \R$, we must have $|\Delta(\R, C')| \leq \frac{1}{k}|\Cov(\R)|$.

The size of $U^d_{L}(\G)$ contradicts with Eq.~\eqref{Eqn: UCondition}. Thus, if $U^d_{L}(\G)$ satisfies Eq.~\eqref{Eqn: UCondition}, $C^d_{S_2}(\G)$ cannot update $\R$ any more.
\end{proof}

\setcounter{equation}{1}

\medskip
\noindent{\textit{11. Proof of Lemma~8}}

\begin{lemma}
\label{Lem:CScope}
$C^{d}_{L'}(\G) \subseteq U^{d}_{L'}(\G) \cap \left(\bigcup_{h = |L'|}^{l(\G)} I_{h} \right)$.
\end{lemma}

\begin{proof}
By the definition of potential set $U^{d}_{L'}(\G)$, we have $C^{d}_{L'}(\G) \subseteq U^{d}_{L'}(\G)$.
Obviously, if a vertex $v \in \bigcup_{h = 0}^{|L'| - 1} I_{h}$, the support of $v$ is less than $|L'|$. Thus, $v$ is unlikely to exist in a $d$-CC on at least $|L'|$ layers. Therefore, we must have $v \in \bigcup_{h = |L'|}^{l(\G)} I_{h}$. Hence, the lemma holds.
\end{proof}

\medskip
\noindent{\textit{12. Proof of Lemma~9}}

\begin{lemma}
\label{Lem:CFilter}
For each vertex $v\!\in\!C^{d}_{L'}(\G)$, there exists a sequence of vertices $w_0, w_1, \dots, w_n$ such that $L' \subseteq L(w_0)$, $w_n = v$, $w_{i+1}$ is placed on a higher level than $w_{i}$, and $(w_i, w_{i + 1})$ is an edge in the index.
\end{lemma}

\begin{proof}
We prove that if a vertex $v$ does not satisfies this condition, $v$ must not exist in $C^{d}_{L'}(\G)$. Obviously, we only need to consider vertices in $\bigcup_{h = |L'|}^{l(\G)} I_{h}$ by Lemma~\ref{Lem:CScope}.

First, we consider the vertex $v$ in the lowest level in the index. Obviously, if $L' \not\subseteq L(v)$, there must exist a layer number $j \in L'$ such that $v \notin C^{d}(G_j)$. By Lemma~1, $v$ cannot be contained in $C^{d}_{L'}(\G)$. Thus, we can remove $v$ from the graph $\G$. After that, we consider the vertices in next level of the lowest level. If $L' \not\subseteq L(u)$,   there must exist a layer number $j' \in L'$ such that $u \notin C^{d}(G_{j'})$. At this time, if none of $u$'s neighbors $w$ in the lowest level such that $L' \subseteq L(w)$, they have already been removed from $\G$, so vertex $u$ has the same neighbors as we build the index. Therefore, for layer number $j' \in L'$, we still have $u \notin C^{d}(G_{j'})$. By Lemma~1, $u$ cannot be contained in $C^{d}_{L'}(\G)$. We can continue this process level by level. This implies that all the vertices that do not satisfy this condition cannot exist in $C^{d}_{L'}(\G)$.
\end{proof}

\medskip
\noindent{\textit{13. Proof of Lemma~10}}

\begin{lemma}
The time complexity of Procedure \textsf{RefineC} is $O(n'l' + m')$.
\end{lemma}

\begin{proof}
To prove the time complexity of \textsf{RefineC}, we at first analyze the cases when an edge can be accessed. Notably, any edge $(u, v)$ on a layer of $\G[U^{d}_{L'}(\G)]$ can be accessed at most three times in the following cases:

1) At line~5 of the \textsf{RefineC} procedure, when computing $d^{+}_{i}(v)$ of all $i \in L'$ for each vertex $v \in Z$, each edge $(u, v)$ on a layer $i \in L'$ will be accessed exactly once.

2) At line~16 or line~27 of the \textsf{RefineC} procedure, when vertex $u$ accesses a vertex $v$ on a higher level, each edge $(u, v)$ on a layer $i \in L'$ will be accessed exactly once.

3) At line~2 of the \textsf{CascadeD} procedure, when updating $d^{+}_{i}(u)$, the edge $(u, v)$ on a layer $i \in L'$ will be accessed. Note that, $(u, v)$ on a layer $i \in L'$ will be accessed only once. This is because, when updating $d^{+}_{i}(u)$, $u$ is already been set to discarded. Thus, $u$ will never have opportunity to visit $v$ any more. Meanwhile, since $u$ is discarded, $v$ will also not visit vertex $u$ in the \textsf{CascadeD} procedure afterwards. As a result, each edge in $\G[U^{d}_{L'}(\G)]$ will be accessed at most once.

Putting them together, the edge access time is at most $O(\sum_{i \in L'} E_{i}(U^{d}_{L'}(\G))) = O(3m') = O(m')$. Meanwhile, at line~19, for each undetermined vertex $v$, we need to check whether $d^{+}_{i}(d)(v) < d$ for all $i \in L'$. So the maximum time cost is $O(|U^{d}_{L'}(\G)||L'|) = O(n'|L'|)$. As a result, the total time cost of Procedure \textsf{RefineC} is $O(n'l' + m')$.

\end{proof}

\medskip
\noindent{\textit{14. Proof of Theorem~1}}

\begin{theorem}
The \textsc{DCCS} problem is NP-complete.
\end{theorem}

\begin{proof}
Given a collection of sets $\F = \{C_1, C_2, \dots, C_n\}$ and $k \in \mathbb{N}$, the max-$k$-cover problem is to find a subset $\R \subseteq \F$ such that $|\R| = k$ and that $|\Cov(\R)|$ is maximized. The max-$k$-cover problem has been proved to be NP-complete unless P $=$ NP~\cite{Ausiello2011Online}.

It is easy to show that the DCCS problem is in NP. We prove the theorem by reduction from the max-$k$-cover problem in polynomial time. Given an instance $(\F, k)$ of the max-$k$-cover problem, we first construct a multi-layer graph $\G$. The vertex set of $\G$ is $\bigcup_{i = 1}^n C_i$. There are $n$ layers in $\G$. An edge $(u, v)$ exists on layer $i$ if and only if $u, v \in C_i$ and $u \ne v$. Then, we construct an instance of the DCCS problem $(\G, d, s, k)$, where $d = 1$ and $s = 1$. The result of the DCCS problem instance $(\G, d, s, k)$ is exactly the result of the max-$k$-cover problem instance $(\F, k)$. The reduction can be done in polynomial time. Thus, the DCCS problem is NP-complete.
\end{proof}

\medskip
\noindent{\textit{15. Proof of Theorem~2}}

\begin{theorem}\label{Thm:GreedyDCCS}
The approximation ratio of \textsf{GD-DCCS} is $1 - \frac{1}{e}$.
\end{theorem}

\begin{proof}
The approximation ratio of the greedy algorithm~\cite{Ausiello2011Online} for the max-$k$-cover problem is $1 - 1/e$. In the \textsf{GD-DCCS} algorithm, after obtaining the set $\F$ of all candidate $d$-CCs (lines~4--7), lines~8--10 select $k$ $d$-CCs from $\F$ in the same way as in the greedy algorithm~\cite{Ausiello2011Online}. Thus, the approximation ratio of the \textsf{GD-DCCS} algorithm is also $1 - 1/e$.
\end{proof}

\medskip
\noindent{\textit{16. Proof of Theorem~3}}

To prove Theorem~3, we first state the following claim. The correctness of the claim has been proved in~\cite{Ausiello2011Online}.

\begin{claim}
Let $\F = \{C_1, C_2, \dots, C_n\}$ and $k \in \mathbb{N}$. Let $\R^{*}$ the subset of $\F$ such that $|\R^*| = k$ and $|\Cov(\R^*)|$ is maximized. Let $\R \subseteq \F$ be a set obtained in the following way. Initially, $\R = \emptyset$. We repeat taking an element $C$ out of $\F$ randomly and updating $\R$ with $C$ according to the two rules specified in Section~\ref{Sec:BApproach-2} until $\F = \emptyset$. Finally, we have $|\Cov(\R)| \geq \frac{1}{4} |\Cov(\R^{*})|$.
\end{claim}

\begin{theorem}
The approximation ratio of \textsf{BU-DCCS} is $1/4$.
\end{theorem}

\begin{proof}
Note that the \textsf{BU-DCCS} algorithm uses the same procedure described in Claim~1 to update $\R$ except that some pruning techniques are applied as well. Therefore, we only need to show that the pruning techniques will not affect the approximation ratio stated in Claim~1. Let $C$ be a $d$-CC pruned by a pruning method and $D_C$ be the set of descendant candidate $d$-CCs of $C$ in the search tree. For all $C' \in D_C$, according to Lemma~2, Lemma~3 or Lemma~4, $C'$ must not update $\R$. By Claim~1, candidate $d$-CCs can be taken in an arbitrary order without affecting the approximation ratio. Therefore, we can safely ignore all the $d$-CCs in $D_C$ without affecting the quality of $\R$. Finally, we have $|\Cov(\R)| \ge \frac{1}{4} |\Cov(\R^*)|$. Thus, the theorem holds.
\end{proof}

\medskip
\noindent{\textit{17. Proof of Theorem~4}}

\begin{theorem}
The approximation ratio of \textsf{TD-DCCS} is $1/4$.
\end{theorem}

\begin{proof}
The \textsf{TD-DCCS} algorithm uses the same procedure described in Claim~1 to update $\R$ and applies some pruning techniques in addition. By the same arguments in the proof of Theorem~3, this theorem holds.
\end{proof}

%
%\subsection{Relationships between Three Types of Layers}
%
%Given $L \subseteq [l(\G)]$, we have $N_L = [l(\G)] - L$, $O_L = \{ j | j \in L, j > \max(N(L))\}$ and $M_L = L - O(L)$. Now, we argue the correctness of the statement. Obviously, a layer number removed from $[l(\G)]$ is a must-not-exist layer. Thus, $N_L = [l(\G)] - L$. Since we remove the layer number in an ascending order from $[l(\G)]$. For any layer number $i < j$ and $i, j \in N_L$, each layer number $i < k < j$ will certainly not be removed afterwards. Therefore, we have $O_L = \{ j | j \in L, j > \max(N(L))\}$. If a layer number $l > \max([l(\G)] - N_L)$, we can decide whether $l$ will be removed in the future. Therefore, $M_L = L - O(L)$.

\section{The \textsf{dCC} Procedure}
\label{Sec: IdCC}

We present the \textsf{dCC} procedure in Fig.~\ref{Fig: dCCDetailed}. It takes as input a multi-layer graph $\G$, a subset $L \subseteq [l(\G)]$ and an integer $d \in \mathbb{N}$ and outputs $C^{d}_{L} (\G)$, the $d$-CC w.r.t. $L$ on $\G$. For each vertex $v \in V(\G)$, let $m(v) = \min_{i \in L} d_{G_i}(v)$ be the minimum degree of $v$ on all layers in $L$. First, we compute $m(v)$ for each vertex $v \in V(\G)$ (line~1). Let $M = \max_{v \in V(\G)} m(v)$ (line~2). For each vertex $v \in V(\G)$, we have $0 \leq m(v) \leq M$. Therefore, we can assign all vertices of $\G$ into $M + 1$ bin according to $m(v)$. To facilitate the computation of $C^{d}_{L}(\G)$, we set up three arrays in the \textsf{dCC} procedure:
\begin{itemize}
\item Array $ver$ stores all vertices in $V(\G)$, which are sorted in ascending order of $m(v)$;
\item Array $pos$ records the position of each vertex $v$ in array $ver$, i.e., $ver[pos[v]] = v$;
\item Array $bin$ records the starting position of each bin, i.e., $bin[i]$ is the offset of the first vertex $v$ in $ver$ such that $m(v) = i$.
\end{itemize}
To build the arrays, we first scan all vertices in $V(\G)$ to determine the size of each bin (lines~4--5). Then, by accumulation from $0$, each element in $bin$ can be easily obtained (lines~6--10). Based on array $bin$, we set $ver[v]$ and $pos[v]$ for each vertex $v \in V(\G)$ (lines~11--14). Since the elements of $bin$ are changed at line~14, we recover $bin$ at lines~15--17.

The main loop (lines~18--31) works as follows: Each time we retrieve the first vertex $v$ remaining in array $ver$ (line~19). If $m(v) < d$, $v$ cannot exist in $C^{d}_{L}(\G)$, so we remove $v$ and its incident edges from $\G$ (line~21). For each vertex $u$ adjacent to $v$ on some layers, we must update $m(u)$ after removing $v$. Note that $m(v)$ can be decreased at most by $1$ since we remove at most one neighbor of $u$ from $\G$. If $m(u)$ is changed, arrays $ver$, $pos$ and $bin$ also need to be updated. Specifically, let $w$ be the first vertex in array $ver$ such that $m(w) = m(u)$ (line~25). We exchange the position of $w$ and $u$ in array $ver$ (line~27). Accordingly, $pos[w]$ and $pos[v]$ are updated (line~28). After that, we increase $bin[m(u)]$ by 1 (line~29) since $u$ is removed.

The main loop is repeated until $m(v) \geq d$ (line~31). Finally, the vertices remaining in $V(\G)$ are outputted as $C^{d}_{L}(\G)$ (line~32).

\begin{figure}[!t]
    \scriptsize
    \fbox{
    \parbox{\figwidth}{
    \textbf{Procedure} \textsf{dCC}$(\G, L, d)$
    \begin{algorithmic}[1]
        \STATE compute $m(v)$ for each vertex $v$ of $\G$
        \STATE $M \gets \max_{v \in V(\G)} m(v)$
        \STATE initialize arrays $bin$, $ver$ and $pos$
        \FOR{each vertex $v \in V$}
            \STATE $bin[m(v)] \gets bin[m(v)] + 1$
        \ENDFOR
        \STATE $start \gets 1$
        \FOR{$i \gets 0$ to $M$}
            \STATE $num \gets bin[i]$
            \STATE $bin[i] \gets start$
            \STATE $start \gets start + num$
        \ENDFOR
        \FOR{each vertex $v \in V$}
            \STATE $pos[v] \gets bin[m(v)]$
            \STATE $ver[pos[v]] \gets v$
            \STATE $bin[m(v)] \gets bin[m(v)] + 1$
        \ENDFOR
        \FOR{$i \gets M$ to $1$}
            \STATE $bin[i] \gets bin[i-1]$
        \ENDFOR
        \STATE $bin[0] = 1$
        \REPEAT
            \STATE $v \gets$ the first vertex remaining in array $ver$
            \IF{$m(v) < d$}
                \STATE remove $v$ and its incident edges from $\G$
                \FOR{each remaining vertex $u$ adjacent to $v$ on some layers}
                    \STATE compute $m(u)$
                    \IF{$m(u)$ is changed}
                        \STATE $w \gets bin[m(u)]$
                        \STATE $pw \gets pos[w], pu \gets pos[u]$
                        \STATE $ver[u] \gets w, ver[w] \gets u$
                        \STATE $pos[u] \gets pw, pos[w] \gets pu$
                        \STATE $bin[m(u)] \gets bin[m(u)] + 1$
                        \STATE $m(u) \gets m'(u)$
                    \ENDIF
                \ENDFOR
            \ENDIF
        \UNTIL{$m(v) \geq d$}
        \RETURN $V(\G)$
    \end{algorithmic}
    }}

    \caption{The \textsf{\small dCC} Procedure.}
    \label{Fig: dCCDetailed}
\vspace{-3em}
\end{figure}

\smallskip

\noindent{\underline{\bf Complexity Analysis.}}
Let $n = |V(\G)|$, $m_{i} = |E_i(\G)|$ and $m = |\bigcup_{i \in L} E_i(\G) |$, the time for computing $m(v)$ for all vertices $v \in V(\G)$ is $O(n|L|)$. The time for setting up arrays $ver$, $pos$ and $bin$ is $O(n)$. In the main loop, the time for updating $m(u)$ of a neighbor vertex $u$ is $O(|L|)$. Let $N_{\G}(u) = \bigcup_{i=1}^{l} N_{G_i}(u)$. Since vertex $u$ can be updated by at most $|N_{\G}(u)|$ times, the maximum number of updating is $O(\sum_{u \in V(\G)} |N_{\G}(u)|) = O(m)$. Consequently, the time complexity of \textsf{dCC} is $O(n|L| + n + m|L|) = O((n+ m)|L|)$. The space complexity of \textsf{dCC} is $O(n)$ since it only stores three arrays.

\section{The \textsf{Update}  Procedure}
\label{Sec: IUpdate}

We present the \textsf{Update} procedure in Fig.~\ref{Fig: UpdateDetailed}. The input of the procedure includes the set $\R$ of temporary top-$k$ diversified $d$-CCs, a newly generated $d$-CC $C$ and $k \in \mathbb{N}$. The procedure updates $\R$ with $C$ according to the rules specified in Section~\ref{Sec:BApproach-2}.

For each $d$-CC $C' \in \R$, we store both $C'$ and the size $|\Delta(\R, C')| $. To facilitate fast updating of $\R$, we build some auxiliary data structures. Specifically, we store $\R$ in two hash tables $M$ and $H$. For each entry in $M$, the key of the entry is a vertex $v$, and the value of the entry is $M[v] = \{ C' | C' \in \R, v \in C' \}$, that is, the set of $d$-CCs $C' \in \R$ containing vertex $v$. For each entry in $H$, the key of the entry is an integer $i$, and the value of the entry $H[i]$ is the set of $d$-CCs $C' \in R$ such that $|\Delta(\R, C')| = i$. Obviously, $C^{*}(\R)$ can be easily obtained from $H$ by retrieving the entry of $H$ indexed by the smallest key.

Given the temporary result set $\R$ and a new $d$-CC $C$, the procedure relies on three key operations to update $\R$, namely \textsf{Size($\R$, $C$)} that returns the size $|\Cov( (\R - \{C^{*}(\R) \} ) \cup \{ C \} )|$, \textsf{Delete($\R$)} that removes $C^{*}(\R)$ from $\R$, and \textsf{Insert($\R$, $C$)} that inserts $C$ to $\R$. We describe these procedures as follows.

\begin{figure}[!t]
    \scriptsize
    \fbox{
    \parbox{\figwidth}{
    \textbf{Procedure} \textsf{Update}$(\R, C, k)$
    \begin{algorithmic}[1]
        \IF{$|R| < k$}
            \STATE \textsf{Insert($\R$, $C$)}
        \ELSE
            \STATE $|\Cov(\R)| = \textsf{size}(M)$
             \IF{\textsf{Size($\R$, $C$)} $\geq (1 + 1/k ) |\Cov(\R)|$}
                \STATE \textsf{Delete($\R$)}
                \STATE \textsf{Insert($\R$, $C$)}
            \ENDIF
        \ENDIF
    \end{algorithmic}

    %\smallskip

    \textbf{Procedure} \textsf{Size}$(\R, C)$
    \begin{algorithmic}[1]
        \STATE obtain $C^{*}(\R)$ and $|\Delta(R, \{ C^{*}(\R) \}|$ from $H$
        \STATE $c \gets 0$
        \FOR{each vertex $v \in C$}
            \IF{$v$ is not a key in $M$}
                \STATE $c \gets c + 1$
            \ELSIF{$v \in C^{*}(\R)$ and $\textsf{size}(M[v]) = 1$}
                \STATE $c \gets c + 1$
            \ENDIF
        \ENDFOR
        \STATE $c \gets c + \textsf{size}(M) - |\Delta(\R, C^{*}(\R))|$
        \RETURN $c$
    \end{algorithmic}

   % \smallskip

    \textbf{Procedure} \textsf{Delete}$(\R)$
    \begin{algorithmic}[1]
        \STATE remove $C^{*}(\R)$ from $H$
        \FOR{each vertex $v \in C^{*}(\R)$}
            \STATE remove $C^{*}(\R)$ from $M[v]$
            \IF{$\textsf{size}(M[v])$}
                \STATE let $C'$ be the element in $M[v]$
                \STATE move $C'$ in $H$ from $H[|\Delta(\R, C')|]$ to $H[|\Delta(\R, C')| + 1]$
                \STATE increase $|\Delta(\R, C')|$ by $1$
            \ELSIF{$\textsf{size}(M[v]) = 0$}
                \STATE remove $v$ from $M$
            \ENDIF
        \ENDFOR
    \end{algorithmic}

   % \smallskip

    \textbf{Procedure} \textsf{Insert}$(\R, C)$
    \begin{algorithmic}[1]
        \STATE add $C$ into $\R$
        \STATE set $|\Delta(\R, C)|$ to $0$
        \FOR{each vertex $v \in C$}
            \IF{$v$ is not a key in $M$}
                \STATE add $v$ into $M$
                \STATE insert $C$ into $M[v]$
                \STATE increase $|\Delta(\R, C)|$ by $1$
            \ELSE
                \IF{$\textsf{size}(M[v]) = 1$}
                    \STATE let $C'$ be the element in $M[v]$
                    \STATE move $C'$ in $H$ from $H[|\Delta(\R, C')|]$ to $H[|\Delta(\R, C')| - 1]$
                    \STATE decrease $|\Delta(\R, C)|$ by $1$
                \ENDIF
                \STATE insert $C$ into $M[v]$
            \ENDIF
        \ENDFOR
        \STATE insert $C$ into $H$ based on $|\Delta(\R, C)|$
    \end{algorithmic}
    }}
   \caption{The \textsf{\small Update} Procedure.}
    \label{Fig: UpdateDetailed}
\vspace{-3em}
\end{figure}

\smallskip

\noindent{\underline{\bf Operation $\mathsf{Size}(\R, C)$.}}
Note that, $\Cov( (\R - \{C^{*}(\R) \} ) \cup \{ C \}$ can be decomposed into three disjoint subsets $\Cov(\R - \{ C^{*}(\R) \} )$, $C - \Cov(\R)$ and $C \cap \Delta(\R, C^{*}(\R))$. In the beginning, we can obtain $C^{*}(\R)$ and $|\Delta(\R, C^{*}(\R) )|$ from $H$ (line~1) and initialize the counter $c$ to $0$ (line~1). For each vertex $v \in C$, if $v$ is not a key in $M$, we have $v \in C - \Cov(\R)$, so we increase $c$ by $1$ (line~5). Otherwise, if $v \in C^{*}(\R)$ and $M[v]$ only contains $C^{*}(\R)$, $c$ is also increased by $1$ (line~7) since $v \in C \cap \Delta(\R, C^{*}(\R))$. Since $|\Cov(\R - \{ C^{*}(\R) \} )|$ is equal to $\textsf{size}(M) - |\Delta(\R, C^{*}(\R))|$, we accumulate $\textsf{size}(M) - |\Delta(\R, C^{*}(\R))|$ to $c$ (line~8) and return $c$ as the result (line~9).

\smallskip

\noindent{\underline{\bf Operation $\mathsf{Delete}(\R)$.}}
First, we retrieve $C^{*}(\R)$ from $H$ (line~1). For each vertex $v \in C^{*}(\R)$, $C^{*}(\R)$ is removed from $M[v]$ (line~3). Note that, if $M[v]$ contains a single element $C'$ after removing $C^{*}(\R)$, $v$ is a vertex only covered by $C'$. Therefore, we move $C'$ from $H[|\Delta(\R, C')|]$ to $H[|\Delta(\R, C')| + 1]$ (line~6) and increase $|\Delta(\R, C')|$ by $1$ (line~7). If $M[v]$ is empty, $v$ is no longer covered by $\R$, so $v$ is removed from $M$ (line~9).

\smallskip

\noindent{\underline{\bf Operation $\mathsf{Insert}(\R, C)$.}}
First, we insert $C$ to $\R$ (line~1) and set $|\Delta(\R, C)|$ to $0$ (line~2). For each vertex $v \in C$, if $v$ is not a key in $M$, we insert an entry with key $v$ and value $C$ to hash table $M$ (lines~5--6). At this moment, $v$ is only covered by $C$, so $|\Delta(\R, C)|$ is increased by $1$(line~7). If $v$ is a key in $M$, $C$ can be directly inserted to $M[v]$ (line~12). Note that, if $M[v]$ contains a single element $C'$ before insertion, $v$ will not be covered only by $C'$ after inserting $C$, so $C'$ is moved in $H$ from $H[|\Delta(\R, C')|]$ to $H[|\Delta(\R, C')| - 1]$ (line~11), and $|\Delta(\R, C')|$ is decreased by $1$ (line~12). After updating $M$, we obtain $|\Delta(\R, C)|$ and insert $C$ to $H$ accordingly (line~14).

By putting them altogether, we have the \textsf{Update} procedure. If $|\R| < k$, we directly insert $C$ to $\R$ (line~2). If $|\R| \ge k$, the \textsf{Size($\R$, $C$)} procedure is invoked to check if $C$ satisfies Rule~2 (line~5). If so, $\R$ is updated with $C$ by invoking \textsf{Delete($\R$)} (line~6) and \textsf{Insert($\R$, $C$)} (line~7).

\smallskip

\noindent{\underline{\bf Complexity Analysis.}}
The space cost for storing $\R$ and maintaining $M$ is $O(\sum_{C' \in \R} |C'|)$, and the space cost for storing $|\Delta(\R, C')|$ and maintaining $H$ is $O(k)$. Thus, the space complexity of \textsf{Update} is $O(2\sum_{C' \in \R} |C_j| + 2k) = O(\sum_{C' \in \R} |C'|)$.

Assume that an entry can be inserted to or deleted from a hash table in constant time. Thus, the time complexity of \textsf{Size($\R$, $C$)}, \textsf{Delete($\R$)} and \textsf{Insert($\R$, $C$)} is $O(|C|)$, $O(|C^{*}(\R)|)$ and $O(|C|)$, respectively. Consequently, the time complexity of \textsf{Update} is $O(\max\{|C|, |C^{*}(\R)|\})$.

\section{The \textsf{InitTopK}  Procedure}
\label{Sec: InitTopK}

\begin{figure}[!t]
    \scriptsize
    \fbox{
    \parbox{\figwidth}{
    \textbf{Procedure} \textsf{InitTopK}$(\G, d, s, k, \R)$
    \begin{algorithmic}[1]
    \STATE $\R \gets \emptyset$
    \FOR{$p \gets 1$ to $k$}
        \STATE $i \gets \arg\max_{i \in [l(\G)]} |\Cov(\R \cup \{  C^{d}(G_i) \} )| - | \Cov(\R)|$
        \STATE $L \gets \{ i \} $
        \STATE $C \gets C^{d}(G_i)$
        \FOR{$q \gets 1$ to $s-1$}
            \STATE $j \gets \arg\max_{j \in [l(\G)] - L} |C \cap C^{d} (G_j)|$
            \STATE $L \gets L \cup \{ j \} $
            \STATE $C \gets C \cap C^{d} (G_j)$
        \ENDFOR
        \STATE $C' \gets \mathsf{dCC}(\G[C], L, d)$
        \STATE $\mathsf{Update}(\R, C')$
    \ENDFOR
    \RETURN $\R$
    \end{algorithmic}
    }}

    \caption{The \textsf{InitTopK} Procedure.}
    \label{Fig:RInit}
\vspace{-2em}
\end{figure}

We present the \textsf{InitTopK} procedure in Fig.~\ref{Fig:RInit}. The input of the procedure includes the multi-layer graph $\G$, $d, s, k \in \mathbb{N}$ and set $\R$ of temporary top-$k$ diversified $d$-CCs. The \textsf{InitTopK} procedure in Section~IV.C initializes $\R$ so that $|\R| = k$.

At first, we set $\R$ as an empty set (line~1). The \textbf{for} loop (lines~2--11) executes $k$ times. In each loop, a candidate $d$-CC is added to $\R$ in the following way: First, we select layer $i$ such that the $d$-core $C^d(G_i)$ can maximumly enlarges $\Cov(\R)$ (line~3). Let $C = C^{d}(G_i)$ and $L = \{i\}$ (line~4--5). Then, we add $s - 1$ other layer numbers to $L$ in a greedy manner. In each time, we choose layer $j \in [l(\G)] - L$ that maximizes $|C \cap C^{d}(G_j)|$, update $L$ to $L \cup \{j\}$ and update $C$ to $C \cap C^{d}(G_j)$ (lines~7--9). When $|L| = s$, we compute the $d$-CC $C^d_L(\G)$ and update $\R$ with $C^d_L(\G)$ (lines~11--12).

\end{document}